\documentclass[10pt]{article}
\usepackage{amssymb,amsfonts,amsmath, amsthm, amsbsy}
\usepackage{caption,color,graphicx,paralist, subcaption, placeins, array, mathtools, url, mdwlist, color} 
\usepackage[text={6.75in,9.5in},centering,letterpaper]{geometry}
\usepackage[hypertexnames=false,colorlinks=true,urlcolor=blue,linkcolor=blue,citecolor=blue]{hyperref}

\everymath={\displaystyle}
 \numberwithin{equation}{section}

\newtheorem{Theorem}{Theorem}[section]

\newtheorem{Lemma}[Theorem]{Lemma}
\newtheorem{Proposition}[Theorem]{Proposition}

\newtheorem{Remark}[Theorem]{Remark}

\newcommand{\R}{\mathbb{R}}

\def\Re{\mathop{\mathrm{Re}}}

\newcommand{\rmO}{\mathrm{O}}

\newcommand{\rmd}{\mathrm{d}}
\newcommand{\rme}{\mathrm{e}}
\newcommand{\rmi}{\mathrm{i}}

\newcommand{\sech}{\,\mathrm{sech}\,}
\newcommand{\eps}{\varepsilon}

\begin{document}
\begin{center}
{\fontsize{15}{15}\fontfamily{cmr}\fontseries{b}\selectfont    
Wave train selection by invasion fronts in the FitzHugh--Nagumo equation
}
\\[0.2in]
Paul Carter$\,^1$ and Arnd Scheel$\,^2$\\
\textit{\footnotesize $\,^1$ University of Arizona, Department of Mathematics, 617 N Santa Rita Ave, Tucson, AZ 85721, USA}\\
\textit{\footnotesize $\,^2$ University of Minnesota, School of Mathematics,   206 Church St. S.E., Minneapolis, MN 55455, USA}\\

\date{\small \today} 
\end{center}

\begin{abstract}
\noindent We study invasion fronts in the FitzHugh--Nagumo equation in the oscillatory regime using singular perturbation techniques. Phenomenologically, localized perturbations of the unstable steady-state grow and spread, creating temporal oscillations whose phase is modulated spatially. The phase modulation appears to be selected by an invasion front that describes the behavior in the leading edge of the spreading process. We construct these invasion fronts for large regions of parameter space using singular perturbation techniques. Key ingredients are the construction of periodic orbits, their unstable manifolds, and the analysis of pushed and pulled fronts in the fast system. Our results predict the wavenumbers and frequencies of oscillations in the wake of the front through a phase locking mechanism. We also identify a parameter regime where nonlinear phase locked fronts are inaccessible in the singularly perturbed geometry of the traveling-wave equation. Direct simulations confirm our predictions and point to  interesting phase slip dynamics.
\end{abstract}

\section{Introduction}\label{s:1}

Front propagation into unstable states has been the object of interest in both mathematics and the applied sciences, starting with the work by Fisher and Kolmogrov--Petrovsky--Piscounov in genetics and population dynamics \cite{fisher,kpp}.
While there continues to be a tremendous amount of interest related to front propagation and invasion in the context of ecology, front propagation has received a considerable amount of attention in the physics community after connections with pattern formation \cite{deelanger} 
and plasma or fluid instabilities were successfully established. We refer to \cite{vS} 
for an extensive recent review. 

Our interest here is, much in the spirit of \cite{deelanger}, 
the role of invasion in the selection of patterns. Here, we think of patterns as sustained spatio-temporally non-constant states. The effect of the invasion process is primarily the selection of spatio-temporally periodic solutions with a wavenumber and frequency determined by parameters, only, rather than initial conditions in a spatially extended system $x\in\R$, say. This selection is in sharp contrast to quenching experiments, where random fluctuations around an unstable state typically lead to spatio-temporally disorganized patterns with a distribution of frequencies and wavenumbers, and possibly with defects playing a major role in the long-term dynamics; see \cite{kotzagiannidis,scheelstevens} for cases where the difference between patterns selected by invasion versus patterns selected from white noise is rather dramatic.

A typical difficulty associated with such pattern-formation contexts is the absence of comparison principles. In addition, pattern-forming fronts tend to be periodic in time, hence solutions to a boundary-value problem for parabolic problems. In those contexts, there are few examples where invasion fronts can be actually constructed; for some of the few examles see \cite{sh,hara} for fronts near Turing-type instabilities  and \cite{ac,ch} for topological constructions in phase separation problems. 

On the other hand, the selection problem can often be answered rather satisfactorily based on linearized information, only; see \cite{hs}.
Nonlinear fronts whose speeds and frequencies are governed by the linear growth in this fashion are often referred to as pulled fronts, whereas fronts that propagate faster than this predicted linear speed are referred to as pushed fronts. We caution however that this simple classification misses subtleties in the linear growth \cite{hslotka,hs}
and possible other resonant interactions \cite{holzerfayescheel}.

In this context, the present work contributes existence results for pulled and pushed invasion fronts that leave behind a spatially propagating wave train. The invasion fronts are stationary in a comoving frame. In particular, wave trains propagate with the same speed as the leading edge of the invasion front. This can be attributed to some extent to linear predictions, which, despite the oscillatory nature of the nonlinear dynamics, predict stationary behavior due to real roots of the dispersion relation. The analysis therefore reduces to the analysis of an ordinary differential equation for traveling-wave solutions. In this ordinary differential equation, we exploit a fast-slow decomposition to construct heteroclinic orbits. 

In the remainder of the introduction, we set up our problem and describe our main result. 

\paragraph{The FitzHugh--Nagumo equation.}
We consider the FitzHugh--Nagumo equation on the real line,
\begin{align}\label{eq_pde}
\begin{split}
u_t &= u_{xx}+f(u)-w\\
w_t &= \eps(u-\gamma w-a),
\end{split}
\end{align}
with cubic nonlinearity $f(u)=u(1-u)(u-a)$, $0<a<1/2$, and $\gamma$ sufficiently small such that $u=a,w=0$ is the unique equilibrium. Throughout, we are interested in $0\leq \eps\ll 1$. Note that the parameter $a$ appears in both the $u$- and the $w$-equation, simply making sure that the equilibrium $u=a$, $w=0$ is explicitly given and controlled by the parameter $a$. The FitzHugh--Nagumo equation is often thought of as an oversimplified model for signal propagation along a nerve axon, but has, in small variations, been used to model phenomena as disparate as transitions to  turbulence in fluids \cite{barkley}, 
CO-oxidation on platinum surfaces \cite{baer,mikhailov},
or heart arrhythmias \cite{Luther2011}. 
One often considers the excitable regime, $a\lesssim 0$ in our parameterization, when the unique equilibrium $u=a$, $w=0$ is stable, but finite-size, yet small perturbations lead to large excursions in phase space and long transients before recovery. Spatial coupling can then sustain activity through excitation pulses or, in two-dimensional media, spiral waves. 

We are here concerned with the oscillatory regime, $0<a<1/2$,  that is, dynamics of spatially constant solutions converge to periodic orbits. In fact, one can readily see that the unique equilibrium is completely unstable, and one can construct invariant regions that guarantee boundedness of solutions for positive times, thus ensuring existence of a limit cycle by Poincar\'e-Bendixson. The limit cycle can be described in more detail when $\eps\ll 1$. Indeed, the system is roughly equivalent to the van-der-Pol oscillator, with well understood relaxation oscillations~\cite{krupaszmolyan20012}.

\paragraph{Traveling waves.} 
In this general context, spatio-temporal dynamics of the FitzHugh--Nagumo equations have been largely understood as being organized around traveling pulses and pulse trains, so-called trigger waves, in the excitable regime, and around spatially homoegeneous oscillations, and their associated long-wavelength periodic modulations in the oscillatory regime. All of those can be found as solutions to an ordinary differential equation. Traveling waves $(u,w)(x,t)=(u,w)(x-ct)$ which travel to the right with positive wave speed $c>0$ solve the first-order system
\begin{align}\label{eq_twode}
\begin{split}
\dot{u} &= v\\
\dot{v} &= -cv-f(u)+w\\
\dot{w} &= -\frac{\eps}{c}(u-\gamma w-a),
\end{split}
\end{align}
where $~\dot{}=\frac{d}{d\xi}$ denotes differentiation with respect to the traveling wave variable $\xi = x-ct$. For $0<\eps/c\ll 1$, this three dimensional ODE is a slow-fast system, with two fast variables, $u$ and $v$,  and one slow variable, $w$. We refer to~\eqref{eq_twode} as the \emph{fast} system. 

The slow-fast structure in the traveling-wave equation has been exploited extensively, in particular in the construction of pulses and pulse trains~\cite{has, joe,cas}.
Note that this system has a unique equilibrium  $(u,v,w)=(a,0,0)$ which we denote by $p\in\R^3$. Of course, $p$ corresponds to the spatially constant equilibrium in \eqref{eq_pde}.

\paragraph{Spreading speeds and steepest fronts.} It turns out that traveling waves connecting the unstable equilibrium to periodic orbits come in families parameterized by the wave speed. Rather than constructing this entire family, we construct fronts that, among all traveling waves (and, in fact, among all traveling waves allowing for a time-periodic modulation in a comoving frame) have the steepest possible decay. Supported by a wealth of examples and by numerical evidence, these steepest fronts are \emph{selected} by compactly (or sufficiently rapidly decaying) initial data \cite{vS}. We encounter, matching the general classification in \cite{vS}, two types of steepest fronts, 
\begin{itemize}
 \item \emph{pulled fronts} possess steepest decay since the real part of leading eigenvalues in the relevant stable manifold achieves a minimum as a function of wave speed for the selected pulled speed;
 \item \emph{pushed fronts} possess steepest decay since the heteroclinic connects to a strong stable manifold for the selected pushed speed. 
\end{itemize}
The speed of pulled fronts is, by definition, determined by the linearization at the unstable state, only. The speed of pushed fronts depends on the shape of the nonlinearity, as typical for codimension-one heteroclinic orbits. We give more background on these concepts and a brief, cautious guide to the literature in~\S\ref{s:spsp}, but we will state our main results using this terminology below. 

\paragraph{Periodic orbits.} For each $0<a<1/2$ and each sufficiently small $\eps>0$, the traveling-wave equation \eqref{eq_twode} possess a family of periodic orbits parameterized by the speed $c>0$, which effectually represent wave train solutions to the FitzHugh--Nagumo equation. The family of periodic orbits naturally breaks into two qualitatively different pieces, the ``trigger'' and the ``phase'' piece. 
The trigger family had been constructed in \cite{STR}. We outline that construction and add the construction of the phase family in~\S\ref{sec_periodicorbits}. More precisely, we construct periodic orbits for any given $c>0$, with $\eps$ sufficiently small, which are
\begin{itemize}
 \item \emph{trigger wave trains} $(u,v,w)_\mathrm{h}(\xi;c,a,\eps)$ for $0<c<c_*(a)$, and
 \item \emph{phase wave trains} $(u,v,w)_\mathrm{nh}(\xi;c,a,\eps)$ for $c_*(a)<c$,
\end{itemize}
where 
\begin{equation}\label{e:csa}
 c_*(a)=\sqrt{\frac{1-a+a^2}{2}}.
\end{equation}
Phase wave-trains are distinguished by the fact that their traveling-wave trajectory passes through the maximum or minimum of the cubic. As a consequence, the singular perturbation analysis is somewhat more subtle, involving a non-hyperbolic piece of the slow manifold. More phenomenologically, these waves possess very long wavelength and roughly resemble the relaxation oscillations in the pure kinetics. Their amplitude is roughly independent of the period, that is they exhibit fully developed oscillations, as opposed to the somewhat more narrowly spaced `trigger waves. Trigger waves are constructed out of heteroclinic orbits in the $u$-equation, propagating roughly with the speed of propagation of these interfaces from the equation with frozen value of $w$. We refer to  \cite{kt} for a discussion of these two types of excitation waves. 

We now state our main result on periodic orbits.
\begin{Theorem}[Periodic Orbits]\label{thm_periodicexistence}
Fix $0<a<1/2$ and $c>0$. Then for each sufficiently small $0<\eps\ll1$, the system~\eqref{eq_twode} admits a periodic orbit $\Gamma_\eps(c)$ with wave speed $c$ and period $L(c;\eps)$. Moreover, fix either $0<c<c_*(a)$ or $c_*(a)<c$; then for all sufficiently small $\eps$,
\[
 L(c;\eps)=\left\{ \begin{array}{ll} 
                    {L_0(c)}{\eps^{-1}}+\rmO(\log \eps),& 0<c<c_*(a),\\
                    {L_0(c)}{\eps^{-1}}+\mathcal{O}(\eps^{-2/3}),& c_*(a)<c.\\
                   \end{array}\right.
\]
for some function $L_0(c)>0$. The error terms are understood in terms of $\eps$ for each fixed $c$. The functions $L(c;\eps)$ and $L_0(c)$ are monotonically increasing in $c$. 
\end{Theorem}
As mentioned above, the new contribution here is the case $c_*(a)<c$, which is discussed in  Proposition~\ref{prop_nh_periodicorbits};  see~\S\ref{sec_periodicorbits} for the geometry of the periodic orbits in the context of the traveling wave equation~\eqref{eq_twode}.

\paragraph{Invasion fronts --- main results.} We focus throughout on $a<1/2$; the case $a>1/2$ is obtained by reflection $u\mapsto 1-u$. Our main results will invoke several speeds in addition to the critical speed $c_*(a)$ where wave trains reach maximal amplitude, defined in \eqref{e:csa}, that we shall define now:
\begin{itemize}
 \item $c_\mathrm{lin}=2\sqrt{a(1-a)}$, the linear or pulled speed;
 \item $c_\mathrm{p}=\frac{1}{\sqrt{2}}(1+a)$, the pushed speed, defined for $a<1/3$;
 \item $c_\mathrm{bs}(a)=\frac{1}{\sqrt{2}}(1-2a)$, the bistable speed for fronts between $u=0$ and $u=1$ for $w=\eps=0$.
\end{itemize}
All those speeds are understood in the limit $\eps=0$ for fixed $a$. Note that $c_\mathrm{lin}>c_\mathrm{bs}$ for $a>a_\mathrm{b}=(3-\sqrt{6})/6$. Our main results will include a construction of extensions of pulled and pushed speeds to finite $\eps>0$; we denote those speeds by the same name slightly abusing notation.

Our main results concern the existence of both pulled and pushed fronts in~\eqref{eq_pde} for sufficiently small $\eps>0$, which are summarized in the following two theorems.
\begin{Theorem}[Pulled fronts]\label{thm_pulledexistence}
Fix $a_\mathrm{b}<a<1/2$. There exists $\eps_0>0$ such that for all $0<\eps<\eps_0$, there exists a wave speed $c_\mathrm{lin}(a,\eps) =c_\mathrm{lin}+\mathcal{O}(\eps)$ such that~\eqref{eq_twode} admits heteroclinic orbits $F^\ell_\eps(a), F^r_\eps(a)$ which are backward asymptotic to the periodic orbit $\Gamma_\eps(c_\mathrm{p}(a,\eps))$ and forward asymptotic to the equilibrium $(a,0,0)$. The selected periodic orbit is a trigger wave train when $a_\mathrm{b}<a<(3-\sqrt{5})/{6}$, and a phase wave train when $(3-{\sqrt{5}})/{6}<a<{1}/{2}$. Moreover, all heteroclinic orbits in a vicinity of $F^\ell$ or $F^r$, for speeds close to $c_\mathrm{lin}(a,\eps)$, possess weaker exponential decay towards $(a,0,0)$. 
\end{Theorem}
\begin{Remark}
The $\ell,r$--superscripts for the fronts $F^\ell_\eps(a), F^r_\eps(a)$ given by Theorem~\ref{thm_pulledexistence} correspond to jumps in the fast subsystem of~\eqref{eq_twode} which originate from the left/right branches of the critical manifold; see~\S\ref{sec_slowfast} for the relevant geometry. In particular, we  have two distinct pulled front solutions: the $u$-profile of the front $F^\ell_\eps(a)$ eventually increases monotonically as $\xi \to \infty$, while that of the front $F^r_\eps(a)$ eventually decreases monotonically as $\xi \to \infty$ for $a>1/3$. We expect the front $F^r_\eps(a)$ to be unstable for $a<1/3$. 
\end{Remark}
\begin{Theorem}[Pushed fronts]\label{thm_pushedexistence}
Fix $0<a<1/3$. There exists $\eps_0>0$ such that for all $0<\eps<\eps_0$, there exists a wave speed $c_\mathrm{p}(a,\eps) =c_\mathrm{p}+\mathcal{O}(\eps)$ such that~\eqref{eq_twode} admits a heteroclinic orbit $P^r_\eps(a)$ which is backward asymptotic to the periodic orbit $\Gamma_\eps(c_\mathrm{p}(a,\eps))$ and decays exponentially as $\xi\to \infty$ to $(a,0,0)$ at the strongest possible rate.  In particular, all heteroclinic orbits in a vicinity of $P^r$, for speeds close to $c_\mathrm{lin}(a,\eps)$, possess weaker exponential decay towards $(a,0,0)$. 
\end{Theorem}

\begin{Remark}\label{r:cg}
 For all fronts constructed here, wave trains are generated by the front in the following sense. Given a wave train, one can calculate a group velocity, that measures the speed of propagation of distrubances by the linearized equation at the wave train; see \cite{ssdefect} for background. One then says that a wave train is ``generated'' at an front interface if the group velocity of the wave train, in a frame moving with the speed of the interface, points away from the interface. A short calculation shows that the group velocity is given through 
 \[
  c_\mathrm{g}=\frac{\rmd \omega}{\rmd k}=c-L\frac{\rmd c}{\rmd L},
 \]
 where $\omega=ck$ and $L=2\pi/k$. Since wave trains and front interface propagate with the same (phase) velocity, $\omega/k=c$, the group velocity in the comoving frame is simply 
 \[
  c_\mathrm{g}-c=-L\frac{\rmd c}{\rmd L},
 \]
 and the sign equals the sign of $-\frac{\rmd L}{\rmd c}$, which, according to Theorem \ref{thm_periodicexistence} is always negative as claimed. 
 
 One can be slightly more precise in the case of phase waves, where $L(c)\sim c\bar{L}/\eps$ at leading order in $\eps$, such that $-L (\rmd c/\rmd L)=-c$, that is, at leading order in $\eps$, the group velocity of phase wave trains vanishes in a steady frame. 
\end{Remark}

\paragraph{Outline.} The remainder of this paper is organized as follows. We discuss linear spreading speeds and front selection criteria in~\S\ref{s:spsp}. In~\S\ref{sec_slowfast}, we consider the singular limit $\eps \to 0$ of the traveling wave ODE~\eqref{eq_twode} in the context of geometric singular perturbation theory, and we construct singular periodic orbits and pushed/pulled front solutions. The persistence of these solutions for $0<\eps \ll 1$ and the proofs of Theorems~\ref{thm_periodicexistence},~\ref{thm_pulledexistence} and~\ref{thm_pushedexistence} are given in~\S\ref{sec_persistence}. We present numerical simulations in~\S\ref{sec_numerics} to visualize the above results, and we conclude with a discussion in~\S\ref{sec_discussion}.

\paragraph{Acknowledgments.} AS gratefully acknowledges support through NSF grant DMS--1311740. 

\section{Spreading speeds and front selection}
\label{s:spsp}
We present a very brief review of speed selection criteria for fronts and motivate the connection with the selection of steepest fronts in invasion processes. 

\paragraph{Linear spreading speeds.} Linearizing \eqref{eq_pde} at the constant state $u=a,w=0$, we find after Fourier-Laplace transform $\rme^{\nu (x-ct) + \lambda t}$, $\nu=\rmi k\in\rmi\R$, 
\begin{equation}\label{e:dm}
 (\lambda-c\nu)\left(\begin{array}{c}
  u\\v
 \end{array}\right) = \left(\begin{array}{cc} \nu^2 +f'(a)& -1\\ \eps & -\gamma \eps\end{array}\right) \left(\begin{array}{c}
  u\\v
 \end{array}\right) ,
\end{equation}
which, taking determinants and writing $\alpha=f'(a)=a(1-a)>0$, is equivalent to the \emph{dispersion relation},
\begin{equation}\label{e:d}
d(\lambda,\nu;\eps):=(\lambda-c\nu -\nu^2-\alpha)(\lambda-c\nu+\gamma\eps)+\eps=0.
\end{equation}
Evaluating at $\lambda=0$ gives the cubic equation for eigenvalues at the linearization of $p$. Following \cite{hs}, 
we are interested in double roots of $d$, 
\begin{equation}\label{e:dd}
 d(\lambda,\nu;\eps)=0,\qquad \partial_\nu d(\lambda,\nu;\eps)=0. 
\end{equation}
This pair of complex equations possesses a finite number of solutions. We are interested in particular solutions, \emph{pinched double roots}, which can be defined as follows. For each double root $(\lambda_*,\nu_*)$, there are (at least) two distinct roots $\nu_\pm(\lambda)$ such that $\nu_\pm(\lambda_*)=\nu_*$. Since $d$ is analytic, we can follow these two roots along a path where $\Re\lambda\nearrow +\infty$. The pinching condition then assumes that $\pm\Re\nu_\pm(\lambda)>0$ for $\Re\lambda\gg 1$; see \cite{hs} for details.

It turns out \cite[Lemma 4.4]{hs} that, generically and in our present situation,  spatially localized initial conditions to the linearized equation grow pointwise if and only if there exists a pinched double root with $\Re\lambda_*>0$. In order to understand spreading speeds, one therefore investigates pinched double roots, and thereby pointwise stability, depending on the wave speed $c$. 
The linear spreading speed is then defined as the largest speed $c$ such that there exists a pinched double root with $\Re\lambda_*\geq 0$. The resulting algebraic equations are usually difficult to analyze analytically. In the case $\eps=0$, one can however readily compute spreading speeds. For this, notice that roots of $d$ are explicitly given through
\[
 \nu_\pm(\lambda)=-\frac{c}{2}\pm\sqrt{\frac{c^2}{4}-\alpha+\lambda}, \qquad \nu_0(\lambda)=\lambda/c.
\]
Pinched double roots can therefore occur when either $\nu_+(\lambda)=\nu_-(\lambda)$, or when $\nu_-(\lambda)=\nu_0(\lambda)$. The latter is excluded for $\lambda>0$, such that one finds the following result. 
\begin{Lemma}[Linear spreading speed]
 The linear spreading speed associated with the dispersion relation \eqref{e:d} is given through 
 \[
  c_\mathrm{lin}(\eps;a)=2\sqrt{\alpha}+\rmO(\eps),
 \]
  for $\eps$ sufficiently small. Moreover, the associated roots $\lambda_*=0$, $\nu_*<0$ are real and $\Re\nu_+(0)$ is minimal for $c=c_\mathrm{lin}(\eps;a)$. 
\end{Lemma}
\begin{proof}
 The existence of $c_\mathrm{lin}$ follows from the implicit function theorem, continuing solutions to \eqref{e:dd} from $\eps=0$. The real part of $\nu_+$ is minimal at $c=c_\mathrm{lin}$ since we have a double root precisely at this point. A local expansion shows that the double root splits into a pair of real roots for increasing $c$, such that the real part of $\Re\nu_+$ increases with $\sqrt{c-c_\mathrm{lin}}$ in this case. For $c<c_\mathrm{lin}$, the two roots become complex and an expansion shows that the real part has asymptotics $-(1/2+\rmO(\eps))c$, and is hence increasing when $c$ decreases. 
\end{proof}

The fact that $\lambda_*$ associated with the pinched double root is real implies that, for the linear equation, the spatio-temporal growth is stationary  in the moving frame at the leading edge, and one can consequently hope for a description of the growth process in terms of stationary solutions in a comoving frame. In order to to characterize fronts with speeds not determined by the linearization, we next turn to a heuristic selection criterion, the selection of the steepest front.

\paragraph{Steepest fronts.} Suppose that we are interested in traveling fronts, that is, in equilibria in a comoving frame $\xi=x-ct$, where $c$ is not necessarily the linear spreading speed, which connect a stable state in the wake, $\xi=-\infty$, to $p$. It turns out that the unstable manifold of stable states is two-dimensional, a fact that we shall establish below for the periodic orbits of interest here. The general argument is more widely applicable; see for instance \cite{ebert,mesuro}. 
Motivated by the fact that we start from compactly supported initial data, we look for traveling-wave solutions with the steepest possible decay. Note that, for $\lambda=0$, the eigenvalues have $\Re\nu_\pm<0$, $\Re\nu_0=\rmO(\eps)\lesssim 0$ for $\gamma\alpha<1$ (when $\gamma\alpha>1$, the system possesses three equilibria). Counting dimensions, we expect robust intersections between the two-dimensional unstable manifold of the periodic  orbit and the two-dimensional strong stable manifold associated with $\nu_\pm$. Decay in this two-dimensional strong stable manifold is steepest when $\Re\nu_+$ is minimal, which occurs precisely for $c=c_\mathrm{lin}$. On the other hand, we may be able to find intersections for \emph{specific} values of $c$ with steeper decay, since $\Re\nu_-(c)<\Re\nu_-(c_\mathrm{lin})$. We refer to such intersections with the one-dimensional super-strong stable manifold as \emph{pushed fronts}.

\section{Slow-fast analysis}\label{sec_slowfast}
In this section, we outline the singular limit geometry of~\eqref{eq_twode} in the context of geometric singular perturbation theory. We begin with a description of the slow reduced system in~\S\ref{sec_reduced}, followed by the fast layer subsystem in~\S\ref{sec_layer}. We then construct singular $\eps=0$ periodic orbits in~\S\ref{sec_periodicorbits} and singular pulled/pushed front solutions in~\S\ref{sec_singularsolns}.

\subsection{Slow subsystem}\label{sec_reduced}
Rescaling the traveling wave variable in~\eqref{eq_twode} by $\tau=\eps \xi$, we obtain the slow system
\begin{align}
\begin{split}\label{eq_slow}
\eps u' &=v\\
\eps v'&= -cv-f(u)+w\\
w'&=-\frac{1}{c}(u-\gamma w-a),
\end{split}
\end{align}
where ``$~{}^\prime~$" denotes $\frac{d}{d\tau}$. Setting $\eps=0$, the flow is restricted to the critical manifold
\begin{align}
\mathcal{M}_0=\{(u,v,w):v=0, w=f(u)\},
\end{align}
 with dynamics given by the reduced equation
\begin{align}
\begin{split}\label{eq_reduced}
u'&=-\frac{u-\gamma f(u)-a}{cf'(u)}.
\end{split}
\end{align}

The critical manifold is composed of three branches $\mathcal{M}_0=\mathcal{M}^\ell_0\cup\mathcal{M}^m_0\cup\mathcal{M}^r_0$, where
\begin{align}
\mathcal{M}_0^\ell&=\{(u,v,w):v=0, w=f(u), u\in (-\infty, u_\ell)\}\\
\mathcal{M}_0^m&=\{(u,v,w):v=0, w=f(u), u\in [u_\ell, u_r]\}\\
\mathcal{M}_0^r&=\{(u,v,w):v=0, w=f(u), u\in (u_r,\infty)\}
\end{align}
and
\begin{align}\label{eq_uell}
u_\ell&=\frac{1}{3}\left(a+1-\sqrt{1-a+a^2}\right)\\
u_r&=\frac{1}{3}\left(a+1+\sqrt{1-a+a^2}\right).
\end{align}
Setting $w_\ell=f(u_\ell)$ and $w_r = f(u_r)$, the points $p_\ell = (u_\ell, 0, w_\ell)$ and $p_r=(u_r,0,w_r)$ denote the locations of the lower left and upper right fold points, respectively, on the critical manifold. The critical manifold and the associated reduced flow~\eqref{eq_reduced} are shown in Figure~\ref{fig:singular_slow}.

\begin{figure}
\centering
\includegraphics[width=0.6\linewidth]{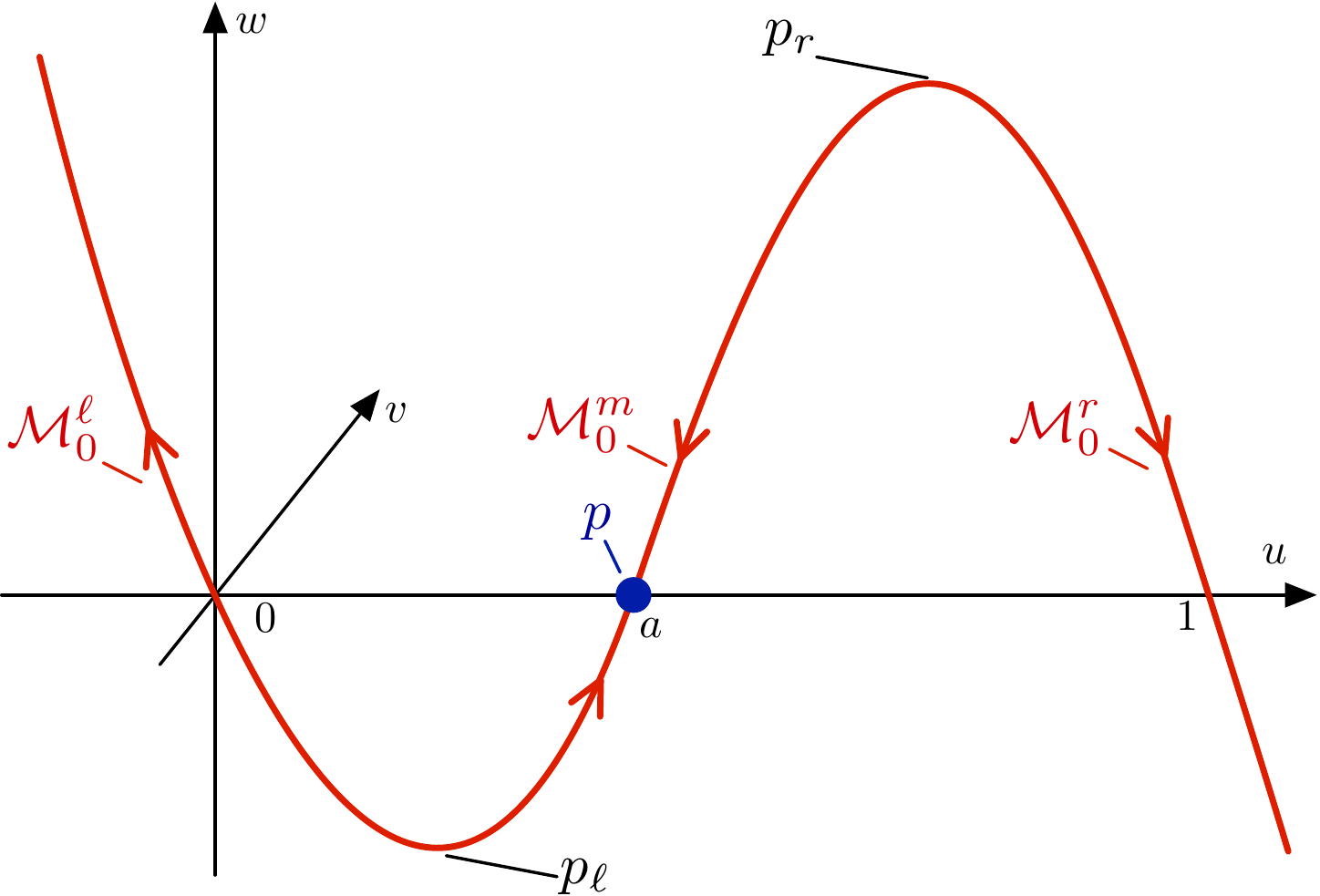}
\caption{Shown is the critical manifold $\mathcal{M}_0=\mathcal{M}^\ell_0\cup\mathcal{M}^m_0\cup\mathcal{M}^r_0$ and the associated reduced flow~\eqref{eq_reduced}.}
\label{fig:singular_slow}
\end{figure}

Away from the folds, the reduced flow on $\mathcal{M}^\ell_0$ satisfies $u'<0$, while the reduced flow on $\mathcal{M}^r_0$ satisfies $u'>0$. On the middle branch $\mathcal{M}^m_\eps$ there is a single equilibrium at $p=(a,0,0)$ which is attracting for the reduced flow.

\subsection{Layer analysis of fronts}\label{sec_layer}

Setting $\eps=0$ in the traveling wave equation~\eqref{eq_twode}, we obtain the layer (fast) subsystem

\begin{align}
\begin{split}\label{eq_layer}
\dot{u} &=v\\
\dot{v}&= -cv-f(u)+w
\end{split}
\end{align}

For each $0<a<1/2$ and each $w\in(w_\ell, w_{r})$, this system has three equilibria $p_i(w), i=1,2,3$, where $p_i(w) = (u_i(w),0)$, and the roots $u_i(w)$ of $w=f(u)$ are numbered in increasing order. By examining the linearization
\begin{align}
\begin{split}
\begin{pmatrix}\dot{U}\\\dot{V}\end{pmatrix}&=\begin{pmatrix} 0 & 1\\ -f'(u) &-c\end{pmatrix}\begin{pmatrix}U\\V\end{pmatrix}
\end{split}
\end{align}
which has eigenvalues
\begin{align}
\begin{split}
\nu_\pm &= \frac{-c\pm\sqrt{c^2-4f'(u)}}{2} 
\end{split}
\end{align}
we see that for $c>0$ the outer equilibria $p_1(w),p_3(w)$ are saddles, and the middle equilibrium $p_2(w)$ is a stable node or focus, depending on the sign of the quantity $c^2-4f'(u_2(w))$. 

We are primarily interested in the layer problem for $w=0$, which contains the equilibrium $(u,v,w)=(a,0,0)$ of the full system (note that $p=p_2(0)$). Linearizing about the equilibrium $p_2(0)=(u_2(0),0)=(a,0)$ in the layer problem~\eqref{eq_layer}, we obtain the eigenvalues
\begin{align}
\begin{split}
\nu_\pm &= \frac{-c\pm\sqrt{c^2+4(a^2-a)}}{2} 
\end{split}
\end{align}
We are interested in the fronts which connect the middle equilibrium $p_2(0)$ to either $p_1(0)$ or $p_3(0)$; the results are summarized in Figures~\ref{fig:layer_ppps} and~\ref{fig:speeds}. For each value of the wavespeed $c>0$ the equilibrium $p_2(0)$ is completely stable. There are two fronts, or heteroclinics: $\phi^\ell$ which connects $p_1(0)$ to $p_2(0)$ for all $c>0$, and $\phi^r$ which connects $p_3(0)$ to $p_2(0)$ for all $c>c_\mathrm{bs}=\frac{1}{\sqrt{2}}(1-2a)$.

\begin{figure}
\hspace{.025\textwidth}
\begin{subfigure}{.3 \textwidth}
\centering
\includegraphics[width=1\linewidth]{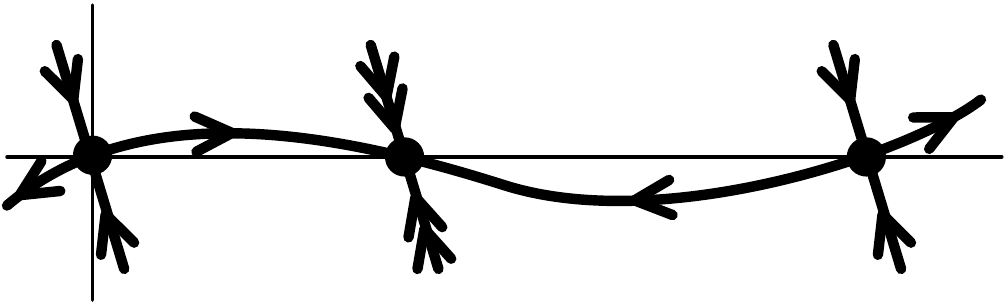}
\caption{$c>c_\mathrm{p}$}
\end{subfigure}
\hspace{.025\textwidth}
\begin{subfigure}{.3 \textwidth}
\centering
\includegraphics[width=1\linewidth]{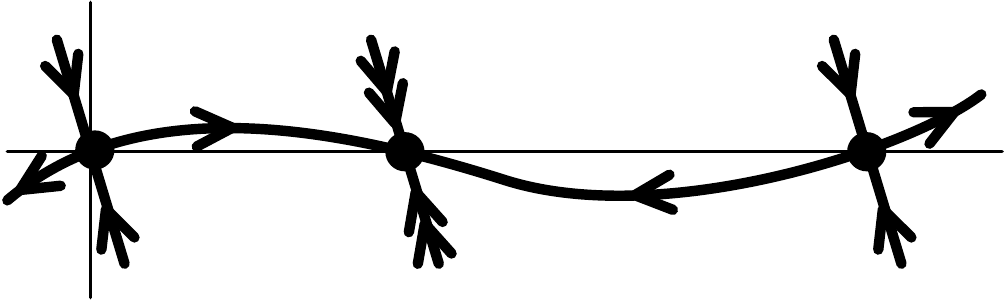}
\caption{$c>c_\mathrm{p}$}
\end{subfigure}
\hspace{.025\textwidth}
\begin{subfigure}{.3 \textwidth}
\centering
\includegraphics[width=1\linewidth]{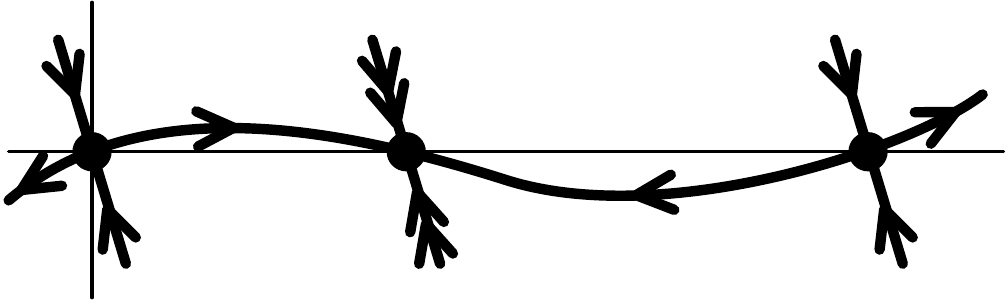}
\caption{$c>c_\mathrm{lin}$}
\end{subfigure}
\hspace{.025\textwidth}\\

\hspace{.025\textwidth}
\begin{subfigure}{.3 \textwidth}
\centering
\includegraphics[width=1\linewidth]{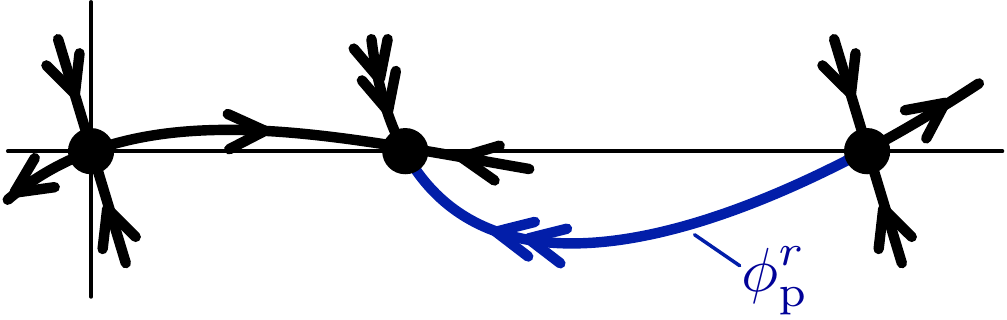}
\caption{$c=c_\mathrm{p}$}
\end{subfigure}
\hspace{.025\textwidth}
\begin{subfigure}{.3 \textwidth}
\centering
\includegraphics[width=1\linewidth]{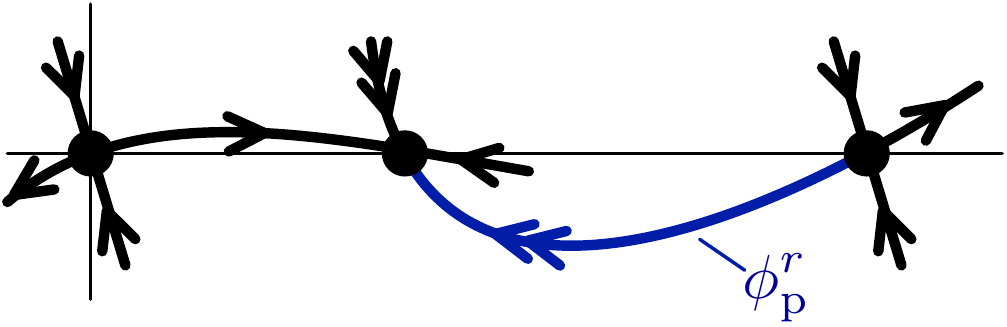}
\caption{$c=c_\mathrm{p}$}
\end{subfigure}
\hspace{.025\textwidth}
\begin{subfigure}{.3 \textwidth}
\centering
\includegraphics[width=1\linewidth]{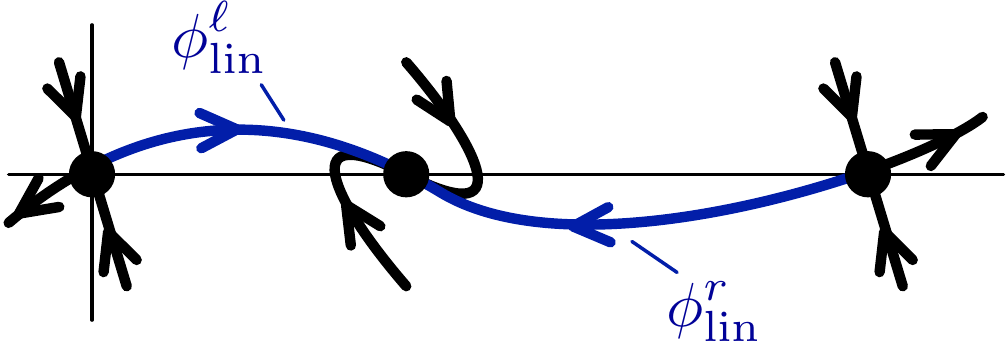}
\caption{$c=c_\mathrm{lin}$}
\end{subfigure}
\hspace{.025\textwidth}\\

\hspace{.025\textwidth}
\begin{subfigure}{.3 \textwidth}
\centering
\includegraphics[width=1\linewidth]{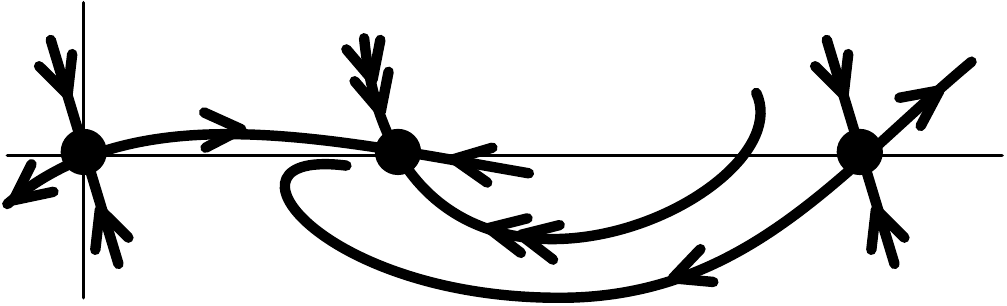}
\caption{$c_\mathrm{bs}<c<c_\mathrm{p}$}
\end{subfigure}
\hspace{.025\textwidth}
\begin{subfigure}{.3 \textwidth}
\centering
\includegraphics[width=1\linewidth]{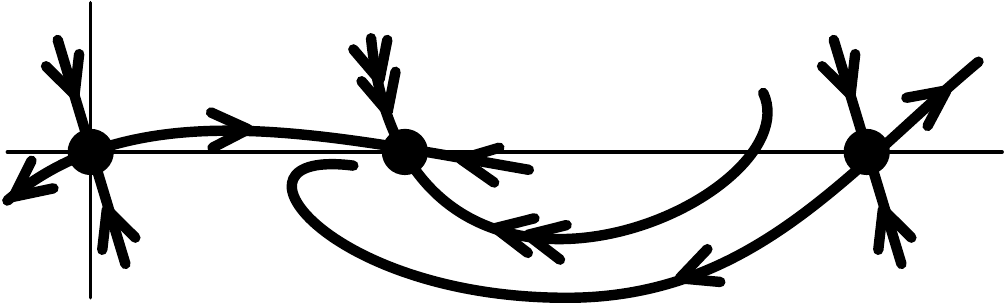}
\caption{$c_\mathrm{lin}<c<c_\mathrm{p}$}
\end{subfigure}
\hspace{.025\textwidth}
\begin{subfigure}{.3 \textwidth}
\centering
\includegraphics[width=1\linewidth]{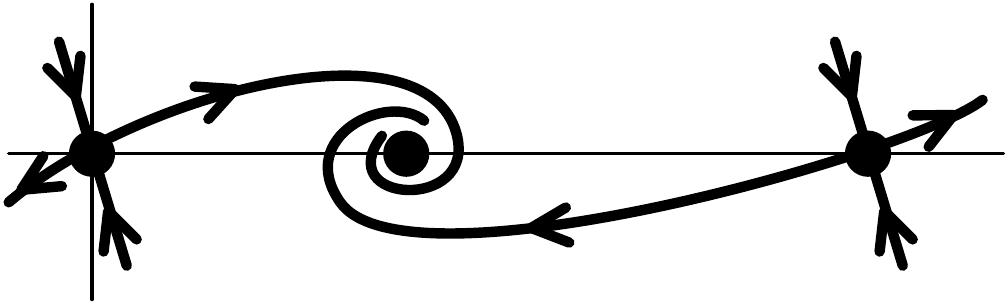}
\caption{$c_\mathrm{bs}<c<c_\mathrm{lin}$}
\end{subfigure}
\hspace{.025\textwidth}\\

\hspace{.025\textwidth}
\begin{subfigure}{.3 \textwidth}
\centering
\includegraphics[width=1\linewidth]{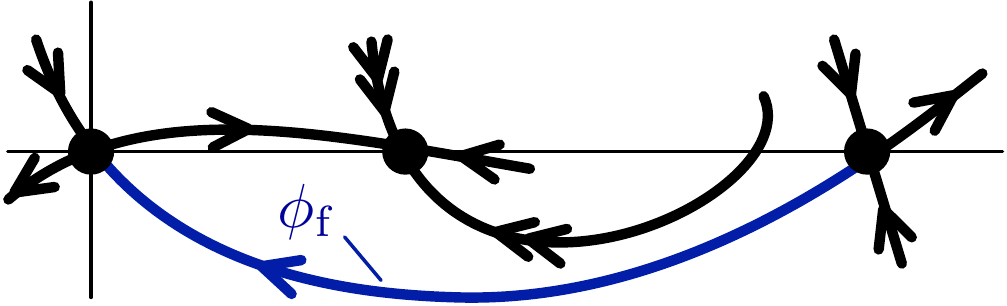}
\caption{$c=c_\mathrm{bs}$}
\end{subfigure}
\hspace{.025\textwidth}
\begin{subfigure}{.3 \textwidth}
\centering
\includegraphics[width=1\linewidth]{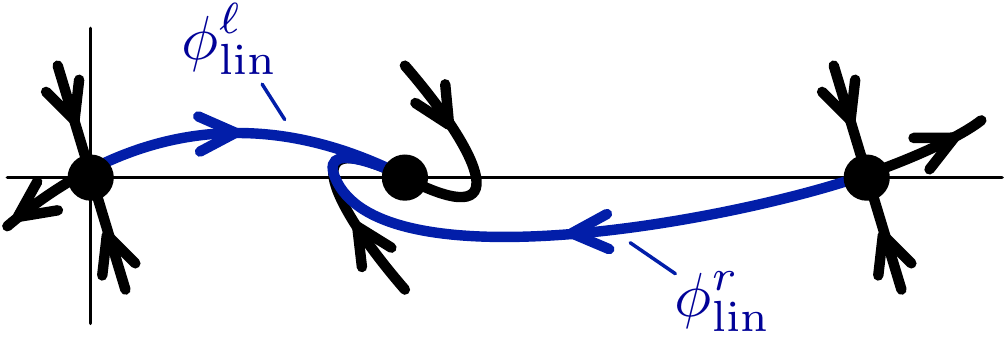}
\caption{$c=c_\mathrm{lin}$}
\end{subfigure}
\hspace{.025\textwidth}
\begin{subfigure}{.3 \textwidth}
\centering
\includegraphics[width=1\linewidth]{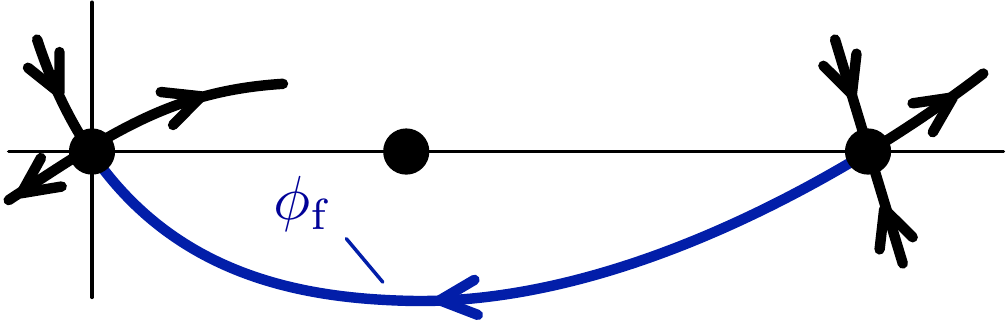}
\caption{$c=c_\mathrm{bs}$}
\end{subfigure}
\hspace{.025\textwidth}\\

\hspace{.025\textwidth}
\begin{subfigure}{.3 \textwidth}
\centering
\includegraphics[width=1\linewidth]{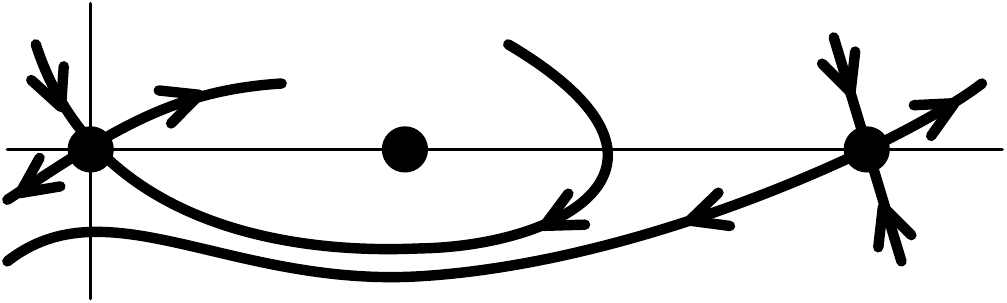}
\caption{$0<c<c_\mathrm{bs}$}
\end{subfigure}
\hspace{.025\textwidth}
\begin{subfigure}{.3 \textwidth}
\centering
\includegraphics[width=1\linewidth]{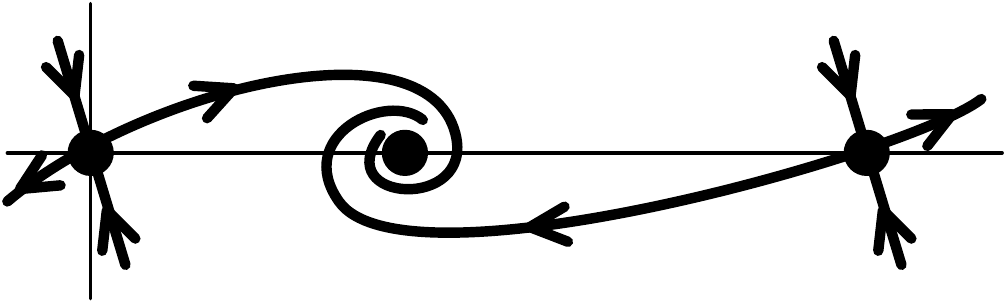}
\caption{$c_\mathrm{bs}<c<c_\mathrm{lin}$}
\end{subfigure}
\hspace{.025\textwidth}
\begin{subfigure}{.3 \textwidth}
\centering
\includegraphics[width=1\linewidth]{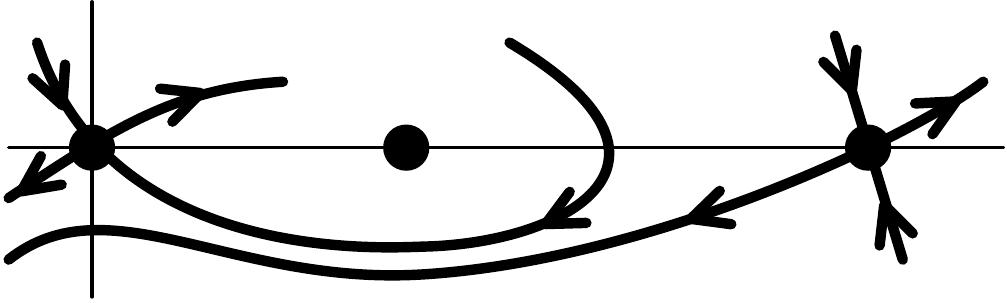}
\caption{$0<c<c_\mathrm{bs}$}
\end{subfigure}
\hspace{.025\textwidth}\\

\hspace{.025\textwidth}
\hspace{.3\textwidth}
\hspace{.025\textwidth}
\begin{subfigure}{.3 \textwidth}
\centering
\includegraphics[width=1\linewidth]{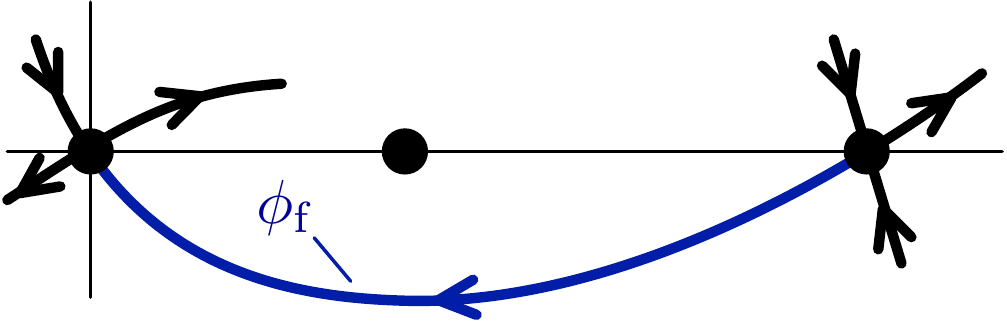}\caption{$c=c_\mathrm{bs}$}
\end{subfigure}
\hspace{.025\textwidth}
\hspace{.3\textwidth}
\hspace{.025\textwidth}\\

\hspace{.025\textwidth}
\hspace{.3\textwidth}
\hspace{.025\textwidth}
\begin{subfigure}{.3 \textwidth}
\centering
\includegraphics[width=1\linewidth]{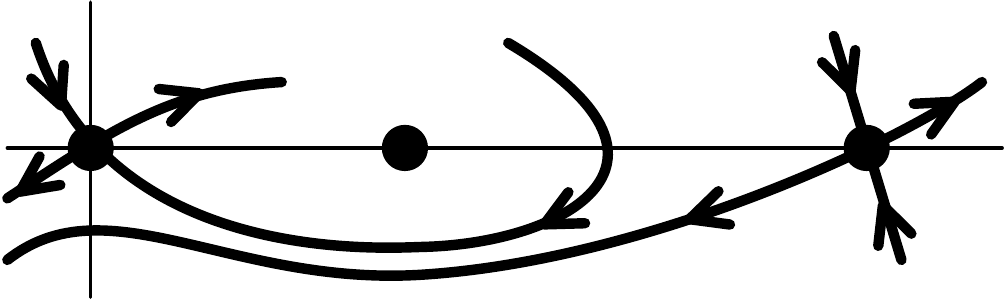}
\caption{$0<c<c_\mathrm{bs}$}
\end{subfigure}
\hspace{.025\textwidth}
\hspace{.3\textwidth}
\hspace{.025\textwidth}
\caption{Shown are the phase portraits in $(u,v)$-space for the layer problem~\eqref{eq_layer} for values of $0<a<a_\mathrm{b}$ (left column), $a_\mathrm{b}<a<1/3$ (middle column), $1/3<a<1/2$ (right column) and $c>0$. The equilibria $p_1(0),p_2(0),p_3(0)$ are depicted from left to right along the $u$-axis in each figure.}
\label{fig:layer_ppps}
\end{figure}

\begin{figure}
\centering
\includegraphics[width=0.6\linewidth]{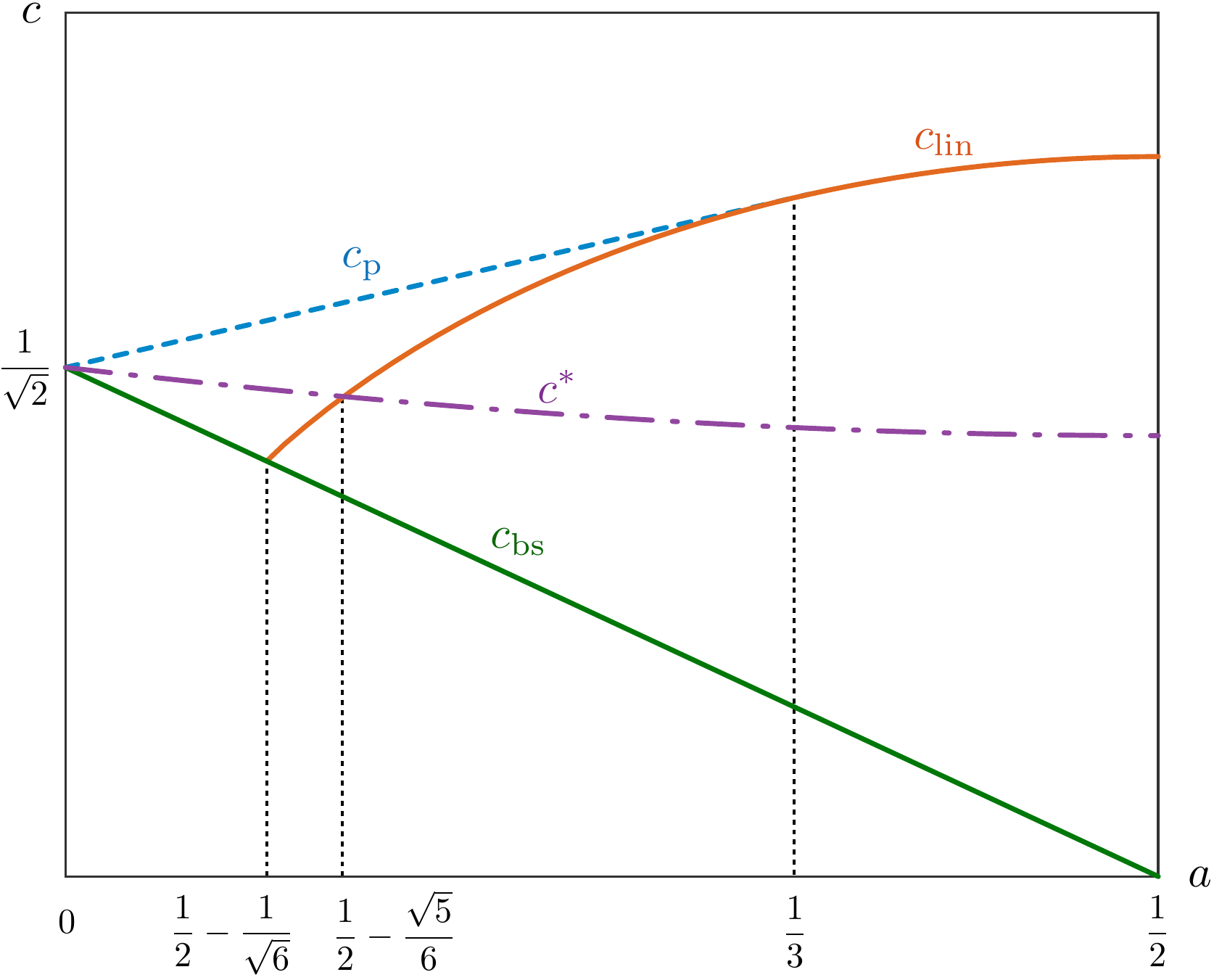}
\caption{Shown is the relation of the speeds $c_\mathrm{bs}(a), c_\mathrm{p}(a), c_\mathrm{lin}(a)$ for $0<a<1/2$. Also shown is the speed $c^*(a)$ which separates the speeds corresponding to periodic orbits in the hyperbolic regime versus those occurring in the nonhyperbolic regime $c\geq c^*(a)$.}
\label{fig:speeds}
\end{figure}

There are two cases of interest: The first is the pulled case in which the fronts coincide with the existence of a double root in the linearization of $p_2(0)$, which we can easily see occurs for $c=c_\mathrm{lin}=2\sqrt{a(1-a)}$ whenever $c_\mathrm{lin}>c_\mathrm{bs}$, which holds for each value of $a_\mathrm{b}<a<1/2$. We denote the associated fronts by $\phi^\ell_\mathrm{lin}, \phi^r_\mathrm{lin}$; see Figure~\ref{fig:layer_ppps}.

The second case is the pushed case in which the equilibrium $p_2(0)$ is a stable node and the front under consideration approaches $p_2(0)$ along a strong stable direction: For $c>2\sqrt{a(1-a)}$, the equilibrium $p_2(0)$ is a stable node with two distinct real eigenvalues. There is a unique solution which approaches $p_2(0)$ along the stronger eigendirection.

In both cases, the rate of decay of the heteroclinic is locally minimal, as one can readily see from the computation of the eigenvalues. 

To determine the wavespeed of this front, we compute the solutions directly. To find $\phi^r_\mathrm{p}$, using the ansatz $v=b(u-1)(u-a)$ we find
\begin{align}
\begin{split}
2b^2u-b^2(1+a)&=-cb+u
\end{split}
\end{align}
from which we compute $b=\frac{1}{\sqrt{2}}$ and $c=c_\mathrm{p}=\frac{1}{\sqrt{2}}(1+a)$ for $0<a<1/3$, and we obtain the explicit solutions
\begin{align}\label{eq_pushedexplicit}
\phi^r_\mathrm{p}(\xi)=\begin{pmatrix} u_\mathrm{p}(\xi)\\v_\mathrm{p}(\xi)\end{pmatrix} = \begin{pmatrix}a+ \frac{1-a}{2}\left(1-\tanh\left(\frac{1-a}{2\sqrt{2}}\xi\right)\right)\\-\frac{(1-a)^2}{4\sqrt{2}}\sech^2\left(\frac{1-a}{2\sqrt{2}}\xi\right)\end{pmatrix}.
\end{align}

We note that there is no pushed front $\phi^r_\mathrm{p}(a)$ for $a>1/3$ and no pushed front $\phi^\ell_\mathrm{p}$ from $p_1(0)$ to $p_2(0)$ for any value of $a$.

\subsection{Periodic orbits}\label{sec_periodicorbits}
In this section, we collect results regarding periodic orbits in the full system for $0<\eps\ll 1$. 

\subsubsection{Hyperbolic regime}
For each value of the wavespeed $0< c <c^*(a)=\sqrt{\frac{1-a+a^2}{2}}$, there exists $w_\mathrm{f}(c),w_\mathrm{b}(c)\in (w_\ell, w_r)$ and fronts $\phi_\mathrm{f}(c), \phi_\mathrm{b}(c)$ which connect the equilibria $p_3(w_\mathrm{f}), p_1(w_\mathrm{f})$ and $p_1(w_\mathrm{b}),p_3(w_\mathrm{b})$, respectively.
\begin{figure}
\centering
\includegraphics[width=0.7\linewidth]{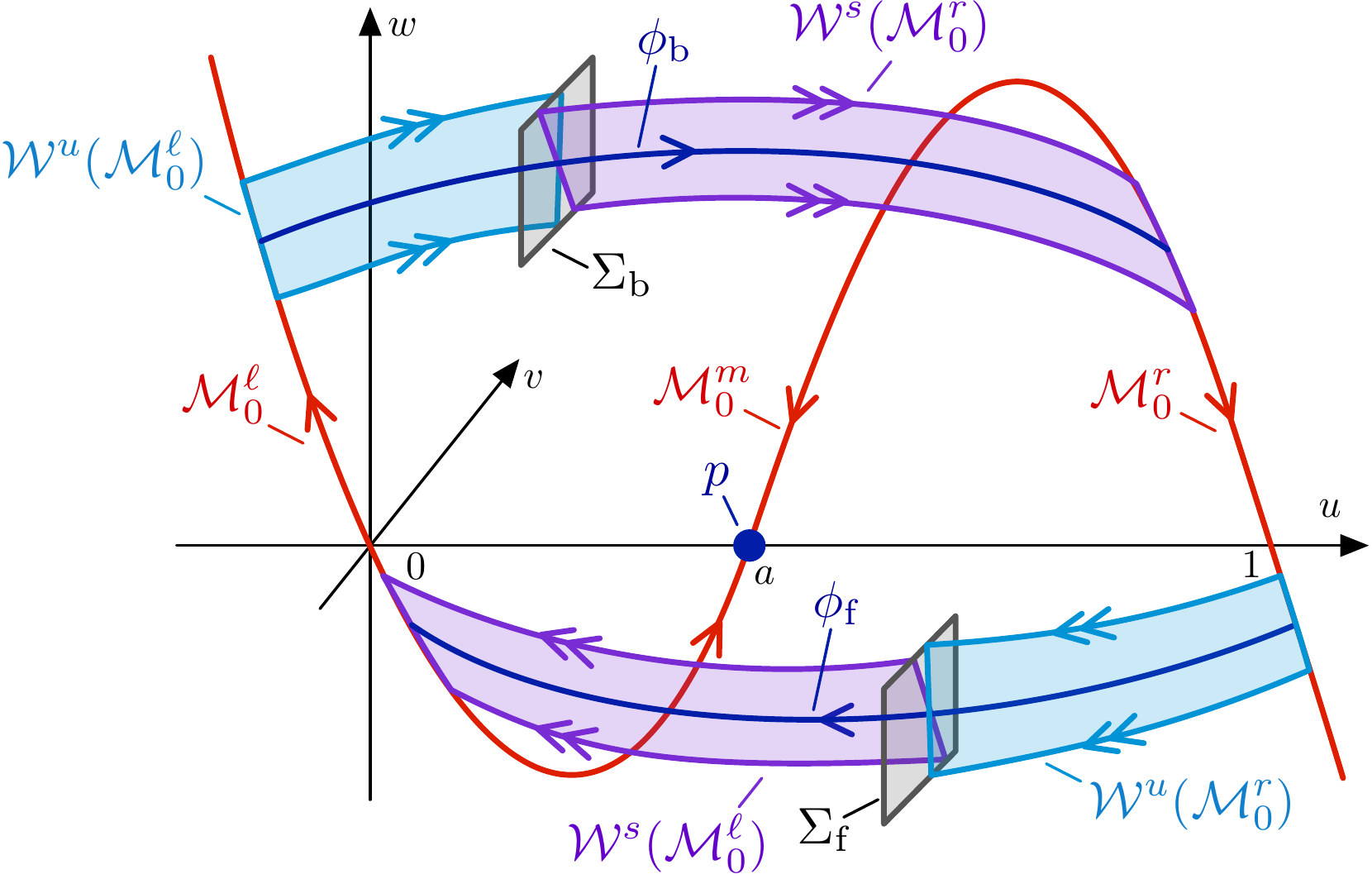}
\caption{Shown is the geometry for of the singular periodic orbit $\Gamma_0(c)$ in the hyperbolic regime $0< c <c^*(a)$. The orbit $\Gamma_0(c)$ is formed by traversing the singular segments $\phi_\mathrm{f}(c),\mathcal{M}^\ell_0,\phi_\mathrm{b}(c),\mathcal{M}^r_0$ in turn.}
\label{fig:PO_hyp}
\end{figure}

By following the front $\phi_\mathrm{f}(c)$, then the slow manifold $\mathcal{M}^\ell_0$, then the front $\phi_\mathrm{b}(c)$, and finally the slow manifold $\mathcal{M}^r_0$, we obtain a singular periodic orbit $\Gamma_0(c)$; see Figure~\ref{fig:PO_hyp} for the geometry of the setup in phase space. It is known~\cite{STR} that for sufficiently small $\eps>0$, this singular structure perturbs to a periodic orbit $\Gamma_\eps(c)$ of the full system.

It is further known~\cite{ESZ} that these periodic orbits are spectrally stable in the original PDE; their spectra is contained in the open left plane, except for a simple eigenvalue at the origin due to translation invariance. This implies that these solutions are in fact nonlinearly stable~\cite{HEN}.

\subsubsection{Nonhyperbolic regime}
For values of the wavespeed $c\geq c^*(a)$, there exist fronts $\phi_\mathrm{f}(c), \phi_\mathrm{b}(c)$ in the fast subsystem in the planes $w=w_\ell$ and $w=w_r$, respectively. These fronts connect each of the fold points of the critical manifold with the saddle equilibrium on the opposite branch. For the critical wavespeed $c=c^*(a)$, these fronts approach the fold along a strong stable direction, while for $c>c^*(a)$, they approach the fold along a center manifold with algebraic decay. 

\begin{figure}
\centering
\includegraphics[width=0.7\linewidth]{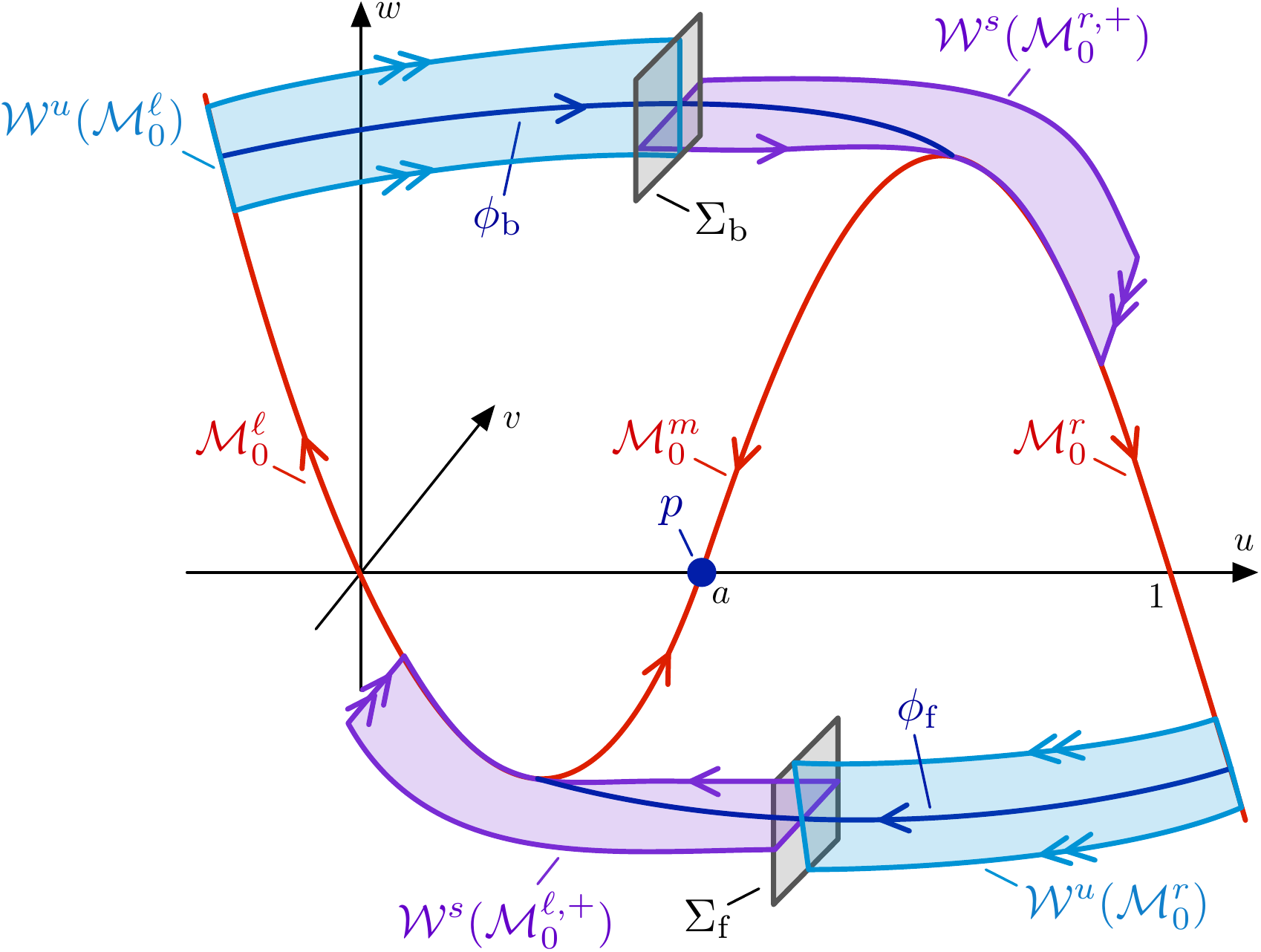}
\caption{Shown is the geometry for of the singular periodic orbit $\Gamma_0(c)$ in the nonhyperbolic regime $c\geq c^*(a)$. The orbit $\Gamma_0(c)$ is formed by traversing the singular segments $\phi_\mathrm{f}(c),\mathcal{M}^\ell_0,\phi_\mathrm{b}(c),\mathcal{M}^r_0$ in turn.}
\label{fig:PO_nonhyp}
\end{figure}

There is still a singular periodic orbit $\Gamma_0(c)$ obtained as before, by following $\phi_\mathrm{f}(c)$, then $\mathcal{M}^\ell_0$, then $\phi_\mathrm{b}(c)$, and finally the slow manifold $\mathcal{M}^r_0$. See Figure~\ref{fig:PO_nonhyp} for the setup. However, persistence of these orbits for small $\eps>0$ and their stability is not known. We refer to this regime as the nonhyperbolic regime due to the loss of normal hyperbolicity at the fold points on the critical manifold.

We have the following proposition regarding the existence of periodic orbits in the nonhyperbolic regime. Its proof will be given in~\S\ref{sec_nhpoproof}.

\begin{Proposition}\label{prop_nh_periodicorbits}
Fix $0<a<1/2$. There exists $\delta_c>0$ such that for each $c> c^*(a)-\delta_c$, and for each sufficiently small $0<\eps\ll1$, there exists a periodic orbit $\Gamma_\eps(c)$ which is $\mathcal{O}(\eps^{2/3})$-close to $\Gamma_0(c)$. 
\end{Proposition}

\begin{Remark}
The $\mathcal{O}(\eps^{2/3})$ estimate in Proposition~\ref{prop_nh_periodicorbits} appears due to the passage near the fold points on the critical manifold.
\end{Remark}

\begin{Remark}
Proposition~\ref{prop_nh_periodicorbits} guarantees, for fixed $a$, the existence of periodic orbits $\Gamma_\eps(c)$ for $c> c^*(a)-\delta_c$, where $\delta_c>0$ is independent of $\eps$. This family of orbits therefore overlaps with those constructed in the hyperbolic regime, forming the single family described in Theorem~\ref{thm_periodicexistence}.
\end{Remark}

\begin{Remark}
We do not consider the PDE stability of the periodic orbits from Proposition~\ref{prop_nh_periodicorbits} in this work. However, we expect that the methods used in~\cite{ESZ} to obtain spectral stability of the periodic orbits in the hyperbolic regime extend to the nonhyperbolic regime with appropriate modifications. \end{Remark}

\subsection{Singular traveling fronts}\label{sec_singularsolns}

Combining results from the previous sections concerning the fast/slow dynamics, we construct singular pushed/pulled fronts which connect the equilibrium at $(a,0,0)$ to one of the singular periodic orbits $\Gamma_0(c)$.

\subsubsection{Pulled fronts}
We recall from~\S\ref{sec_layer} that for $c=c_\mathrm{lin}(a)=2\sqrt{a(1-a)}$ for each value of $0<a<1/2$, there exist pulled fronts $\phi^\ell_\mathrm{lin}(a),\phi^r_\mathrm{lin}(a)$ which connect $(a,0,0)$ to the left and right branches $\mathcal{M}^\ell_0$ and $\mathcal{M}^r_0$ of the critical manifold, respectively.

Further, from~\S\ref{sec_periodicorbits}, for each $c>0$ there exists a singular periodic orbit $\Gamma_0(c)$. For $c <c^*(a)=\sqrt{\frac{1-a+a^2}{2}}$, this periodic orbit is in the hyperbolic regime, while $c\geq c^*(a)$ constitutes the nonhyperbolic regime.

Therefore, for $0<a<\frac{1}{2}-\frac{\sqrt{5}}{6}$, at $c=c_\mathrm{lin}(a)$, there are pulled fronts coinciding with a singular periodic orbit $\Gamma_0(c)$ in the hyperbolic regime, while for $\frac{1}{2}-\frac{\sqrt{5}}{6}<a<\frac{1}{2}$, the pulled fronts coincide with a periodic orbit in the nonhyperbolic regime.

This allows us to construct singular pulled fronts as follows: For each value of the wave speed $c=c_\mathrm{lin}(a)$ for certain values of $0<a<1/2$, there is a singular invasion front $F^r_0(a)$ which follows the periodic orbit $\Gamma_0(c_\mathrm{lin}(a))$, then a segment of the right branch $\mathcal{M}^r_0$ of the critical manifold, then the front $\phi^r_\mathrm{lin}(a)$. This can only happen in the case that the periodic orbit traverses a fast segment $\phi_\mathrm{f}$ which lies below $w=0$; otherwise the front $\phi^r_\mathrm{lin}(a)$ does not exist. This happens whenever $c_\mathrm{lin}(a)>c_\mathrm{bs}=\frac{1}{\sqrt{2}}(1-2a)$, which occurs for $a_\mathrm{b}<a<\frac{1}{2}$. Hence there is a singular invasion front $F^r_0(a)$ for each $a_\mathrm{b}<a<\frac{1}{2}$. The geometry of the singular pulled front is shown in Figure~\ref{fig:singular_front_pulled}.

\begin{figure}
\centering
\includegraphics[width=0.7\linewidth]{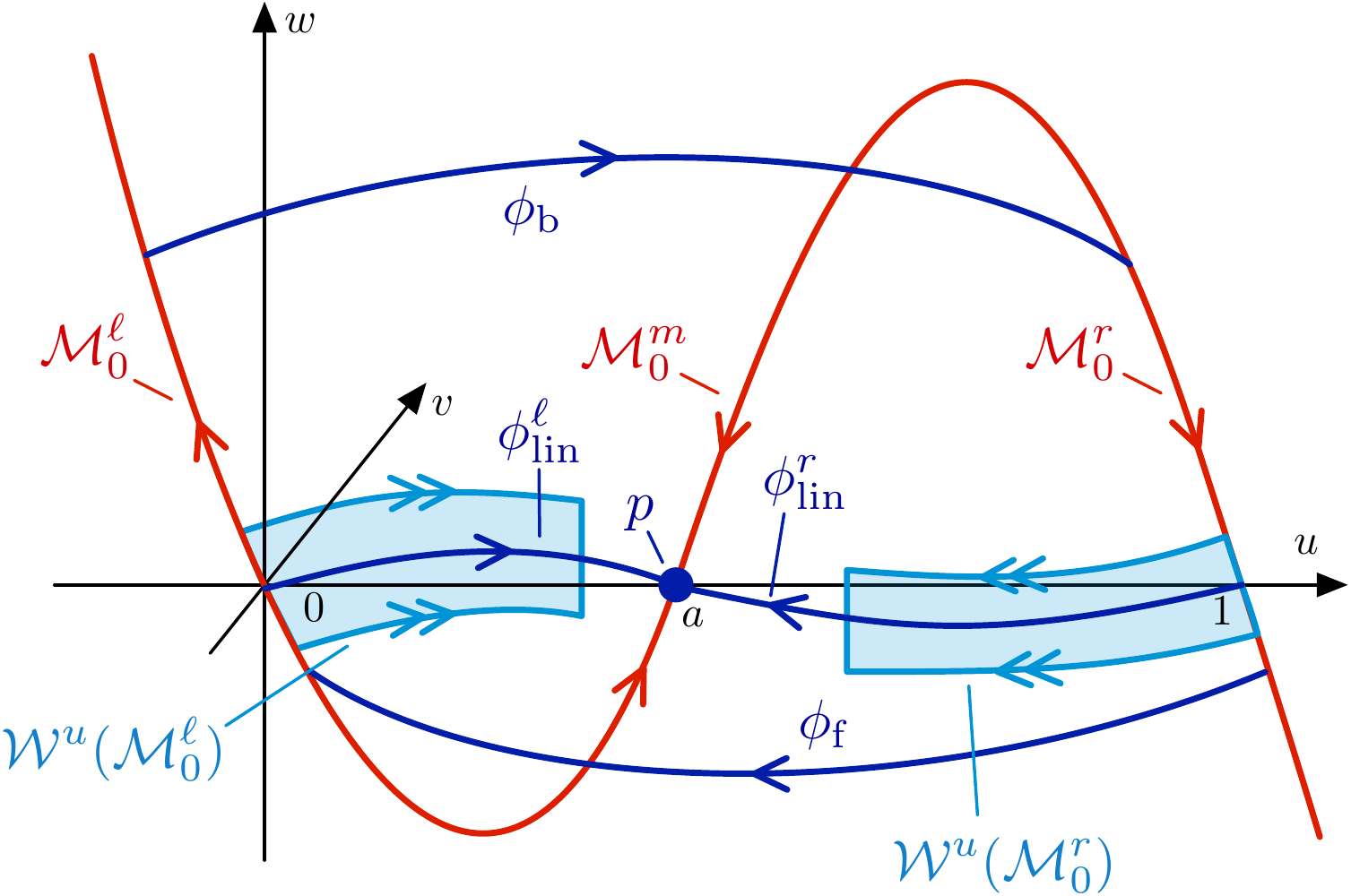}
\caption{Shown is the geometry of the singular pulled fronts $F^r_0(a)$ and $F^\ell_0(a)$. The front $F^r_0(a)$ is formed by following the singular periodic orbit $\Gamma_0(c_\mathrm{lin}(a))$ formed by $\phi_\mathrm{f}, \mathcal{M}^\ell_0, \phi_\mathrm{b}, \mathcal{M}^r_0$, then a segment of the right branch $\mathcal{M}^r_0$ of the critical manifold, then the front $\phi^r_\mathrm{lin}(a)$; $F^\ell_0(a)$ is formed by following $\Gamma_0(c_\mathrm{lin}(a))$, then a segment of the left branch $\mathcal{M}^\ell_0$ of the critical manifold, then the front $\phi^\ell_\mathrm{lin}(a)$.}
\label{fig:singular_front_pulled}
\end{figure}

For the same values of $a_\mathrm{b}<a<\frac{1}{2}$, there is a second front $F^\ell_0(a)$ which follows $\Gamma_0(c_\mathrm{lin}(a))$, then a segment of the left branch $\mathcal{M}^\ell_0$ of the critical manifold, then the front $\phi^\ell_\mathrm{lin}(a)$. 

\subsubsection{Pushed fronts}
Similarly to the pulled case, we construct singular invasion fronts $P^r_0(a)$ at $c=c_\mathrm{p}(a)=\frac{1}{\sqrt{2}}(1+a)$ for each value of $0<a<1/3$ by following the periodic orbit $\Gamma_0(c_\mathrm{p}(a))$, then a segment of the right branch $\mathcal{M}^r_0$ of the critical manifold, then the front $\phi^r_\mathrm{p}(a)$; see Figure~\ref{fig:singular_front_pushed}. We note that the associated periodic orbit $\Gamma_0(c_\mathrm{p}(a))$ always occurs in the nonhyperbolic regime. 

\begin{figure}
\centering
\includegraphics[width=0.7\linewidth]{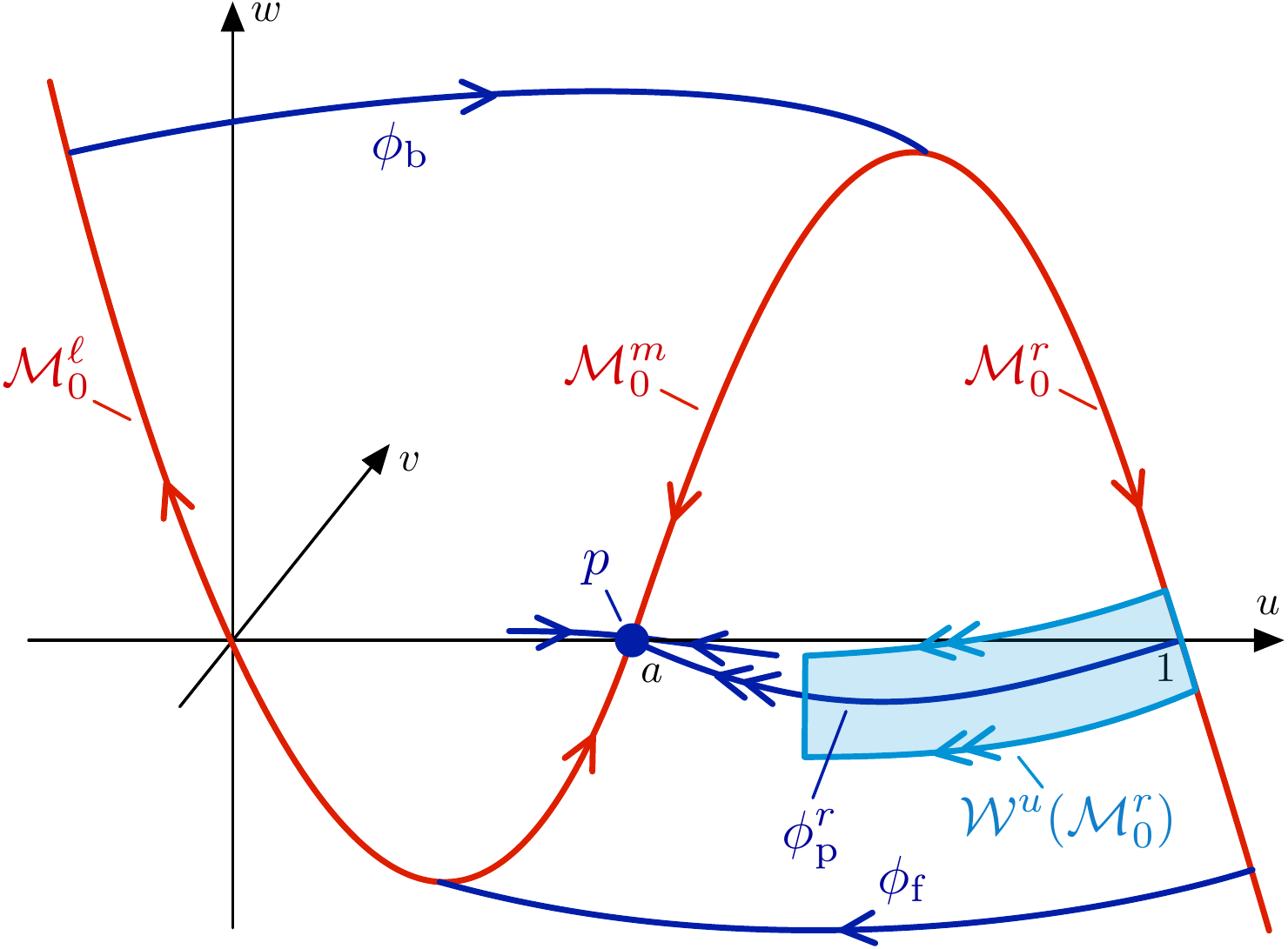}
\caption{Shown is the geometry of the singular pushed front $P^r_0(a)$, which is formed by concatenating the singular periodic orbit $\Gamma_0(c_\mathrm{lin}(a))$ formed by $\phi_\mathrm{f}, \mathcal{M}^\ell_0, \phi_\mathrm{b}, \mathcal{M}^r_0$, with a segment of the right branch $\mathcal{M}^r_0$ of the critical manifold, then the front $\phi^r_\mathrm{p}(a)$.}
\label{fig:singular_front_pushed}
\end{figure}

\section{Persistence of solutions for $0<\eps\ll1$}\label{sec_persistence}

We use geometric singular perturbation theory to show that the singular solutions constructed in~\S\ref{sec_singularsolns} perturb to solutions of the full problem for sufficiently small $\eps>0$. Using information about the periodic orbits $\Gamma_\eps(c)$ for small $\eps$, we find the desired invasion fronts as intersections of the unstable manifold $\mathcal{W}^\mathrm{u}(\Gamma_\eps(c))$ of the periodic orbit and solutions with appropriate decay properties on the stable manifold $\mathcal{W}^\mathrm{s}(p)$ of the equilibrium.

In~\S\ref{sec_persistenceinvtmflds}, we collecting results of Fenichel theory regarding persistence of invariant manifolds for $0<\eps \ll1$, and we further describe the flow near the fold points on the critical manifold. A proof of Proposition~\ref{prop_nh_periodicorbits} concerning the existence of periodic orbits in the nonhyperbolic regime is given in~\S\ref{sec_nhpoproof}. We then prove a technical proposition in~\S\ref{sec_poinvariantmanifolds} concerning the closeness of the stable/unstable manifolds of the periodic orbits of Theorem~\ref{thm_periodicexistence} to the stable/unstable manifolds of the left and right slow manifolds $\mathcal{M}^{\ell,r}_\eps$. Finally we complete the proofs of Theorems~\ref{thm_periodicexistence},~\ref{thm_pulledexistence} and~\ref{thm_pushedexistence} in~\S\ref{sec_periodicorbitthmproof},~\S\ref{sec_pulledproof} and~\S\ref{sec_pushedproof}, respectively.

\subsection{Persistence of invariant manifolds}\label{sec_persistenceinvtmflds}
We first recall some standard results from geometric singular perturbation theory. Away from the fold points, the branches $\mathcal{M}^\ell_0, \mathcal{M}^m_0, \mathcal{M}^r_0$ are normally hyperbolic and therefore persist as locally invariant slow manifolds $\mathcal{M}^\ell_\eps, \mathcal{M}^m_\eps, \mathcal{M}^r_\eps$, on which the flow is an $\mathcal{O}(\eps)$ perturbation of the slow flow~\eqref{eq_reduced}. Therefore the flow on $\mathcal{M}^\ell_\eps$ satisfies $u'<0$, on  $\mathcal{M}^r_\eps$ we have $u'>0$, and on $\mathcal{M}^m_\eps$ there is a single equilibrium at $p$ which is attracting for the slow flow.

Further the two dimensional stable/unstable manifolds $\mathcal{W}^{s/u}(\mathcal{M}^\ell_0)$ formed as the union of the stable/unstable manifolds of the saddle equilibria $p_1(w)$ of the layer equations~\eqref{eq_layer} persist for $0<\eps \ll 1$ as locally invariant two-dimensional stable/unstable manifolds $\mathcal{W}^{s/u}(\mathcal{M}^\ell_\eps)$ of the perturbed slow manifold $\mathcal{M}^\ell_\eps$. These manifolds are foliated by stable/unstable fibers which form an invariant family. The same holds for the saddle-type critical manifold $\mathcal{M}^r_0$ and we obtain corresponding perturbed stable/unstable manifolds $\mathcal{W}^{s/u}(\mathcal{M}^r_\eps)$. 

The middle branch $\mathcal{M}^m_0$ is completely stable and hence has a three dimensional stable manifold $\mathcal{W}^\mathrm{s}(\mathcal{M}^m_0)$ which persists as $\mathcal{W}^\mathrm{s}(\mathcal{M}^m_\eps)$ for sufficiently small $\eps>0$. Since the equilibrium $p$ on $\mathcal{M}^m_\eps$ is attracting under the slow flow, the manifold $\mathcal{W}^\mathrm{s}(\mathcal{M}^m_\eps)$ in fact coincides with the stable manifold $\mathcal{W}^\mathrm{s}(p)$. This stable manifold can be separated into slow/fast components: In the fast system~\eqref{eq_twode} for $\eps=0$ there are two strictly negative eigenvalues, while a third negative eigenvalue which is $\mathcal{O}(\eps)$ comes from the slow flow tangential to the slow manifold $\mathcal{M}^m_\eps$ for $\eps>0$. Hence the equilibrium $p$ has a two dimensional strong stable manifold $\mathcal{W}^\mathrm{ss}(p)$ whose tangent space is $\mathcal{O}(\eps)$-close to the plane $\{w=0\}$ (when $\eps=0$, $\mathcal{W}^\mathrm{ss}(p)$ lies entirely in this plane).

\subsubsection{Local analysis near the folds}
We consider the lower left fold $p_\ell$. The analysis near the upper fold is similar. The lower fold point is given by the fixed point $p_\ell=(u_\ell,0,w_\ell)$ of the layer problem~\eqref{eq_layer} where from~\eqref{eq_uell} we have
\begin{align*}
u_\ell = \frac{1}{3}\left(a+1-\sqrt{1-a+a^2}\right),
\end{align*}
and $w_\ell=f(u_\ell)$. The linearization of~\eqref{eq_twode} about this fixed point has one negative real eigenvalue $-c < 0$ and a double zero eigenvalue, since $f'(u_\ell) = 0$. 

We note that the geometry near the fold is similar to that considered in~\cite[\S4]{cas} (in fact the scenarios are identical up to a reversal of time and change of orientation), and hence we draw on the local analysis as presented in~\cite{cas}. We first move to a local coordinate system in a neighborhood of $p_\ell$: there exists a neighborhood $V_\ell$ of $p_\ell$, in which we can perform the following $C^k$-change of coordinates $\Phi_\eps \colon V_\ell \to \mathbb{R}^3$ to~\eqref{eq_twode}, which is $C^k$-smooth in $c\in I_c$ and $\eps>0$ sufficiently small. We apply $\Phi_\eps$ in the neighborhood of the fold point and rescale time by a positive constant so that~\eqref{eq_twode} becomes
\begin{align}
\begin{split}
\label{eq_canonicalfold}
x' &= \left(y- x^2 +h(x,y,\eps;c)\right), \\
y' &=\eps g(x,y,\eps;c), \\
z' &= z\left(-\theta+\mathcal{O}(x,y,z,\eps)\right), \\
\end{split}
\end{align}
where $\theta$ is a positive constant, and $h,g$ are $C^k$-functions satisfying
\begin{align*}\begin{split}
h(x,y,\eps;c) &= \mathcal{O}(\eps,xy,y^2,x^3), \\
g(x,y,\eps;c) &= 1+ \mathcal{O}(x,y,\eps),
\end{split}\end{align*}
uniformly in $c\in I_c$. In the transformed system~\eqref{eq_canonicalfold}, the $x,y$-dynamics is decoupled from the dynamics in the $z$-direction along the straightened out strong stable fibers. Thus, the flow is fully described by the dynamics on the two-dimensional invariant center manifold $z = 0$ and by the one-dimensional dynamics along the fibers in the $z$-direction. Figure~\ref{fig:singularfold} depicts the singular $\eps=0$ flow of~\eqref{eq_canonicalfold}.

\begin{figure}[h]
\centering
\includegraphics[width=0.5\linewidth]{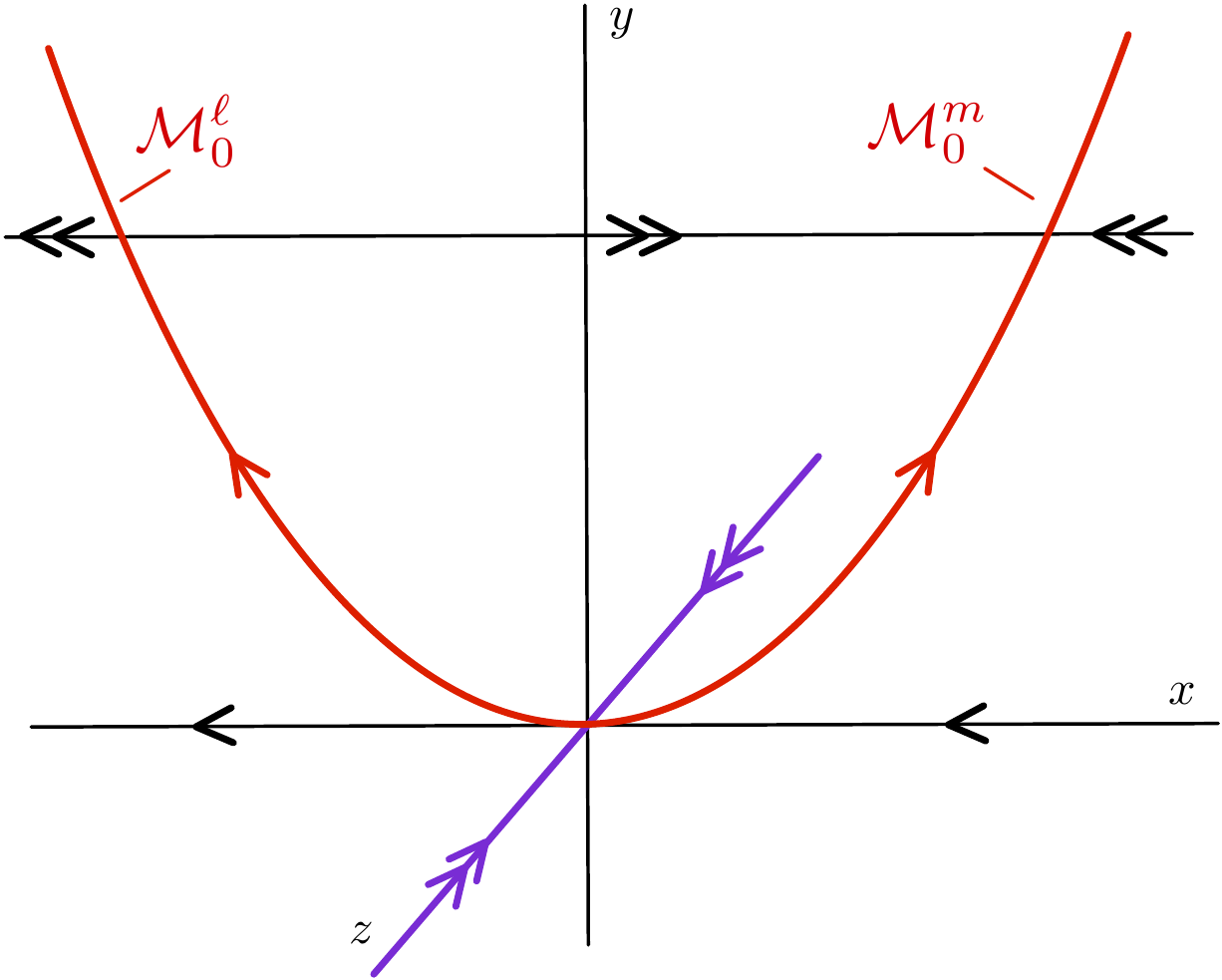}
\caption{Shown are the singular $\eps=0$ dynamics in the local coordinates of the lower left fold point. The trajectory $\mathcal{M}^{\ell,+}_0$ is formed by concatenating $\mathcal{M}^\ell_0$ with the positive $x$-axis.}
\label{fig:singularfold}
\end{figure}

We consider the flow of~\eqref{eq_canonicalfold} on the invariant manifold $z=0$. We append an equation for $\eps$, arriving at the system
\begin{align}
\begin{split}
\label{eq_caonicalfoldctr}
x' &= y- x^2 +h(x,y,\eps;c), \\
y' &=   \eps g(x,y,\eps;c), \\
\eps'&= 0.
\end{split}
\end{align}
For $\eps=0$, this system possesses a critical manifold given by $\{(x,y):y-x^2+h(x,y,0;c)=0\}$, which in a sufficiently small neighborhood of the origin is shaped as a parabola opening upwards. The branch of this parabola for $x<0$ is repelling and corresponds to the manifold $\mathcal{M}^\ell_0$. We define $\mathcal{M}^{\ell,+}_0$ to be the singular trajectory obtained by appending the fast trajectory given by the line $\{(x,0): x>0\}$ to the repelling branch $\mathcal{M}^\ell_0$ of the critical manifold; see Figure~\ref{fig:singularfold}. In~\cite[\S4]{cas} it was shown that, for sufficiently small $\eps>0$, $\mathcal{M}^{\ell,+}_0$ perturbs to a trajectory $\mathcal{M}^{\ell,+}_\eps$ on $z = 0$, represented as a graph $y = y^\mathrm{s}_\eps(x;c)$, which is $\mathcal{O}(\eps^{2/3})$-close  in  $C^0$ and $\mathcal{O}(\eps^{1/3})$-close in  $C^1$ to $\mathcal{M}^{\ell,+}_0$, uniformly in $c\in I_c$. Therefore, the manifold $\mathcal{W}^\mathrm{s}(\mathcal{M}^{\ell,+}_0)$ composed of the strong stable fibers of the singular trajectory $\mathcal{M}^{\ell,+}_0$ also perturbs to a two-dimensional locally invariant manifold $\mathcal{W}^\mathrm{s}(\mathcal{M}^{\ell,+}_\eps)$ which is $\mathcal{O}(\eps^{2/3})$-close in $C^0$ and $\mathcal{O}(\eps^{1/3})$-close  in $C^1$ to $\mathcal{W}^\mathrm{s}(\mathcal{M}^{\ell,+}_0)$.

\begin{figure}
\centering
\includegraphics[width=0.6\linewidth]{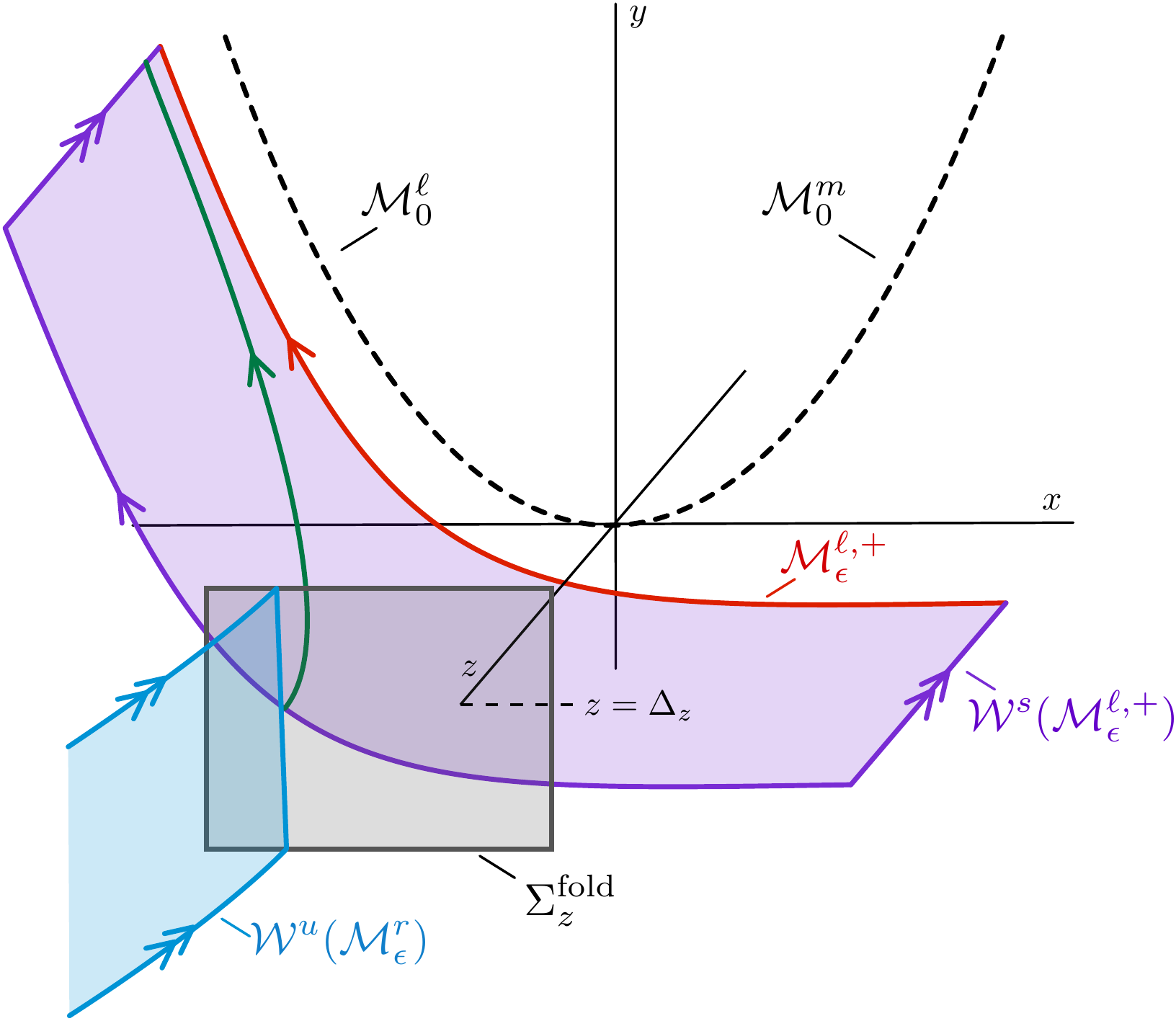}
\caption{Shown is the geometry near the lower left fold point for $c\approx c^*(a)$ and $0<\eps\ll 1$. The manifold $\mathcal{W}^\mathrm{s}(\mathcal{M}^{\ell,+}_\eps)$ is formed by the stable fibers of the trajectory $\mathcal{M}^{\ell,+}_\eps$. The manifolds $\mathcal{W}^\mathrm{s}(\mathcal{M}^{\ell,+}_\eps)$ and $\mathcal{W}^\mathrm{u}(\mathcal{M}^r_\eps)$ intersect transversely in the section $\Sigma^\textrm{fold}_z$.}
\label{fig:lower_fold_critical}
\end{figure}

\begin{figure}
\centering
\includegraphics[width=0.6\linewidth]{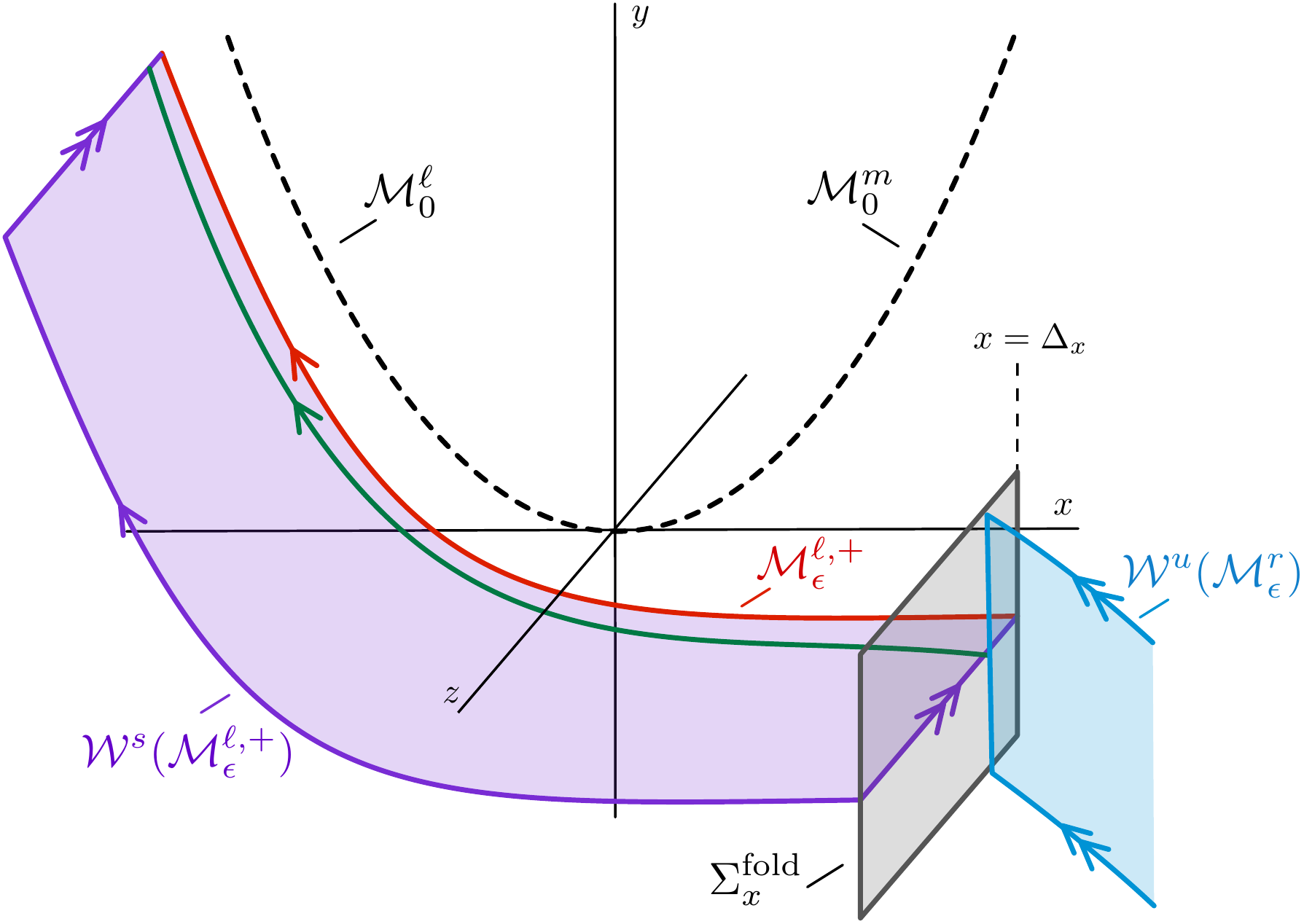}
\caption{Shown is the geometry near the lower left fold point for $c> c^*(a)$ and $0<\eps\ll 1$. The manifolds $\mathcal{W}^\mathrm{s}(\mathcal{M}^{\ell,+}_\eps)$ and $\mathcal{W}^\mathrm{u}(\mathcal{M}^r_\eps)$ intersect transversely in the section $\Sigma^\textrm{fold}_x$.}
\label{fig:lower_fold_center}
\end{figure}

For $c>c^*(a)$, the manifolds $\mathcal{W}^\mathrm{s}(\mathcal{M}^{\ell,+}_0)$ and $\mathcal{W}^\mathrm{u}(\mathcal{M}^r_0)$ intersect along the front $\phi_\mathrm{f}(c)$. The following proposition concerns the transversality of this connection.

\begin{Proposition}
\label{prop_transfront}
Fix $0<a<1/2$. There exists $\delta_c>0$ such that for each $c>c^*(a)-\delta_c$ and each sufficiently small $\eps>0$, the manifolds $\mathcal{W}^\mathrm{s}(\mathcal{M}^{\ell,+}_\eps)$ and $\mathcal{W}^\mathrm{u}(\mathcal{M}^r_\eps)$ intersect transversely.
\end{Proposition}
\begin{proof}
We note that to the right of the fold point (that is, for $x>0$ in the local coordinates in $V_\ell$), the trajectory $\mathcal{M}^{\ell,+}_0$ lies in a plane of constant $w$ since in this region $\mathcal{M}^{\ell,+}_0$ is described by the fast $\eps=0$ flow. Thus we need to show transversality of the manifolds $\mathcal{W}^\mathrm{s}(\mathcal{M}^{\ell,+}_0)$ and $\mathcal{W}^\mathrm{u}(\mathcal{M}^r_0)$ with respect to $w$, which is a parameter for the fast $\eps=0$ flow.

It suffices to prove transversality at $\eps=0$ for each $c\geq c^*(a)$. By the $C^1$-dependence of the manifolds on $c$, this transversality persists for $c>c^*(a)-\delta_c$. The fact that this transversality persists for small $\eps>0$ follows from the $\mathcal{O}(\eps^{1/3})$-closeness of $\mathcal{W}^\mathrm{s}(\mathcal{M}^{\ell,+}_\eps)$ and $\mathcal{W}^\mathrm{s}(\mathcal{M}^{\ell,+}_0)$ as  $C^1$-manifolds. 

To continue, we consider the planar system~\eqref{eq_layer}, for which there exists the heteroclinic connection $\phi_\mathrm{f}(c)$ for each $c\geq c^*(a)$ which connects the equilibria $p_3(w_\ell)$ and $p_1(w_\ell)$.

Thus the manifolds $\mathcal{W}^\mathrm{s}(\mathcal{M}^{\ell,+}_0)$ and $\mathcal{W}^\mathrm{u}(\mathcal{M}^r_0)$ intersect in the full system along $\phi_\mathrm{f}(c)$. Since $\mathcal{M}^{\ell,+}_0$ lies in the plane $y=0$ in the region $x>0$ (and thus so do its fast fibers since the fast flow is confined to $y= $ const planes), we have that $\mathcal{W}^\mathrm{s}(\mathcal{M}^{\ell,+}_0)$ is tangent to the plane $y=0$ along $\phi_\mathrm{f}(c)$; equivalently $\mathcal{W}^\mathrm{s}(\mathcal{M}^{\ell,+}_0)$ is tangent to the plane $w=w_\ell$ in the original $(u,v,w)$-coordinates. In~\eqref{eq_layer}, from regular perturbation theory, the unstable manifold of the  equilibrium $p_3(w)$ (given by the trajectory $\phi_\mathrm{f}(c)$ at $w=w_\ell$) breaks smoothly in $w$ and thus $\mathcal{W}^\mathrm{u}(\mathcal{M}^r_0)$ is transverse to planes $w= $ const; in particular this gives the necessary transversality of $\mathcal{W}^\mathrm{s}(\mathcal{M}^{\ell,+}_0)$ and $\mathcal{W}^\mathrm{u}(\mathcal{M}^r_0)$.

We therefore obtain the desired transverse intersection of $\mathcal{W}^\mathrm{s}(\mathcal{M}^{\ell,+}_\eps)$ and $\mathcal{W}^\mathrm{u}(\mathcal{M}^r_\eps)$ $c>c^*(a)-\delta_c$ and each sufficiently small $\eps>0$. The geometry of this intersection for $\eps>0$ is depicted in Figure~\ref{fig:lower_fold_critical} for $c\approx c^*(a)$ and in Figure~\ref{fig:lower_fold_center} for $c>c^*(a)+\delta_c$.
\end{proof}

The analysis near the upper right fold point follows in a similar fashion. There it is possible to construct a singular trajectory $\mathcal{M}^{r,+}_0$ analogous to $\mathcal{M}^{\ell,+}_0$. This trajectory also possesses a stable manifold $\mathcal{W}^\mathrm{s}(\mathcal{M}^{r,+}_0)$ which perturbs as $\mathcal{O}(\eps^{2/3})$ in $C^0$ and as $\mathcal{O}(\eps^{1/3})$ in $C^1$. For $\eps=0$, this manifold analogously intersects $\mathcal{W}^\mathrm{u}(\mathcal{M}^\ell_0)$ along the heteroclinic trajectory $\phi_\mathrm{b}(c)$ for $c\geq c^*(a)$ in the plane $w=w_r$. We have the following, which is proved similarly to Proposition~\ref{prop_transfront}.
\begin{Proposition}
\label{prop_transback}
Fix $0<a<1/2$. There exists $\delta_c>0$ such that for each $c>c^*(a)-\delta_c$ and each sufficiently small $\eps>0$, the manifolds $\mathcal{W}^\mathrm{s}(\mathcal{M}^{r,+}_\eps)$ and $\mathcal{W}^\mathrm{u}(\mathcal{M}^\ell_\eps)$ intersect transversely.
\end{Proposition}
See Figure~\ref{fig:PO_nonhyp} for a depiction of the results of Propositions~\ref{prop_transfront} and ~\ref{prop_transback}.

\subsection{Construction of periodic orbits in nonhyperbolic regime}\label{sec_nhpoproof}
In this section, we provide a proof of Proposition~\ref{prop_nh_periodicorbits}.
\begin{proof}[Proof of Proposition~\ref{prop_nh_periodicorbits}]

For technical reasons, it is best to split the proof of Proposition~\ref{prop_nh_periodicorbits} into two cases: (i) $c\approx c^*(a)$ and (ii) $c> c^*(a)$, though the arguments in each case are similar. 

We begin with case (i) and fix $0<a<1/2$. For sufficiently small $\delta_c>0$, we define the interval $I_c = [c^*(a)-\delta_c, c^*(a)+\delta_c]$. When $c=c^*(a)$, the fronts $\phi_\mathrm{f}(c), \phi_\mathrm{b}(c)$ approach the folds along the unique strong stable direction. Our approach for constructing periodic orbits which follow $\Gamma_0(c)$ for $c\in I_c$ will be to place a two-dimensional section $\Sigma_z^\mathrm{fold}$ near the lower left fold point $p_\ell$ transverse to the strong stable eigendirection and consider the associated Poincar\'e map $\Pi :\Sigma_z^\mathrm{fold}\to \Sigma_z^\mathrm{fold}$; fixed points of this map correspond to periodic orbits.

We determine the location of $\mathcal{W}^\mathrm{u}(\mathcal{M}^r_\eps)$ in the neighborhood $V_\ell$. From Proposition~\ref{prop_transfront}, we know that $\mathcal{W}^\mathrm{u}(\mathcal{M}^r_0)$ intersects $\mathcal{W}^\mathrm{s}(\mathcal{M}^{\ell,+}_0)$ transversely for $\eps=0$ along the front $\phi_\mathrm{f}(c^*(a))$, and this intersection persists for $c\in I_c $ and sufficiently small $\eps>0$. 

 We define the exit section $\Sigma_z^\mathrm{fold}$ by
\begin{align}
\Sigma_z^\mathrm{fold}=\{z=\Delta_z\};
\end{align}
see Figure~\ref{fig:lower_fold_critical} for the setup. For $c \in I_c$ and sufficiently small $\eps>0$, the intersection of $\mathcal{W}^\mathrm{u}(\mathcal{M}^r_\eps)$  and $\mathcal{W}^\mathrm{s}(\mathcal{M}^{\ell,+}_\eps)$ occurs at a point 
\begin{align}
(x,y,z) = (x_\ell(c,\eps),y^\mathrm{s}_\eps(x_\ell(c,\eps);c), \Delta_z)\in \Sigma_z^\mathrm{fold},
\end{align}
and thus we may expand $\mathcal{W}^\mathrm{u}(\mathcal{M}^r_\eps)$ in $\Sigma_z^\mathrm{fold}$ as a graph $x=x^\mathrm{u}_\eps(y;c)$ where
\begin{align}
x^\mathrm{u}_\eps(y;c) = x_\ell(c,\eps)+\mathcal{O}\left(|y-y^\mathrm{s}_\eps(x_\ell(c,\eps);c)|\right), \quad  |y|\leq \Delta_y,
\end{align}
for some small $\Delta_y>0$. 

We now consider a small interval of initial conditions in $\Sigma_z^\mathrm{fold}$ which we will evolve \emph{backwards} in time until they return to the section $\Sigma_z^\mathrm{fold}$; that is, we consider the inverse map $\Pi_\mathrm{z}^{-1}:\Sigma_z^\mathrm{fold}\to\Sigma_z^\mathrm{fold}$. Since $\mathcal{W}^\mathrm{u}(\mathcal{M}^r_\eps)\cap \Sigma_z^\mathrm{fold}$ is given by a graph over $ |y|\leq \Delta_y$, we can define for each $ |\tilde{y}|\leq \Delta_y$ the curve $\mathcal{I}_{\tilde{y}} = \{(x,\tilde{y},\Delta_z):  |x|\leq \Delta_x\}\subset\Sigma_z^\mathrm{fold}$ for sufficiently small $\Delta_x>0$ which transversely intersects $\mathcal{W}^\mathrm{u}(\mathcal{M}^r_\eps)$ in $\Sigma_z^\mathrm{fold}$ at the point $(x^\mathrm{u}_\eps(\tilde{y};c), \tilde{y}, \Delta_z)$; see the left panel of Figure~\ref{fig:sigmafoldsetup}.

\begin{figure}
\hspace{.025\textwidth}
\begin{subfigure}{.45 \textwidth}
\centering
\includegraphics[width=1\linewidth]{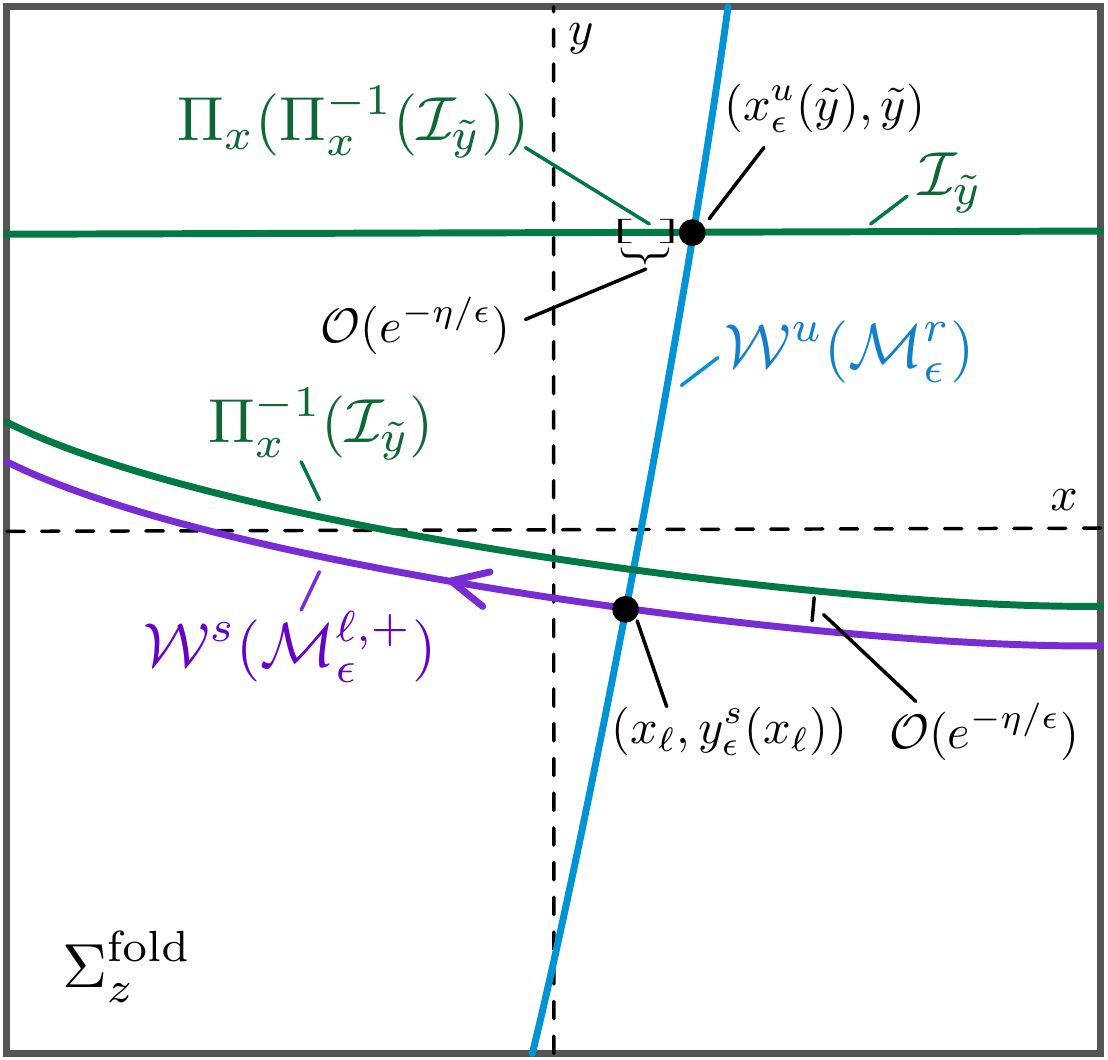}
\end{subfigure}
\hspace{.05\textwidth}
\begin{subfigure}{.45 \textwidth}
\centering
\includegraphics[width=1\linewidth]{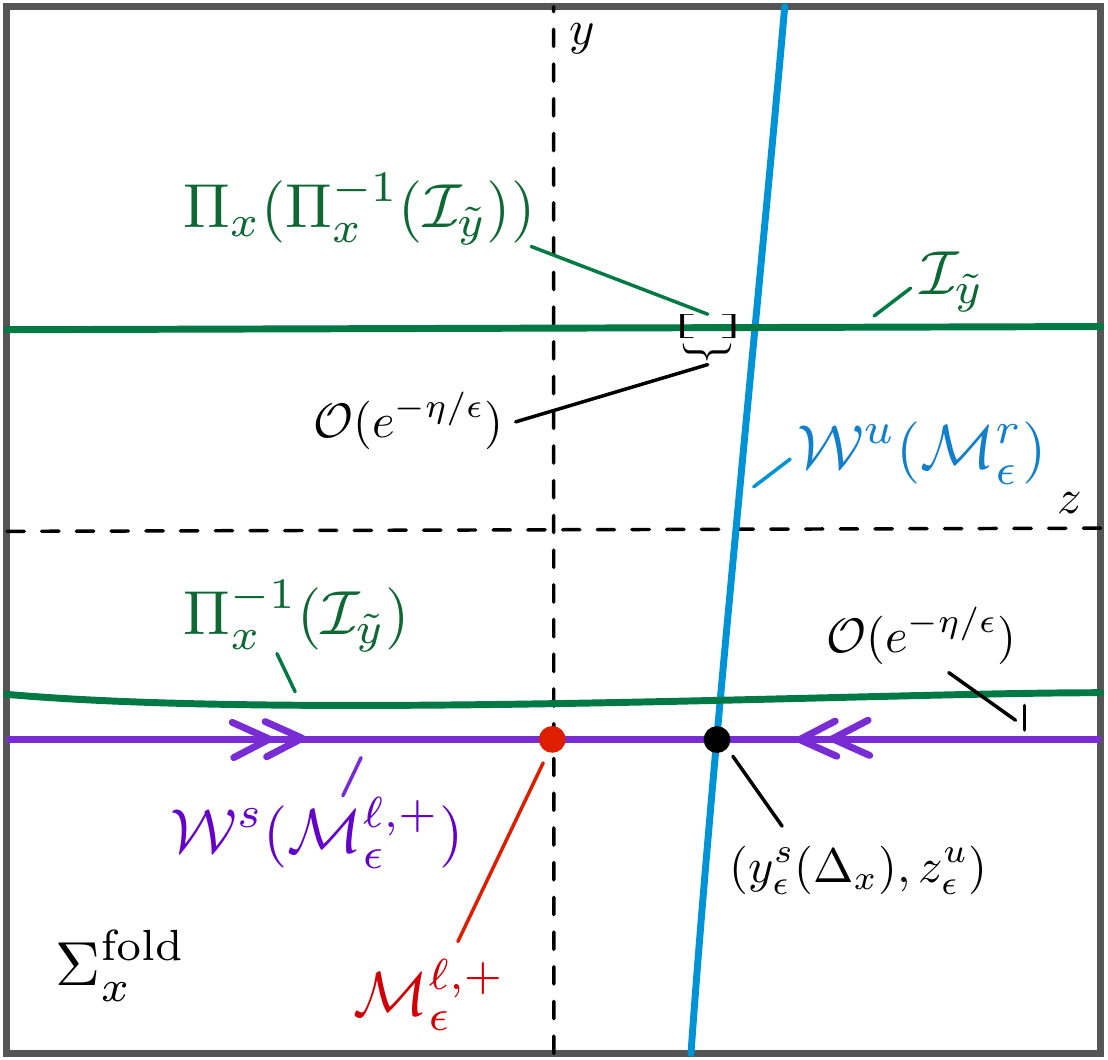}
\end{subfigure}
\hspace{.025\textwidth}
\caption{Shown is the setup for case (i) in the section $\Sigma^\textrm{fold}_z$ (left) and for case (ii) in the section $\Sigma^\textrm{fold}_x$ (right) in the proof of Proposition~\ref{prop_nh_periodicorbits}.}
\label{fig:sigmafoldsetup}
\end{figure}

Due to this transverse intersection, when evolving $\mathcal{I}_{\tilde{y}}$ backwards in time, by the exchange lemma~\cite{sec}, $\mathcal{I}_s$ traces out a two-dimensional manifold $\bar{\mathcal{I}}_{\tilde{y}}$ which aligns exponentially close to $\mathcal{W}^\mathrm{s}(\mathcal{M}^r_\eps)$ before transversely intersecting $\mathcal{W}^\mathrm{u}(\mathcal{M}^\ell_\eps)$ near the plane $w\approx w_r$, due to Proposition~\ref{prop_transback}. Due to this second transverse intersection, again by the exchange lemma, $\bar{\mathcal{I}}_{\tilde{y}}$ then aligns exponentially close to $\mathcal{W}^\mathrm{s}(\mathcal{M}^\ell_\eps)$ and hence arrives in the section $\Sigma_z^\mathrm{fold}$ aligned exponentially close to $\mathcal{W}^\mathrm{s}(\mathcal{M}^{\ell,+}_\eps)$; see Figure~\ref{fig:sigmafoldsetup} (left panel).

Therefore, we have that $\Pi_\mathrm{z}^{-1}(\mathcal{I}_{\tilde{y}})$ is given by a curve 
\begin{align}
(x,y) = \left(\tilde{x},y^\mathrm{s}_\eps(\tilde{x};c)+h(\tilde{x}, \tilde{y}, c, \eps)\right)
\end{align}
for $|\tilde{x}|\leq \Delta_x$, where $h(\tilde{x}, \tilde{y}, c, \eps) = \mathcal{O}(e^{-\eta/\eps})$ uniformly in $(\tilde{x}, \tilde{y}, c)$.

We now consider $\Pi_\mathrm{z}$ applied to $\Pi_\mathrm{z}^{-1}(\mathcal{I}_{\tilde{y}})$. Of course $\Pi_\mathrm{z}(\Pi_\mathrm{z}^{-1}(\mathcal{I}_{\tilde{y}}))\subset \mathcal{I}_{\tilde{y}}$, but the image is contained in an exponentially thin interval within $\mathcal{I}_{\tilde{y}}$. In particular for a point $\left(\tilde{x},y^\mathrm{s}_\eps(\tilde{x};c)+h(\tilde{x}, \tilde{y}, c, \eps)\right) \in \Pi_\mathrm{z}^{-1}(\mathcal{I}_{\tilde{y}})$ with $|\tilde{x}|\leq \Delta_x$, the image under $\Pi_\mathrm{z}$ is given by
\begin{align}
\Pi_\mathrm{z}\left(\tilde{x},y^\mathrm{s}_\eps(\tilde{x};c)+h(\tilde{x}, \tilde{y}, c, \eps)\right) = \left(x^\mathrm{u}_\eps(\tilde{y};c)+\mathcal{O}(e^{-\eta/\eps}), \tilde{y}\right),
\end{align}
where the exponentially small errors are uniform in $\tilde{x}, \tilde{y}, c$ and the derivatives with respect to these variables are also exponentially small. To find a fixed point, we set the argument equal to the right hand side and obtain
\begin{align}\begin{split}
\tilde{x} &= x^\mathrm{u}_\eps(\tilde{y};c)+\mathcal{O}(e^{-\eta/\eps})\\
\tilde{y} &= y^\mathrm{s}_\eps(\tilde{x};c)+\mathcal{O}(e^{-\eta/\eps}).
\end{split}
\end{align}

Hence we search for zeros of $\mathcal{F}(\tilde{x},\tilde{y};c,\eps)$ where
\begin{align}\begin{split}
\mathcal{F}(\tilde{x},\tilde{y};c,\eps):=\begin{pmatrix}\tilde{x} - x^\mathrm{u}_\eps(\tilde{y};c)+\mathcal{O}(e^{-\eta/\eps})\\
\tilde{y} - y^\mathrm{s}_\eps(\tilde{x};c)+\mathcal{O}(e^{-\eta/\eps}) \end{pmatrix}.
\end{split}
\end{align}

We have that
\begin{align}\begin{split}
D_{(\tilde{x}, \tilde{y})}\mathcal{F}(0,0;c^*(a),0):=\begin{pmatrix} 1 & K \\
0 & 1\end{pmatrix},
\end{split}
\end{align}
for some $K$ independent of $\eps$, and so by the implicit function theorem, for $c\in I_c$ and sufficiently small $\eps>0$, we can solve for a solution which occurs when
\begin{align}\begin{split}
\begin{pmatrix} \tilde{x} \\ \tilde{y}\end{pmatrix}&= \begin{pmatrix} \tilde{x}_\mathrm{p}(c,\eps) \\ \tilde{y}_\mathrm{p}(c,\eps)\end{pmatrix} :=\begin{pmatrix} x_\ell(c,\eps) \\ y^\mathrm{s}_\eps( x_\ell(c,\eps);c)\end{pmatrix}+\mathcal{O}(e^{-\eta/\eps}),
\end{split}
\end{align}
which corresponds to a periodic orbit.

We now turn to the case (ii) for which it remains to consider values $c>c^*(a)+\delta_c$. The argument is similar to case (i), and hence we only outline the differences. For small $\Delta_x>0$, we define the section
\begin{align}
\Sigma_x^\mathrm{fold} = \{x=\Delta_x\};
\end{align}
see Figure~\ref{fig:lower_fold_center} for the setup.

Since $\phi_\mathrm{f}(c)$ approaches the fold along a center direction for $c>c^*(a)$, $\phi_\mathrm{f}(c)$ is attracted to the local center manifold $z=0$ in forward time. Hence for each $\Delta_z>0$, by possibly shrinking $\Delta_x$ if necessary, it is possible to guarantee that $\phi_\mathrm{f}(c)$ intersects the section $\Sigma_x^\mathrm{fold}$ in the plane $y=0$ at a point $(\Delta_x, 0, z_\mathrm{f}(c))$ where $0\leq z_\mathrm{f}(c) <\Delta_z/2$. Therefore in the section $\Sigma_x^\mathrm{fold}$, the manifold $\mathcal{W}^\mathrm{u}(\mathcal{M}^r_\eps)$ can be represented as a graph 
\begin{align}
z = z_\mathrm{f}(c)+\mathcal{O}(y, \eps), \quad |y|\leq \Delta_y.
\end{align}
For each $|\tilde{y}|\leq \Delta_y$, we now define the curve $\mathcal{I}_{\tilde{y}} = \{(\Delta_x,\tilde{y},z):  |z|\leq \Delta_z\}\subset\Sigma_x^\mathrm{fold}$ which transversely intersects $\mathcal{W}^\mathrm{u}(\mathcal{M}^r_\eps)$ provided $\Delta_y$ is sufficiently small.

Since $\mathcal{W}^\mathrm{s}(\mathcal{M}^{\ell,+}_\eps)$ is $\mathcal{O}(\eps^{2/3})$ close in $C^0$ to $\mathcal{W}^\mathrm{s}(\mathcal{M}^{\ell,+}_0)$, in the section $\Sigma_x^\mathrm{fold}$, the manifold $\mathcal{W}^\mathrm{s}(\mathcal{M}^{\ell,+}_\eps)$ lies in the plane $y = y^\mathrm{s}_\eps(\Delta_x;c) = \mathcal{O}(\eps^{2/3})$ for $|z|\leq \Delta_z$, and therefore intersects $\mathcal{W}^\mathrm{u}(\mathcal{M}^r_\eps)$ transversely at the point $(y,z) = ( y^\mathrm{s}_\eps(\Delta_x;c), z^\mathrm{u}_\eps(c))$ where $z^\mathrm{u}_\eps(c) = z_\mathrm{f}(c)+\mathcal{O}(\eps^{2/3})$; see Figure~\ref{fig:sigmafoldsetup} (right panel).

 Proceeding as before, we consider the inverse map $\Pi_x^{-1}:\Sigma_x^\mathrm{fold}\to \Sigma_x^\mathrm{fold}$, under which $\Pi_x^{-1}(\mathcal{I}_{\tilde{y}})$ returns to $\Sigma_x^\mathrm{fold}$ exponentially aligned with $\mathcal{W}^\mathrm{s}(\mathcal{M}^{\ell,+}_\eps)$. We again consider the image $\Pi_x(\Pi_x^{-1}(\mathcal{I}_{\tilde{y}}))$, and proceeding similarly as before, we find a fixed point corresponding to a periodic orbit at
 \begin{align}\begin{split}
\begin{pmatrix} \tilde{y} \\ \tilde{z}\end{pmatrix}&= \begin{pmatrix} \tilde{y}_\mathrm{p}(c,\eps) \\ \tilde{z}_\mathrm{p}(c,\eps)\end{pmatrix} :=\begin{pmatrix} 0 \\ z_\mathrm{f}(c)\end{pmatrix}+\mathcal{O}(\eps^{2/3}).
\end{split}
\end{align}
 \end{proof}

\subsection{Stable and unstable manifolds of the periodic orbits $\Gamma_\eps(c)$}\label{sec_poinvariantmanifolds}

In order to find a heteroclinic which connects $\Gamma_\eps(c)$ and $p$, we first aim at understanding $\mathcal{W}^\mathrm{u}(\Gamma_\eps(c))$. The periodic orbit $\Gamma_\eps(c)$ is found by perturbing from a singular structure $\Gamma_0(c)$ which consists of concatenated slow and fast segments, where the slow segments are portions of the branches $\mathcal{M}^\ell_0, \mathcal{M}^r_0$ of the critical manifold. When $0<\eps\ll 1$, we aim to show that along portions of $\Gamma_\eps(c)$ which are nearby $\mathcal{M}^\ell_\eps,  \mathcal{M}^r_\eps$, the unstable manifold $\mathcal{W}^\mathrm{u}(\Gamma_\eps(c))$ is, in an appropriate sense, close to $\mathcal{W}^\mathrm{u}(\mathcal{M}^\ell_\eps), \mathcal{W}^\mathrm{u}(\mathcal{M}^r_\eps)$, respectively.

Consider the periodic orbit $\Gamma_\eps(c)$ for $c>0$, and recall that $w_\mathrm{f}(c),w_\mathrm{b}(c)\in [w_\ell, w_r]$ denote the jump-off heights along the associated singular periodic orbit $\Gamma_0(c)$. (Note that $w_\mathrm{f}(c)=w_\ell$ and $w_\mathrm{f}(c)=w_r$ for all $c\geq c^*(a)$.)

\begin{Proposition}\label{prop_poinvmflds}
Fix $0<a<1/2$ and $c>0$, and let $w_m, w_M$ satisfy $w_\mathrm{f}(c)<w_m<w_M<w_\mathrm{b}(c)$. Then there exists $\Delta, \eta>0$ such that for all $\eps>0$ sufficiently small, the following holds. Consider the Fenichel neighborhoods
\begin{align}
\begin{split}
U^\ell&:=\left\{q=(u,v,w): d(q,\mathcal{M}^\ell_\eps)<\Delta, w_m<w<w_M\right\}\\
U^r&:=\left\{q=(u,v,w): d(q,\mathcal{M}^r_\eps)<\Delta, w_m<w<w_M\right\}.
\end{split}
\end{align}
Within $U^\ell$ the periodic orbit $\Gamma_\eps(c)$ is $\mathcal{O}(e^{-\eta/\eps})$-close to $\mathcal{M}^\ell_\eps$ in the $C^1$-topology, and the manifolds $\mathcal{W}^\mathrm{u}(\Gamma_\eps(c)), \mathcal{W}^\mathrm{s}(\Gamma_\eps(c))$ are $\mathcal{O}(e^{-\eta/\eps})$-close in the  $C^1$-topology to $\mathcal{W}^\mathrm{u}(\mathcal{M}^\ell_\eps), \mathcal{W}^\mathrm{s}(\mathcal{M}^\ell_\eps)$, respectively.

Within $U^r$, the same statements hold with respect to the manifolds $\mathcal{M}^r_\eps$ and $\mathcal{W}^\mathrm{u}(\mathcal{M}^r_\eps), \mathcal{W}^\mathrm{s}(\mathcal{M}^r_\eps)$.
\end{Proposition}
\begin{proof}
We focus on the hyperbolic case $0<c<c^*(a)$; the nonhyperbolic case $c\geq c^*(a)$ is similar. We consider a Poincar\'e section $\Sigma_\mathrm{f}$ transverse to the periodic orbit along the front $\phi_\mathrm{f}(c)$. We study the Poincar\'e map $\Pi_p:\Sigma_\mathrm{f}\to \Sigma_\mathrm{f}$ to determine the structure of the manifolds $\mathcal{W}^\mathrm{u}(\Gamma_\eps(c)), \mathcal{W}^\mathrm{s}(\Gamma_\eps(c))$. We can analogously place a section $\Sigma_\mathrm{b}$ transverse to the periodic orbit along the front $\phi_\mathrm{b}(c)$; see Figure~\ref{fig:PO_hyp}. Then $\Pi_p$ is the composition of two maps $\Pi_p = \Pi_{bf}\circ \Pi_{fb}$ where $\Pi_{fb}:\Sigma_\mathrm{f} \to \Sigma_\mathrm{b}$ and $\Pi_{bf}:\Sigma_\mathrm{b} \to \Sigma_\mathrm{f}$ denote the transition maps between the sections $\Sigma_\mathrm{f}, \Sigma_\mathrm{b}$.

Within the section $\Sigma_\mathrm{f}$, the manifolds $\mathcal{W}^\mathrm{u}(\mathcal{M}^r_\eps)$ and $\mathcal{W}^\mathrm{s}(\mathcal{M}^\ell_\eps)$ intersect transversely at a point in which is $\mathcal{O}(e^{-\eta/\eps})$-close to $\Gamma_\eps(c)$. We choose local coordinates $\{(X,Y):|X|,|Y|\leq \Delta\}$ in the section $\Sigma_\mathrm{f}$ so that $\mathcal{W}^\mathrm{u}(\mathcal{M}^r_\eps)\cap \Sigma_\mathrm{f} = \{X=0\}$ and $\mathcal{W}^\mathrm{s}(\mathcal{M}^\ell_\eps)\cap \Sigma_\mathrm{f} = \{Y=0\}$. In these coordinates, $\Gamma_\eps(c)\cap \Sigma_\mathrm{f}$ is given as a point $(X,Y)=\mathcal{O}(e^{-\eta/\eps})$. We now shift these coordinates so that $\Gamma_\eps(c)\cap \Sigma_\mathrm{f}=(0,0)$ and thus $\mathcal{W}^\mathrm{u}(\mathcal{M}^r_\eps)$ and $\mathcal{W}^\mathrm{s}(\mathcal{M}^\ell_\eps)$ are given by curves $X=\mathcal{O}(e^{-\eta/\eps})$ and $Y=\mathcal{O}(e^{-\eta/\eps})$, respectively; see Figure~\ref{fig:sigmafb} (left panel).

\begin{figure}
\hspace{.025\textwidth}
\begin{subfigure}{.45 \textwidth}
\centering
\includegraphics[width=1\linewidth]{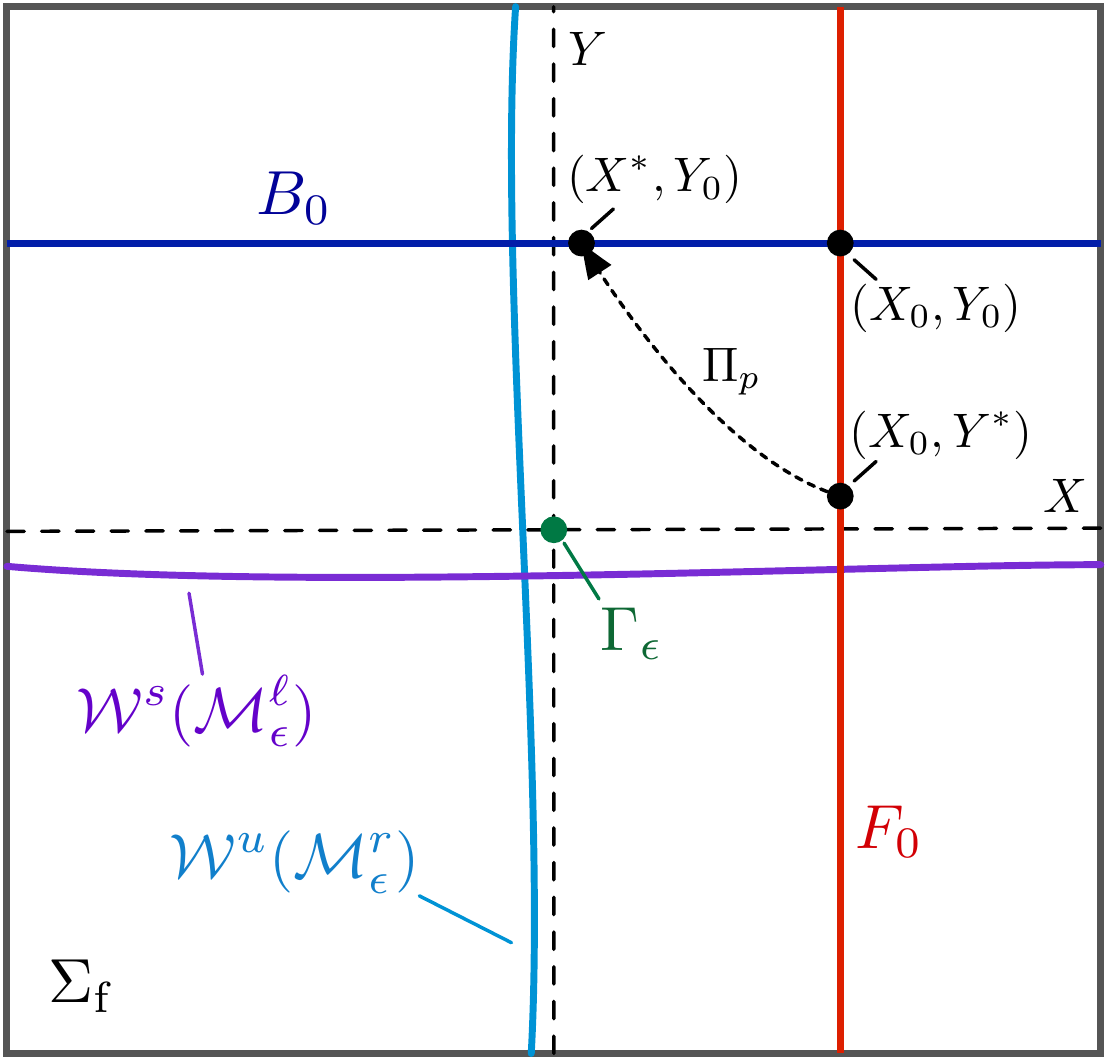}
\end{subfigure}
\hspace{.05\textwidth}
\begin{subfigure}{.45 \textwidth}
\centering
\includegraphics[width=1\linewidth]{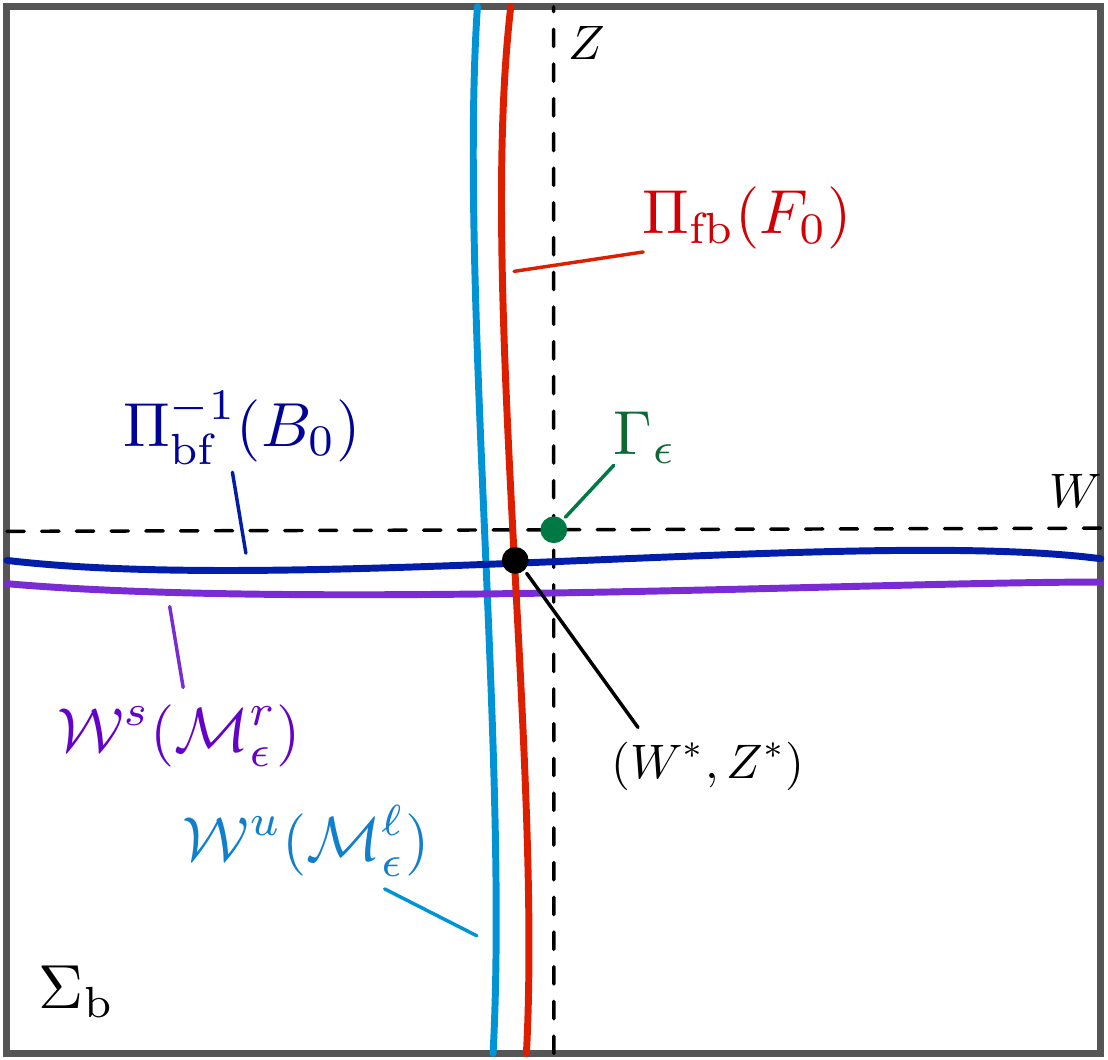}
\end{subfigure}
\hspace{.025\textwidth}
\caption{Shown is the setup in the section $\Sigma_\mathrm{f}$ (left) and $\Sigma_\mathrm{b}$ (right) in the proof of Proposition~\ref{prop_poinvmflds}.}
\label{fig:sigmafb}
\end{figure}

Within the section $\Sigma_\mathrm{b}$, we proceed similarly and choose coordinates $\{(W,Z):|W|,|Z|\leq \Delta\}$ so that $\Gamma_\eps(c)\cap \Sigma_\mathrm{f}=(0,0)$ and the manifolds $\mathcal{W}^\mathrm{u}(\mathcal{M}^\ell_\eps)$ and $\mathcal{W}^\mathrm{s}(\mathcal{M}^r_\eps)$ are given by curves $W=\mathcal{O}(e^{-\eta/\eps})$ and $Z=\mathcal{O}(e^{-\eta/\eps})$, respectively; see Figure~\ref{fig:sigmafb} (right panel).

We wish to obtain estimates on the map $\Pi_p:(X,Y)\mapsto \Pi_p(X,Y)$. We consider a line $B_0:=\{(X,Y_0):|X|\leq\Delta\}\subset \Sigma_\mathrm{f}$. We now consider the image $\Pi^{-1}_{bf}(B_0)$ in the section $\Sigma_\mathrm{b}$. Since $B_0$ transversely intersects $\mathcal{W}^\mathrm{u}(\mathcal{M}^r_\eps)$, in backwards time, by the exchange lemma $\Pi^{-1}_{bf}(B_0)$ is aligned exponentially close to $\mathcal{W}^\mathrm{s}(\mathcal{M}^r_\eps)$ which is given by a curve $Z=\mathcal{O}(e^{-\eta/\eps})$. We next consider a line $F_0:=\{(X_0,Y):|Y|\leq\Delta\}\subset \Sigma_\mathrm{f}$,and we consider the image $\Pi_{fb}(F_0)$ in the section $\Sigma_\mathrm{b}$. Again by the exchange lemma, since $F_0$ transversely intersects $\mathcal{W}^\mathrm{s}(\mathcal{M}^\ell_\eps)$, in $\Sigma_\mathrm{b}$ $\Pi_{fb}(F_0)$ is aligned exponentially close to $\mathcal{W}^\mathrm{u}(\mathcal{M}^\ell_\eps)$. 

Therefore the two curves $\Pi^{-1}_{bf}(B_0)$ and $\Pi_{fb}(F_0)$ intersect transversely at a point $(W,Z)=(W^*,Z^*) = \mathcal{O}(e^{-\eta/\eps})$. We now reverse the maps and, again by the exchange lemma, deduce that $ \Pi_{bf}(W^*,Z^*) = (X^*,Y_0)$ where $X^*= \mathcal{O}(e^{-\eta/\eps})$, and similarly $ \Pi^{-1}_{fb}(W^*,Z^*) = (X_0,Y^*)$ where $Y^*= \mathcal{O}(e^{-\eta/\eps})$; see Figure~\ref{fig:sigmafb} (left panel).

Therefore, given $(X_0,Y_0)$, there exists $X^*,Y^*$ such that $\Pi_p(X_0,Y^*) = (X^*,Y_0)$, and $X^*,Y^*$ and their derivatives with respect to $X_0,Y_0$ are $\mathcal{O}(e^{-\eta/\eps})$; also, by construction, the point $(X,Y)=(0,0)$ is a fixed point of $\Pi_p$. On the domain of definition of $\Pi$, it follows that the derivative $D\Pi_p$ is uniformly expanding on the cone $|X|\leq \mu |Y|$, and similarly $D\Pi_p^{-1}$ is expanding on the cone $|Y|\leq \mu |X|$ for any sufficiently small fixed $\mu>0$. 

Therefore, the splitting of the tangent space at the fixed point $(X,Y)=(0,0)$ along the $X$ and $Y$ coordinate directions is an almost hyperbolic splitting in the sense of~\cite[\S3]{newhousepalis}. Hence by~\cite[Theorem 3.1]{newhousepalis}, the fixed point $(X,Y)=(0,0)$ corresponding to $\Gamma_\eps(c)$ is a hyperbolic set and therefore has invariant stable and unstable manifolds $\mathcal{W}^\mathrm{u}(\Gamma_\eps(c)), \mathcal{W}^\mathrm{s}(\Gamma_\eps(c))$ given as graphs over the subspaces $X=0$ and $Y=0$, respectively~\cite{hirsch1970stable}. By evolving $\mathcal{W}^\mathrm{u}(\Gamma_\eps(c)), \mathcal{W}^\mathrm{s}(\Gamma_\eps(c))$ forward (resp. backward) under the flow of~\eqref{eq_twode} and again using the exchange lemma, it is clear that within the Fenichel neighborhoods $U_\ell, U_r$, the stable/unstable manifolds $\mathcal{W}^\mathrm{u}(\Gamma_\eps(c)), \mathcal{W}^\mathrm{s}(\Gamma_\eps(c))$ must align $C^1$ $\mathcal{O}(e^{-\eta/\eps})$-close to the stable/unstable manifolds of the slow manifolds $\mathcal{M}^\ell_\eps, \mathcal{M}^r_\eps$ as claimed.

For the nonhyperbolic case, the same argument can be applied by replacing the manifold $\mathcal{W}^\mathrm{s}(\mathcal{M}^\ell_\eps)$ with $\mathcal{W}^\mathrm{s}(\mathcal{M}^{\ell,+}_\eps)$ in the section $\Sigma_\mathrm{f}$, and by similarly replacing $\mathcal{W}^\mathrm{s}(\mathcal{M}^r_\eps)$ with $\mathcal{W}^\mathrm{s}(\mathcal{M}^{r,+}_\eps)$ in $\Sigma_\mathrm{b}$; see Figure~\ref{fig:PO_nonhyp}.

%
%
%
\end{proof}

\subsection{Proof of Theorem~\ref{thm_periodicexistence}}\label{sec_periodicorbitthmproof}
In this section, we briefly conclude the proof of Theorem~\ref{thm_periodicexistence}.
\begin{proof}[Proof of Theorem~\ref{thm_periodicexistence}]
The existence of a continuous family of periodic orbits in the hyperbolic/nonhyperbolic regimes follows from known results in~\cite{STR}, combined with Proposition~\ref{prop_nh_periodicorbits} above; see~\S\ref{sec_periodicorbits}. It remains to consider the final statements regarding the period $L(c;\eps)$ of the solutions.

First, we comment on the monotonicity of the period. We claim that the derivative $\partial_cL(c;\epsilon)$ has fixed sign. Suppose $\partial_cL(c;\epsilon)=0$ for some value of $c$. Let $\phi_\mathrm{per}(\xi;c,\epsilon)=(u_\mathrm{per},w_\mathrm{per})(\xi;c,\epsilon)$ denote the family of traveling periodic orbits. Each element of this family is a stationary solution to the PDE~\eqref{eq_pde} in the appropriate comoving frame $\xi=x-ct$
\begin{align}\begin{split}
D\partial_{\xi\xi}\phi_\mathrm{per} + c\partial_\xi \phi_\mathrm{per} + F(\phi_\mathrm{per}) = 0,
\end{split}
\end{align}
where
\begin{align}\begin{split}
D = \begin{pmatrix} 1 & 0\\0& 0 \end{pmatrix}, \quad F(u,w) = \begin{pmatrix}  f(u)-w\\\epsilon(u-\gamma w-a) \end{pmatrix}
\end{split}
\end{align}
along with periodic boundary conditions $\phi_\mathrm{per}(\xi+L;c,\epsilon) = \phi_\mathrm{per}(\xi;c,\epsilon)$, for all $\xi \in \mathbb{R}$. Equivalently, by rescaling the traveling wave variable $\xi = L\theta$, $\phi_\mathrm{per}$ satisfies the equation
\begin{align}\begin{split}\label{eq_rescaledperiodicpde}
\frac{4\pi^2}{L^2}D\partial_{\theta\theta}\phi_\mathrm{per} + \frac{2\pi}{L}c\partial_\theta \phi_\mathrm{per} + F(\phi_\mathrm{per}) = 0
\end{split}
\end{align}
with the fixed periodic boundary conditions $\phi_\mathrm{per}(\theta+2\pi;c,\epsilon) = \phi_\mathrm{per}(\theta;c,\epsilon)$, for all $\theta \in \mathbb{R}$. The linearization of~\eqref{eq_rescaledperiodicpde} about $\phi_\mathrm{per}$ is given by the operator
\begin{align}\begin{split}
\mathcal{L}_\mathrm{per}:=\frac{4\pi^2}{L^2}D\partial_{\theta\theta} + \frac{2\pi}{L}c\partial_\theta + F'(\phi_\mathrm{per})
\end{split}
\end{align}
and due to translation invariance $\partial_\theta\phi_\mathrm{per}$ lies in the kernel of $\mathcal{L}_\mathrm{per}$.

By differentiating~\eqref{eq_rescaledperiodicpde} with respect to $c$, if $\partial_cL(c;\epsilon)=0$, then the derivative $\partial_c\phi_\mathrm{per}$ satisfies the equation
\begin{align}\begin{split}
\frac{4\pi^2}{L^2}D\partial_{\theta\theta}(\partial_c\phi_\mathrm{per}) + \frac{2\pi}{L}c\partial_\theta (\partial_c\phi_\mathrm{per})+ F'(\phi_\mathrm{per})\partial_c\phi_\mathrm{per} = -\frac{2\pi}{L}\partial_\theta \phi_\mathrm{per},
\end{split}
\end{align}
that is, $\partial_c\phi_\mathrm{per}$ lies in the generalized kernel of $\mathcal{L}_\mathrm{per}$. From this we deduce that $\phi_\mathrm{per}$ has a Floquet multiplier $1$ of algebraic multiplicity $2$, which contradicts the hyperbolicity of $\phi_\mathrm{per}$ obtained in~\S\ref{sec_poinvariantmanifolds}. We deduce that the period $L(c;\epsilon)$ is strictly monotone in $c$; the fact that $L$ is increasing follows from the expansions below.

To obtain the asymptotics of $L(c;\eps)$ in $\eps$, we will use proximity of the periodic orbits $\Gamma_\eps(c)$ to the singular solutions $\Gamma_0(c)$. We first consider the hyperbolic regime and fix $0<c<c^*(a)$. Then $\Gamma_\eps(c)$ follows an orbit in phase space which remains $\mathcal{O}(\eps)$-close to the singular orbit $\Gamma_0(c)$ which traverses the front $\phi_\mathrm{f}$, the piece of the slow manifold $\mathcal{M}^\ell_0$ between $w=w_\mathrm{f}$ and $w=w_\mathrm{b}$, the front $\phi_\mathrm{b}$, and finally the piece of the slow manifold $\mathcal{M}^r_0$ from $w=w_\mathrm{b}$ and $w=w_\mathrm{f}$. We proceed by estimating the time spent along each portion. 

Along the slow manifold $\mathcal{M}^\ell_0$, $\Gamma_\eps(c)$ is $\mathcal{O}(\eps)$-close to the curve $w=f(u)$ and the $w$-dynamics are given by
\begin{align}
\dot{w} &= \frac{\eps}{c}(u-\gamma w-a)\\
&=\frac{\eps}{c}(f^{-1}(w)-\gamma w-a+\mathcal{O}(\eps)),
\end{align}
where $f^{-1}(w)$ is interpreted as the smallest of the three roots of $f(u)=w$. We determine the time spent along this portion as
\begin{align}\begin{split}
L^\ell_\eps(c) &= \int_{w_\mathrm{f}}^{w_\mathrm{b}}\frac{c}{\eps(f^{-1}(w)-\gamma w-a+\mathcal{O}(\eps))}dw\\
&=\frac{c}{\eps}\int_{w_\mathrm{f}}^{w_\mathrm{b}}\frac{1+\mathcal{O}(\eps))}{(f^{-1}(w)-\gamma w-a)}dw.
\end{split}
\end{align}
We define
\begin{align}\begin{split}
L^\ell_0(c)&=c\int_{w_\mathrm{f}}^{w_\mathrm{b}}\frac{1}{(f^{-1}(w)-\gamma w-a)}dw
\end{split}
\end{align}
and note that $L^\ell_0(c)>0$ and $L^\ell_0(c)$ is an increasing function of $c$ due to the fact that $w_\mathrm{f}$ and $w_\mathrm{b}$ decrease and increase, respectively, as $c$ increases. Further, we have that
\begin{align}\begin{split}
\epsilon L^\ell_\eps(c) &= L^\ell_0(c)+\mathcal{O}(\epsilon).
\end{split}
\end{align}
Along $\mathcal{M}^r_0$, we can proceed similarly and define
\begin{align}\begin{split}
L^r_0(c)&=c\int_{w_\mathrm{b}}^{w_\mathrm{f}}\frac{1}{(f^{-1}(w)-\gamma w-a)}dw,
\end{split}
\end{align}
where now $f^{-1}(w)$ refers to the largest of the three roots of $u=f(w)$. Again we have that $L^r_0(c)>0$ and $L^r_0(c)$ is an increasing function of $c$. Finally we obtain that the time spent near $\mathcal{M}^r_0$ is given by $L^r_\eps(c)$ where
\begin{align}\begin{split}
\epsilon L^r_\eps(c) &= L^r_0(c)+\mathcal{O}(\epsilon).
\end{split}
\end{align}

Finally, the full period $L(c,\eps)$ is obtained as the sum $L^\ell_\eps(c)+L^r_\eps(c)$ spent within a $\mathcal{O}(\eps)$ neighborhood of the slow manifolds, plus the jump time spent traveling between the slow manifolds along the fronts $\phi_\mathrm{f}$ and $\phi_\mathrm{b}$. To estimate these jump times, we note that each jump consists of a finite time segment between small Fenichel neighborhoods of $\mathcal{M}^\ell_0$ and $\mathcal{M}^r_0$, as well as transitions from the boundaries of the Fenichel neighborhoods to $\mathcal{O}(\epsilon)$ neighborhoods of $\mathcal{M}^\ell_0$ and $\mathcal{M}^r_0$. These latter transition times can be estimated as $\mathcal{O}(\log \eps)$ using corner-type estimates (see, for instance,~\cite[Theorem~4.5]{carter2016stability}). Hence we obtain
\begin{align}\begin{split}
L(c,\epsilon) &= L^\ell_\eps(c)+L^r_\eps(c)+\mathcal{O}(\log \eps)\\
&=\eps^{-1}L_0(c)+\mathcal{O}(\log \eps)
\end{split}
\end{align}
where $L_0(c)= L^\ell_0(c)+L^r_0(c)$.

We now consider the nonhyperbolic regime and fix $c>c^*(a)$. We proceed similarly by considering the flow along the slow manifolds, noting that now $\Gamma_\eps(c)$ traverses near the folds along $\mathcal{M}^{\ell,+}_\eps$ and $\mathcal{M}^{r,+}_\eps$. By the $C^1-\mathcal{O}(\eps^{1/3})$-proximity of the manifolds $\mathcal{M}^{\ell,+}_\eps$ and $\mathcal{M}^{r,+}_\eps$ to their singular $\epsilon=0$ counterparts, and the fact that $\Gamma_\eps(c)$ remains within $\mathcal{O}(\eps^{2/3})$ of the singular solution $\Gamma_0(c)$, we can estimate the time spent along each of  $\mathcal{M}^{\ell,+}_\eps$ and $\mathcal{M}^{r,+}_\eps$ between $w=w_\ell$ and $w=w_r$. Near $\mathcal{M}^{\ell,+}_\eps$, the time spent is approximated by
\begin{align}\begin{split}
L^\ell_\eps(c) &= \int_{w_\ell}^{w_r}\frac{c+\mathcal{O}(\eps^{2/3})}{\eps(f^{-1}(w)-\gamma w-a+\mathcal{O}(\eps^{1/3}))}dw\\
&=\frac{c}{\eps}\int_{w_\ell}^{w_r}\frac{1+\mathcal{O}(\eps^{1/3})}{(f^{-1}(w)-\gamma w-a)}dw.
\end{split}
\end{align}
By defining
\begin{align}\begin{split}
L^\ell_0(c)&=c\int_{w_\ell}^{w_r}\frac{1}{(f^{-1}(w)-\gamma w-a)}dw,
\end{split}
\end{align}
we note $L^\ell_0(c)>0$ and $L^\ell_0(c)$ is an increasing function of $c$, and we have that 
\begin{align}\begin{split}
\epsilon L^\ell_\eps(c) &= L^\ell_0(c)+\mathcal{O}(\epsilon^{1/3}).
\end{split}
\end{align}
We define $L^r_\eps(c), L^r_0(c)$ analogously and set $L_0(c) = L^\ell_0(c)+L^r_0(c)$; by similar arguments as in the hyperbolic case above, we obtain
\begin{align}\begin{split}
L(c,\epsilon) &= L^\ell_\eps(c)+L^r_\eps(c)+\mathcal{O}(\log \eps)\\
&=\eps^{-1}L_0(c)+\mathcal{O}(\epsilon^{-2/3}).
\end{split}
\end{align}

\end{proof}

\subsection{The pulled case: proof of Theorem~\ref{thm_pulledexistence}}\label{sec_pulledproof}
In the pulled front case, we aim to find an intersection of the manifolds $\mathcal{W}^\mathrm{u}(\Gamma_\eps(c))$ and $\mathcal{W}^\mathrm{s}(p)$ coinciding with a double root in the linearization about the equilibrium $p$. From the layer analysis in~\S\ref{sec_layer}, for $\eps=0$ this double root occurs at $c=c_\mathrm{lin}(a)$. In the full three-dimensional system for $\eps=0$, there is an additional zero eigenvalue present due to the third equation $\dot{w}=0$. This eigenvalue perturbs and becomes negative for $0<\eps \ll1$, so that the stable manifold $\mathcal{W}^\mathrm{s}(p)$ is three dimensional; see also~\S\ref{sec_persistenceinvtmflds}.

Hence for $\eps=0$, we define the strong stable manifold $\mathcal{W}^\mathrm{ss}(p)$ to be that which corresponds to the two strictly negative eigenvalues; this manifold lies in the plane $\{w=0\}$. For $0<\eps \ll 1$, this manifold perturbs to a two-dimensional manifold which is $\mathcal{O}(\eps)$-close to its $\eps=0$ counterpart. Further the eigenvalues perturb and may split; however, by solving for $c=c_\mathrm{lin}(a)+\mathcal{O}(\eps)=:c_\mathrm{lin}(a,\eps)$, we can ensure that the double eigenvalue persists for $0<\eps \ll 1$.

We therefore have the following.

\begin{proof}[Proof of Theorem~\ref{thm_pulledexistence}]
The manifolds $\mathcal{W}^\mathrm{ss}(p)$ and $\mathcal{W}^\mathrm{u}(\mathcal{M}^r_0)$ intersect transversely for $\eps=0$. This transverse intersection persists for $0<\eps\ll 1$, with the double eigenvalue occurring for $c=c_\mathrm{lin}(a,\eps)$. By Proposition~\ref{prop_poinvmflds}, the unstable manifold $\mathcal{W}^\mathrm{u}(\Gamma(c_\mathrm{p}(a,\eps)))$ is $\mathcal{O}(e^{-\eta/\eps})$-close to $\mathcal{W}^\mathrm{u}(\mathcal{M}^r_\eps)$ and hence also transversely intersects $\mathcal{W}^\mathrm{ss}(p)$ along a solution orbit $F^r_\eps(a)$.

A similar argument can be applied to the manifolds $\mathcal{W}^\mathrm{u}(\mathcal{M}^\ell_\eps)$ and $\mathcal{W}^\mathrm{ss}(p)$, which results in a second distinct heteroclinic solution $F^\ell_\eps(a)$.
\end{proof}

\subsection{The pushed case: proof of Theorem~\ref{thm_pushedexistence}}\label{sec_pushedproof}
We recall from~\S\ref{sec_layer} that for $0<a<1/3$, the front $\phi^r_\mathrm{p}(a)$ is constructed by identifying the unique solution which decays to $p$ with the strongest rate for $\eps=0$. Again, as outlined above, when $\eps>0$, the equilibrium $p$ is completely stable and has a three dimensional stable manifold $\mathcal{W}^\mathrm{s}(p)$. In the pushed regime, there are three distinct real eigenvalues, and hence for $0<\eps \ll 1$, within this manifold is a one-dimensional `super' strong stable manifold $\mathcal{W}^\mathrm{sss}(p)$ which is $\mathcal{O}(\eps)$-close to the front $\phi^r_\mathrm{p}(a)$.

To construct invasion fronts in the pushed case, we aim to find an intersection of the manifolds $\mathcal{W}^\mathrm{u}(\Gamma_\eps(c))$ and $\mathcal{W}^\mathrm{s}(p)$ along the strong stable manifold $\mathcal{W}^\mathrm{sss}(p)$.

First a lemma.
\begin{Lemma}\label{lem_cmel}
Fix $0<a<1/3$ and let $\tilde{c} = c-c_\mathrm{p}(a)$. For all sufficiently small $\eps>0$ and $\tilde{c}$, the distance function defining the separation of the manifolds $\mathcal{W}^\mathrm{u}(\mathcal{M}^r_\eps)$ and $\mathcal{W}^\mathrm{sss}(p)$ is given by
\begin{align}
D(\tilde{c},\eps) = M_\mathrm{f}^c\tilde{c}+M_\mathrm{f}^\eps \eps+ \mathcal{O}(\eps^2+\tilde{c}^2),
\end{align}
where $M_\mathrm{f}^c,M_\mathrm{f}^\eps<0$ are given by
\begin{align}
M_\mathrm{f}^c&=-\int^{\infty}_{-\infty} e^{c_\mathrm{p}(a)\xi}v_\mathrm{p}(\xi)^2\,d\xi\\
M_\mathrm{f}^\eps&=\frac{1}{c_\mathrm{p}(a)}\int^{\infty}_{-\infty} \int_{-\infty}^\xi e^{c_\mathrm{p}(a)\zeta}v_\mathrm{p}(\zeta)(u_\mathrm{p}(\xi)-a)\,d\zeta,
\end{align}
where $u_\mathrm{p}(\xi),v_\mathrm{p}(\xi)$ are given by~\eqref{eq_pushedexplicit}.
\end{Lemma}
\begin{proof}
To compute this distance function, we apply Melnikov theory.

We consider the planar system~\eqref{eq_layer} for $w=0$
\begin{align}\label{eq_layer0}
\begin{split}
\dot{u}&=v\\
\dot{v} &= -cv-f(u),
\end{split}
\end{align}
As stated in~\S\ref{sec_layer}, for $c=c_\mathrm{p}(a)$, this system possesses a heteroclinic connection $\phi^r_\mathrm{p}(\xi)=(u_\mathrm{p}(\xi),v_\mathrm{p}(\xi))$ between the critical points $(u,v)=(1,0)=p_3(0)$ and $(u,v)=(a,0)=p_2(0)=p$ that lies in the intersection of $\mathcal{W}^\mathrm{u}(p_3(0))$ and $\mathcal{W}^\mathrm{sss}(p)$. We now compute the distance between $\mathcal{W}^\mathrm{u}(p_3(0))$ and $\mathcal{W}^\mathrm{sss}(p)$ to first order in $c-c_\mathrm{p}(a)$. We consider the adjoint equation of the linearization of~\eqref{eq_layer0} about the front $\phi_r$ given by
\begin{align}\label{eq_layer0adj}
\dot{\psi} = \begin{pmatrix}
0&f'(u_\mathrm{p}(\xi))\\[1em]
-1&c_\mathrm{p}(a)
\end{pmatrix}\psi. 
\end{align}
The space of solutions which grow as $\xi\to\infty$ with exponential rate at most
\begin{align}
\begin{split}
\nu\leq -\nu^+ &= \frac{c_\mathrm{p}(a)-\sqrt{c_\mathrm{p}(a)^2+4(a^2-a)}}{2} 
\end{split}
\end{align}
is one-dimensional and spanned by 
\begin{align}
\begin{split}
\psi_\mathrm{p}(\xi):= e^{c_\mathrm{p}(a)\xi}\begin{pmatrix}-\dot{v}_\mathrm{p}(\xi)\\ \dot{u}_\mathrm{p}(\xi)\end{pmatrix}\\
=e^{c_\mathrm{p}(a)\xi}\begin{pmatrix}-\dot{v}_\mathrm{p}(\xi)\\ v_\mathrm{p}(\xi)\end{pmatrix}
\end{split}
\end{align}
Let $F_0$ denote the right hand side of~\eqref{eq_layer0}, and define the Melnikov integral
\begin{align*}
M_\mathrm{f}^c&= \int^{\infty}_{-\infty} D_cF_0(\phi^r_\mathrm{p}(\xi))\cdot \psi_\mathrm{p}(\xi)\,d\xi\\
&=-\int^{\infty}_{-\infty} e^{c_\mathrm{p}(a)\xi}v_\mathrm{p}(\xi)^2\,d\xi\\
&<0.
\end{align*}
This integral measures the distance between $\mathcal{W}^\mathrm{u}(p_3(0))$ and $\mathcal{W}^\mathrm{sss}(p)$ to first order in $c-c_\mathrm{p}(a)$.  

We next compute the distance between $\mathcal{W}^\mathrm{u}(p_3(0))$ and $\mathcal{W}^\mathrm{sss}(p)$ to first order in $\epsilon$. We now consider the adjoint equation of the linearization of the full system~\eqref{eq_twode} about the front $\phi_r$ at $\epsilon=0$, which is given by
\begin{align}\label{eq_linadj}
\dot{\Psi} = \begin{pmatrix}
0&f'(u_\mathrm{p}(\xi))&0\\[1em]
-1&c_\mathrm{p}(a)&0\\
0&-1&0
\end{pmatrix}\Psi. 
\end{align}
The space of solutions which grow as $\xi\to\infty$ with exponential rate at most $-\nu^+$ is two dimensional and spanned by 
\begin{align}
\begin{split}
\Psi^1_\mathrm{p}(\xi):= \begin{pmatrix}-e^{c_\mathrm{p}(a)\xi}\dot{v}_\mathrm{p}(\xi)\\ e^{c_\mathrm{p}(a)\xi}\dot{u}_\mathrm{p}(\xi)\\ -\int_{-\infty}^\xi e^{c_\mathrm{p}(a)\zeta}\dot{u}_\mathrm{p}(\zeta)\end{pmatrix}, \qquad \Psi^2_\mathrm{p}(\xi):= \begin{pmatrix}0\\0\\1\end{pmatrix},
\end{split}
\end{align}
and $\Psi_\mathrm{p}^1$ is the unique such solution which converges to zero as $t\to -\infty$. We denote by $F_1$ denote the right hand side of~\eqref{eq_twode}, and we define the Melnikov integral
\begin{align*}
M_\mathrm{f}^\eps&= \int^{\infty}_{-\infty} D_\eps F_1(\phi^r_\mathrm{p}(\xi))\cdot \Psi^1_\mathrm{p}(\xi)\,d\xi\\
&=\frac{1}{c_\mathrm{p}(a)}\int^{\infty}_{-\infty} \int_{-\infty}^\xi e^{c_\mathrm{p}(a)\zeta}v_\mathrm{p}(\zeta)(u_\mathrm{p}(\xi)-a)\,d\zeta \,d\xi\\
&<0,
\end{align*}
which measures the distance between $\mathcal{W}^\mathrm{u}(p_3(0))$ and $\mathcal{W}^\mathrm{sss}(p)$ to first order in $\epsilon$.  

Setting $\tilde{c} = c-c_\mathrm{p}(a)$, we are now able to write the distance function
\begin{align}
D(\tilde{c},\eps) = M_\mathrm{f}^c\tilde{c}+M_\mathrm{f}^\eps \eps+ \mathcal{O}(\eps^2+\tilde{c}^2),
\end{align}
which defines the separation of $\mathcal{W}^\mathrm{sss}(p)$ and $\mathcal{W}^\mathrm{u}(\mathcal{M}^r_\eps)$ for $0<\eps\ll1$ and $\tilde{c}$ sufficiently small.
\end{proof}

We can now complete the proof of Theorem~\ref{thm_pushedexistence}.
\begin{proof}[Proof of Theorem~\ref{thm_pushedexistence}]
The manifolds $\mathcal{W}^\mathrm{sss}(p)$ and $\mathcal{W}^\mathrm{u}(\mathcal{M}^r_0)$ intersect for $\eps=0$. These manifolds perturb for $0<\eps\ll 1$. By Proposition~\ref{prop_poinvmflds}, the unstable manifold $\mathcal{W}^\mathrm{u}(\Gamma(c))$ is $\mathcal{O}(e^{-\eta/\eps})$-close to $\mathcal{W}^\mathrm{u}(\mathcal{M}^r_\eps)$. We show that by adjusting $c$, it is possible to find an intersection of $\mathcal{W}^\mathrm{u}(\Gamma(c))$ and $\mathcal{W}^\mathrm{sss}(p)$.

In Lemma~\ref{lem_cmel}, the distance between $\mathcal{W}^\mathrm{sss}(p)$ and $\mathcal{W}^\mathrm{u}(\mathcal{M}^r_\eps)$ was computed as
\begin{align}
D(\tilde{c},\eps) = M_\mathrm{f}^c\tilde{c}+M_\mathrm{f}^\eps \eps+ \mathcal{O}(\eps^2+\tilde{c}^2).
\end{align}
Since $\mathcal{W}^\mathrm{u}(\Gamma(c_\mathrm{p}(a,\eps)))$ is $\mathcal{O}(e^{-\eta/\eps})$-close to $\mathcal{W}^\mathrm{u}(\mathcal{M}^r_\eps)$, the separation between $\mathcal{W}^\mathrm{u}(\Gamma(c))$ and $\mathcal{W}^\mathrm{sss}(p)$ is given by the modified distance function
\begin{align}
\hat{D}(\tilde{c},\eps) = M_\mathrm{f}^c\tilde{c}+M_\mathrm{f}^\eps \eps+ \mathcal{O}(\eps^2+\tilde{c}^2),
\end{align}
where the exponentially small terms have been absorbed into the $\mathcal{O}(\eps^2)$ term. We can solve for an intersection by setting $\hat{D}(\tilde{c},\eps)=0$, which occurs when
\begin{align}\label{eq_cpexpansion}
c =  c_\mathrm{p}(a,\eps):=c_\mathrm{p}(a)-\frac{M_\mathrm{f}^\eps}{M_\mathrm{f}^c}\epsilon+\mathcal{O}(\eps^2),
\end{align}
where the coefficient $\frac{M_\mathrm{f}^\eps}{M_\mathrm{f}^c}>0$. This intersection occurs along the desired solution orbit $P^r_\eps(a)$.

\end{proof}

\section{Direct simulations --- corroboration and more phenomena}\label{sec_numerics} 

In this section, we present results of direct numerical simulations which visualize the predicted pulled and pushed fronts of Theorems~\ref{thm_pulledexistence} and~\ref{thm_pushedexistence}, as well as some additional interesting phenomena in the region $a=\mathcal{O}(\eps)$. All simulations were performed for $\eps=0.001$ and $\gamma=0$, using a Runge-Kutta scheme initialized with the homogeneous state $(u,w)=(a,0)$ and a small amplitude random perturbation near the left edge of the spatial domain.

\paragraph{Pulled fronts.}
Our results in Theorem~\ref{thm_pulledexistence} predict the existence of `left' and `right' pulled fronts $F^\ell_\eps(a), F^r_\eps(a)$ for each value of $a\in ( 1/2-1/\sqrt{6}, 1/2)$ with wave speed $c=c_\mathrm{lin}(a,\eps)=c_\mathrm{lin}(a)+\mathcal{O}(\eps)$. While Theorem~\ref{thm_pulledexistence} guarantees the existence of these solutions for all such values of $a$, we expect that the right pulled front $F^r_\eps(a)$ loses stability for $a<1/3$ as Theorem~\ref{thm_pushedexistence} predicts the existence of a steeper pushed front. Furthermore, in this region the singular front $\phi_\mathrm{lin}^r = (u_\mathrm{lin}^r,v_\mathrm{lin}^r)$ in the underlying fast system~\eqref{eq_layer} is non-monotone and therefore by Sturm-Liouville theory the linearization
\begin{align}
\mathcal{L} U := U_{\xi\xi}+c_\mathrm{lin}U_\xi +f'(u_\mathrm{lin}^r)U
\end{align}
of the layer Nagumo PDE
\begin{align}
u_t = u_{xx} +f(u)
\end{align}
in the corresponding comoving frame $\xi=x-c_\mathrm{lin}t$ admits a positive eigenvalue. We expect that this eigenvalue persists in the full system for $\eps>0$, rendering the front $F^r_\eps(a)$ unstable in the PDE~\eqref{eq_pde} for $a<1/3$.

 Figure~\ref{fig:pulledfronts} depicts spacetime plots of the variable $u(x,t)$ of left and right pulled fronts obtained for $a=0.4$ as well as a left pulled front obtained for $a=0.2$. The left pulled fronts are distinguished via the fact that profile eventually increases monotonically as $\xi = x-ct \to \infty$. 

\begin{figure}
\begin{subfigure}{.33 \textwidth}
\centering
\includegraphics[width=1\linewidth]{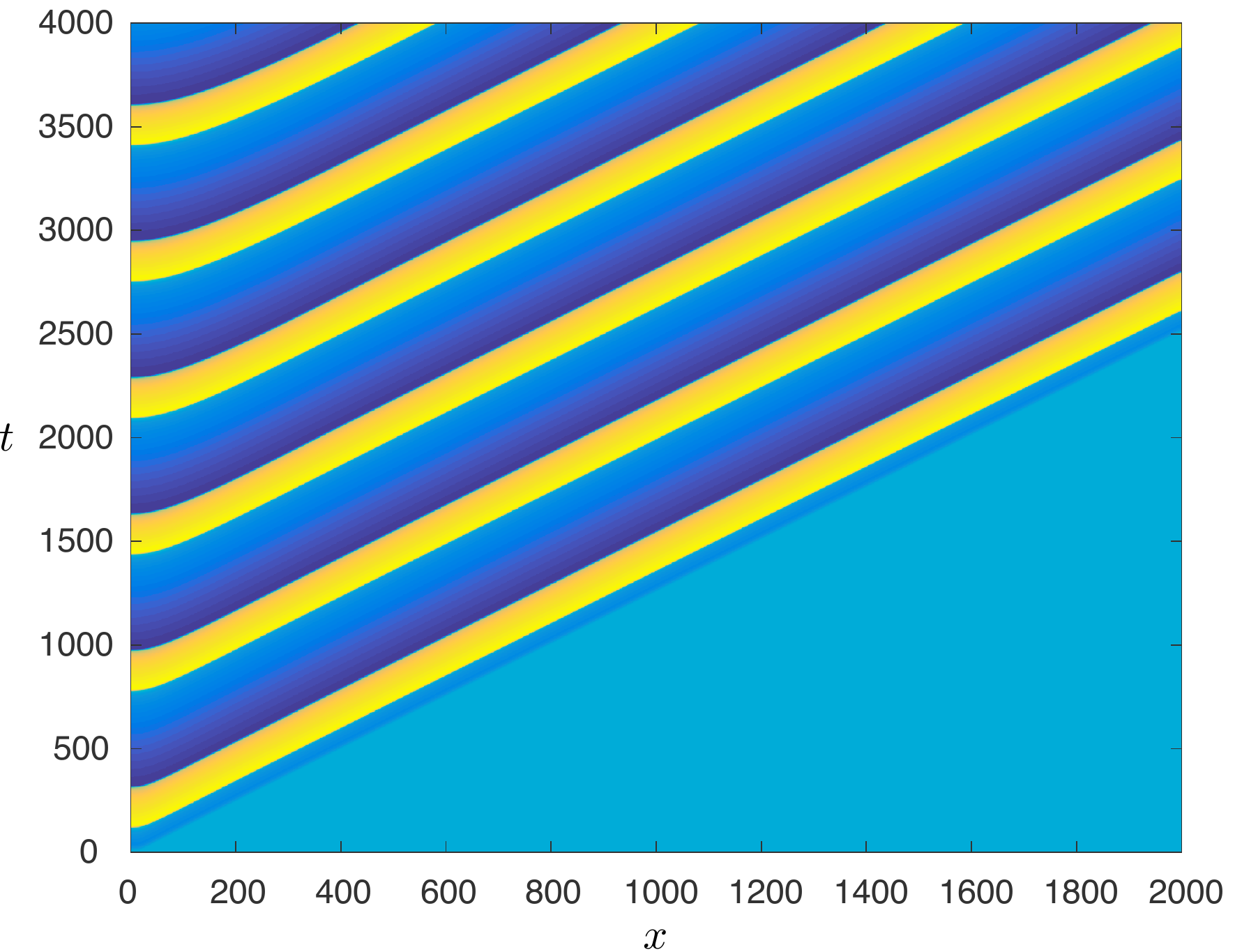}
\end{subfigure}
\begin{subfigure}{.33 \textwidth}
\centering
\includegraphics[width=1\linewidth]{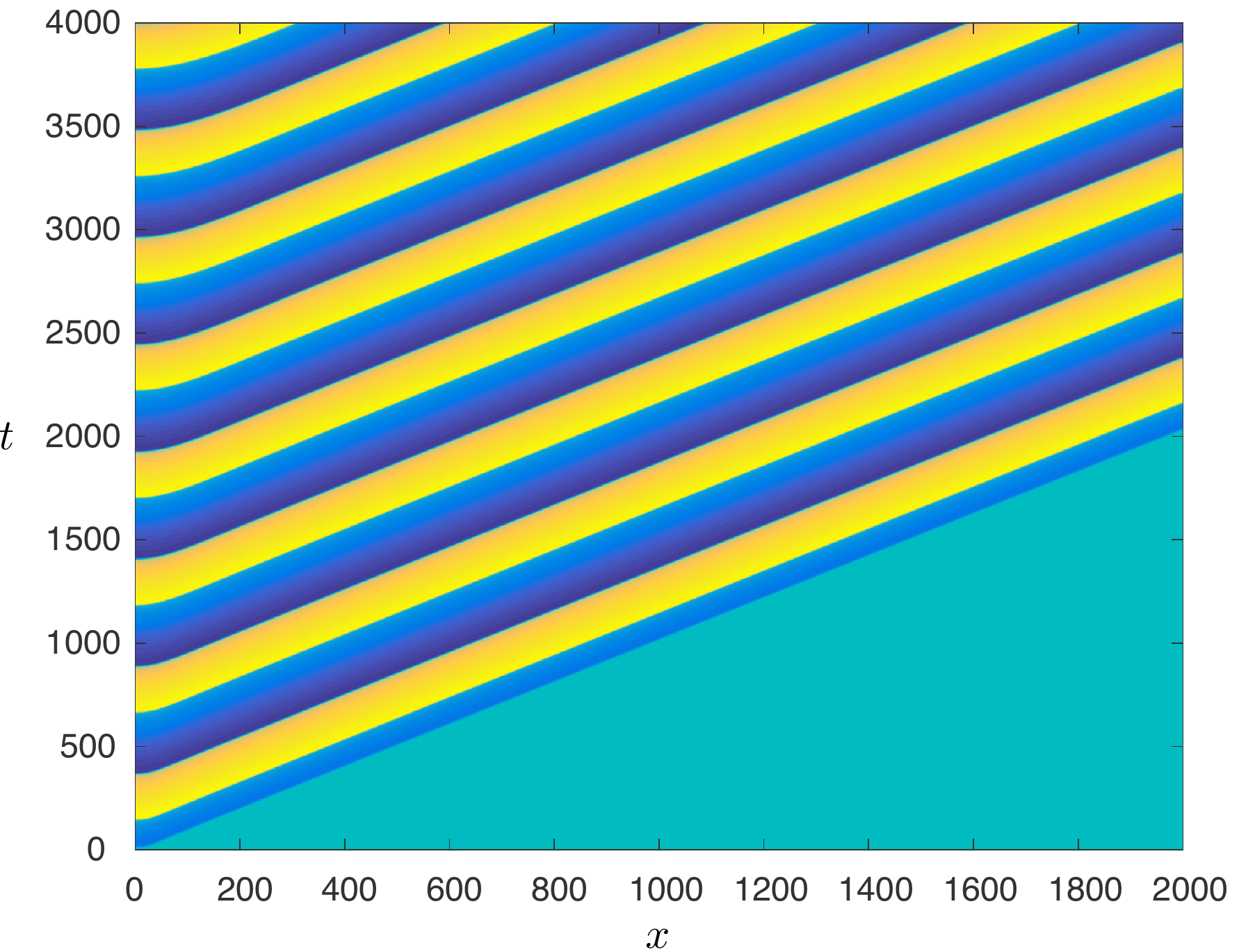}
\end{subfigure}
\begin{subfigure}{.33 \textwidth}
\centering
\includegraphics[width=1\linewidth]{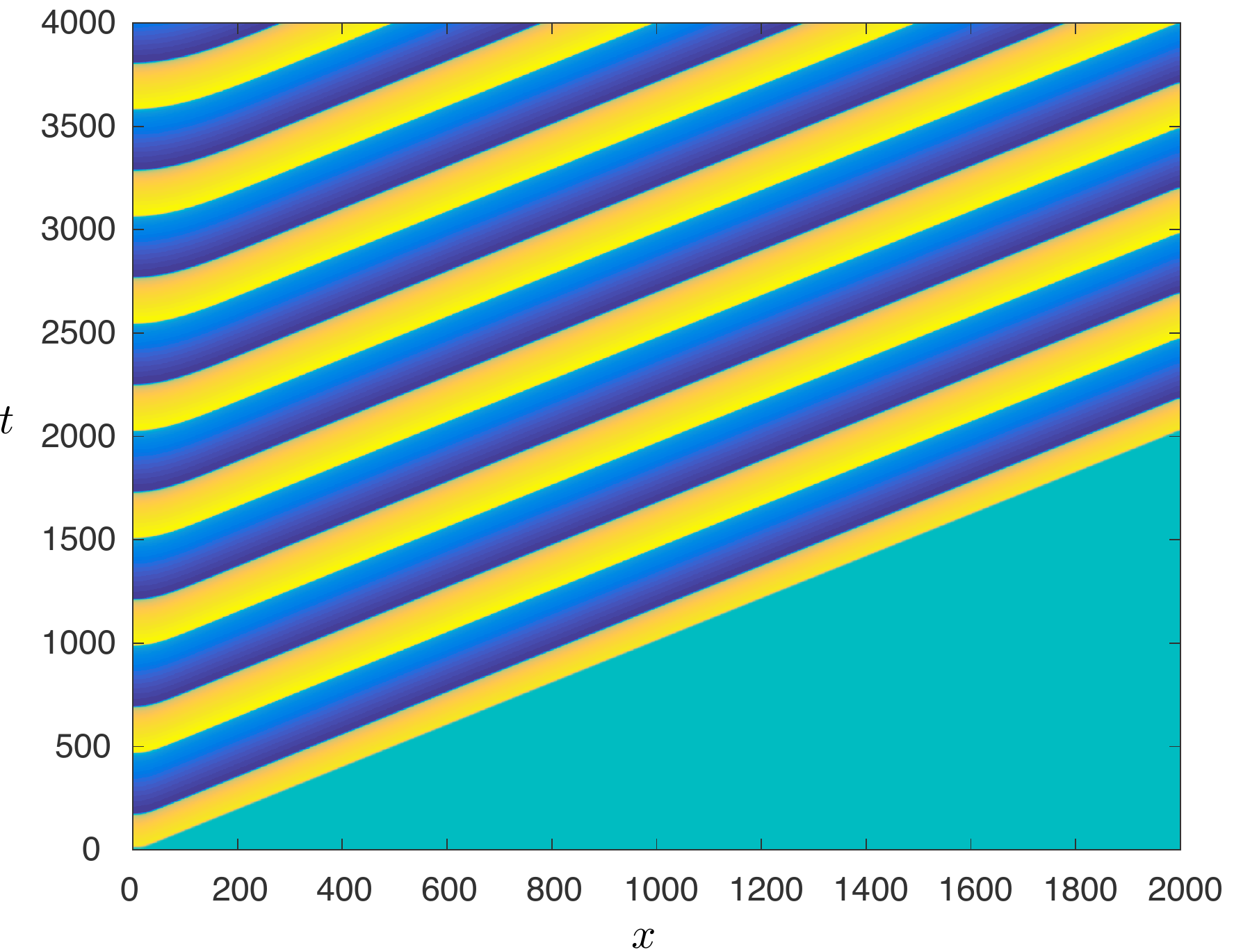}
\end{subfigure}
\caption{Shown are spacetime plots of $u(x,t)$ corresponding to left pulled fronts obtained for $a=0.2$ (left panel) and $a=0.4$ (center panel) and a right pulled front obtained for $a=0.4$ (right panel).}
\label{fig:pulledfronts}
\end{figure}

\paragraph{Pushed fronts.}  We illustrate the results of Theorem~\ref{thm_pushedexistence}, which predicts pushed fronts for each $0<a<1/3$. Figure~\ref{fig:pushedfronts} depicts spacetime plots of the variable $u(x,t)$ of pushed fronts obtained for values of $a=\{0.05, 0.1, 0.2\}$.

\begin{figure}
\begin{subfigure}{.33 \textwidth}
\centering
\includegraphics[width=1\linewidth]{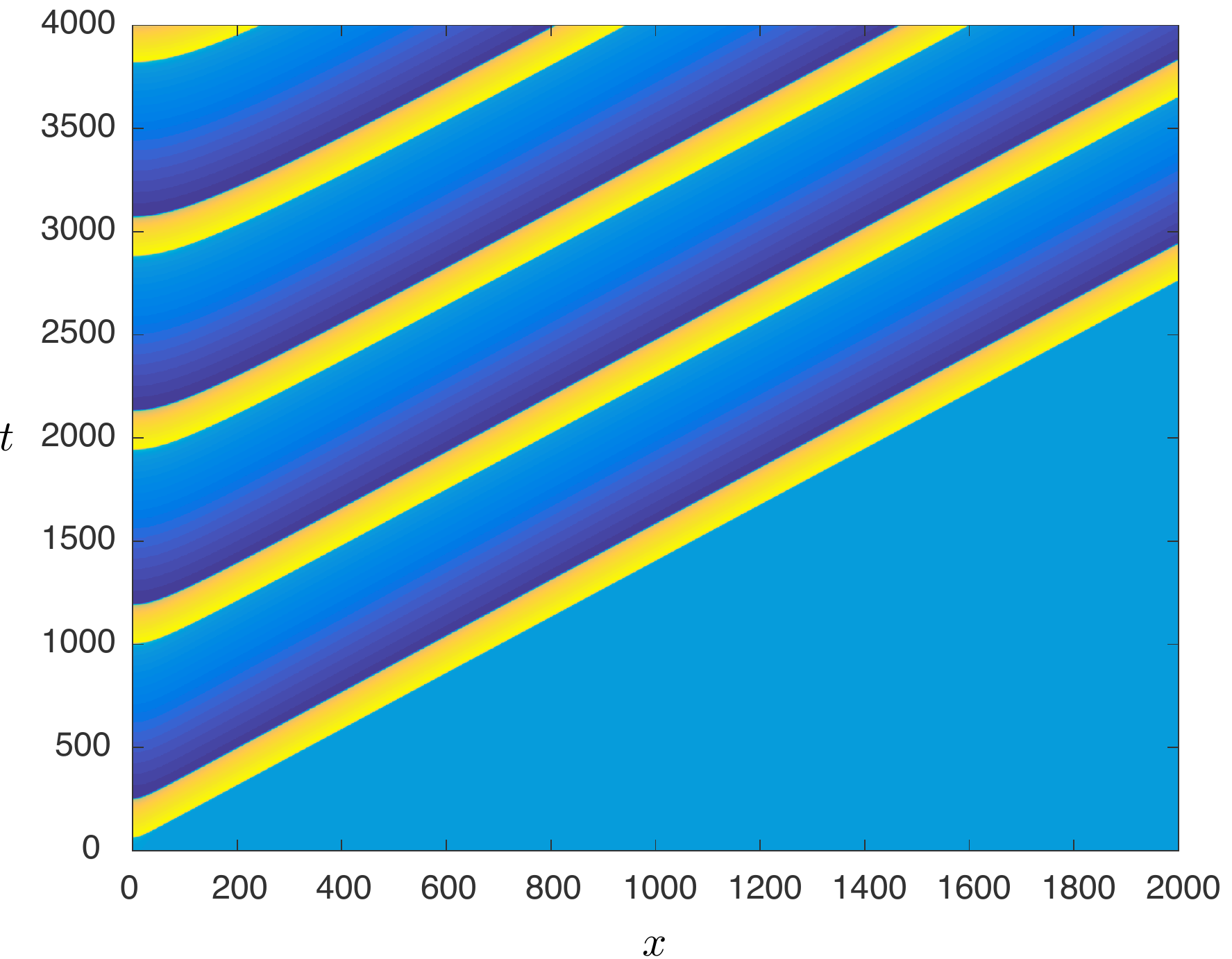}
\end{subfigure}
\begin{subfigure}{.33 \textwidth}
\centering
\includegraphics[width=1\linewidth]{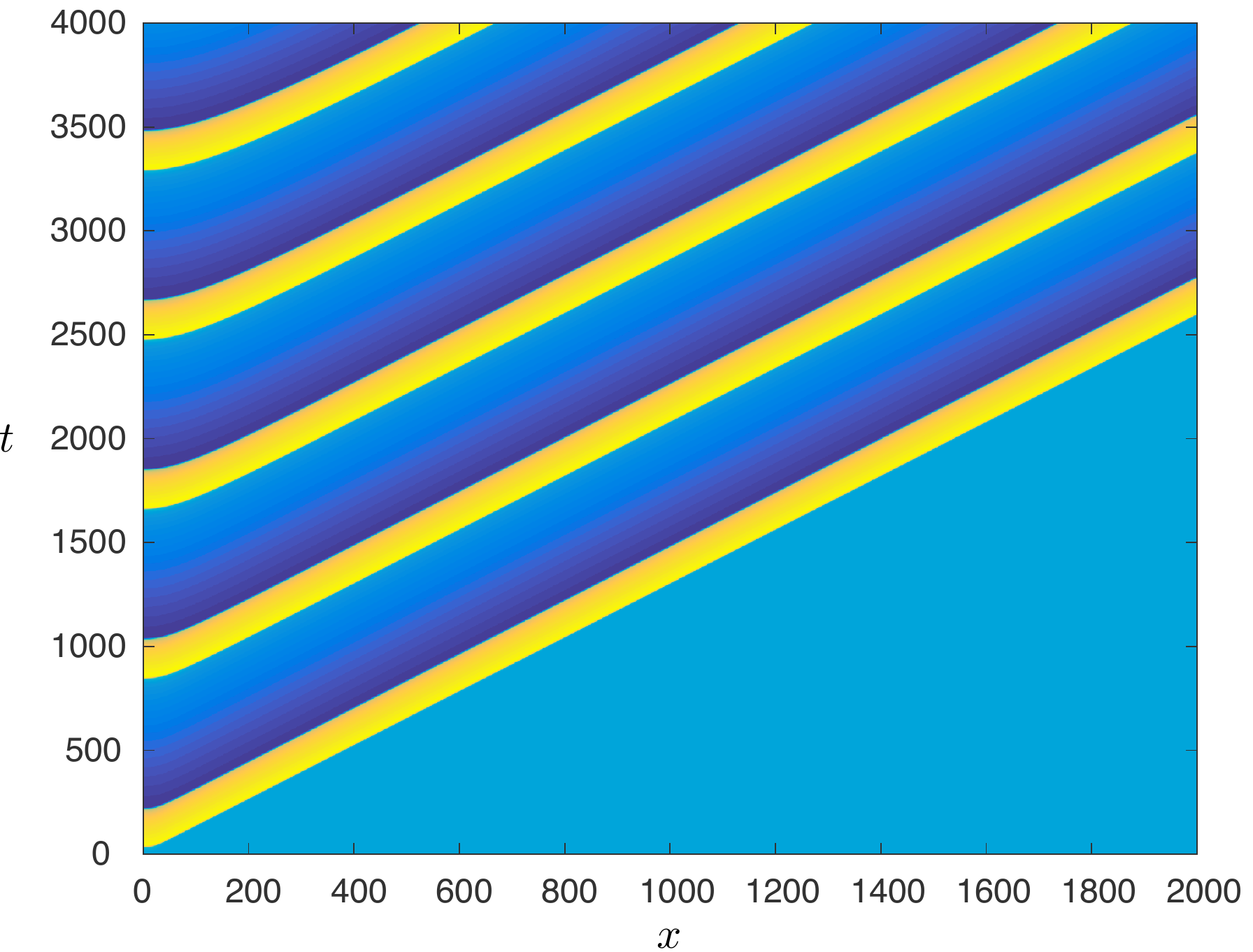}
\end{subfigure}
\begin{subfigure}{.33 \textwidth}
\centering
\includegraphics[width=1\linewidth]{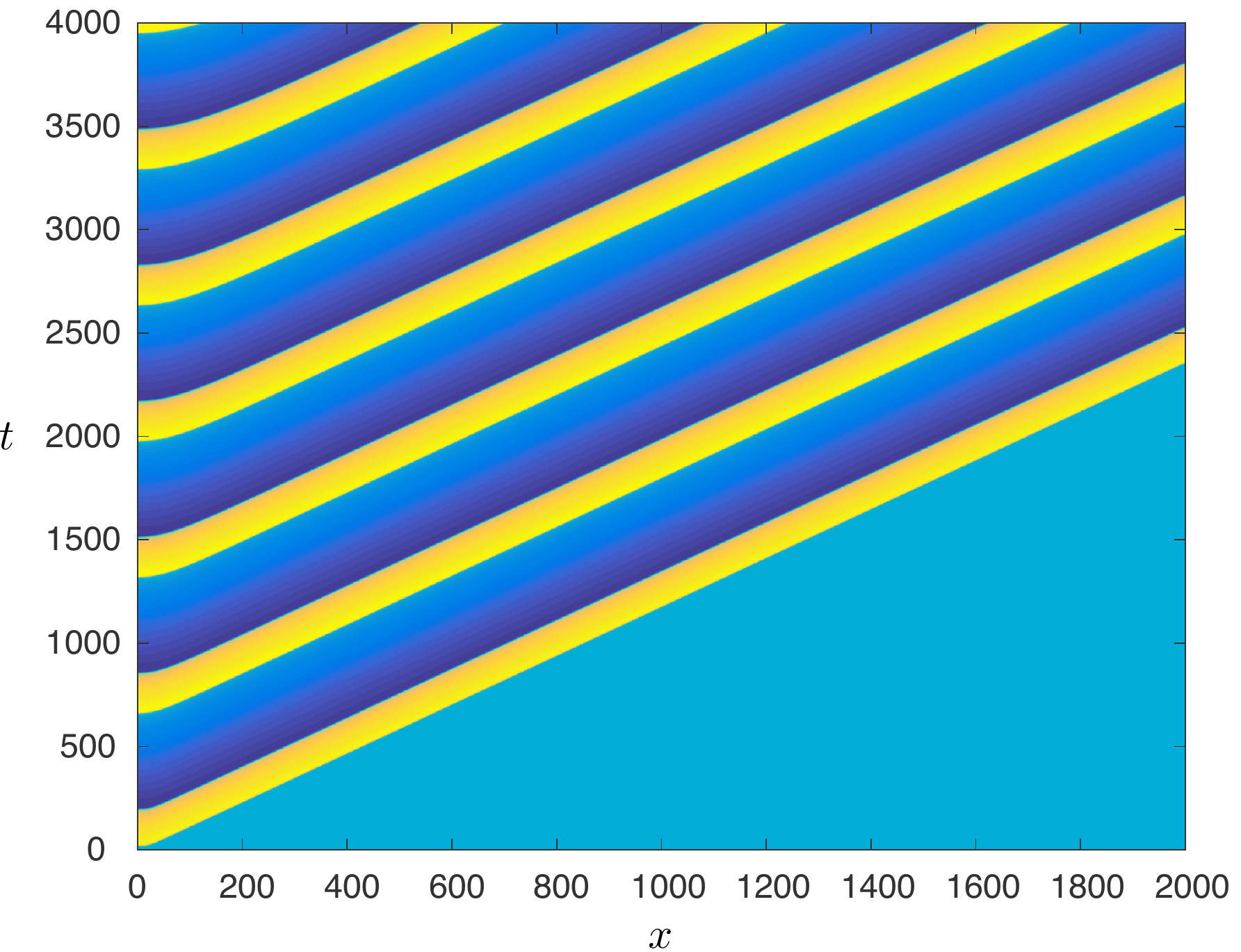}
\end{subfigure}
\caption{Shown are spacetime plots of $u(x,t)$ corresponding to pushed fronts obtained for $a=0.05$ (left panel),  $a=0.1$ (center panel), and $a=0.2$ (right panel).}
\label{fig:pushedfronts}
\end{figure}

\paragraph{Speed predictions.}
We compare our predictions with speeds obtained numerically in direct simulations. For pulled fronts, we compute the predicted speed $c_\mathrm{lin}(a,\eps)$ by solving for pinched double roots using a Newton continuation solver. To compute the predicted speed $c_\mathrm{p}(a,\eps)$ for pushed fronts, we use the explicit solutions~\eqref{eq_pushedexplicit} and the expression~\eqref{eq_cpexpansion} for $c_\mathrm{p}(a,\eps)$ in terms of the Melnikov integrals $M_\mathrm{f}^\eps, M_\mathrm{f}^c$, to obtain the leading order approximation
\begin{align}
c_\mathrm{p}(a,\eps) \approx \frac{1+a}{\sqrt{2}}-\frac{3\sqrt{2}}{a(1+a)}\left(-1+H\left[\frac{1+a}{1-a}\right]\right) \epsilon \end{align}
where the Melnikov integrals were evaluated in Mathematica, and $H[x]$ denotes the harmonic number function. The approximations for the predicted speeds $c_\mathrm{lin}(a,\eps)$ and $c_\mathrm{p}(a,\eps)$ are depicted in Figure~\ref{fig:speedpredictions} by magenta and dashed cyan curves, respectively.

\begin{figure}
\centering
\includegraphics[width=0.6\linewidth]{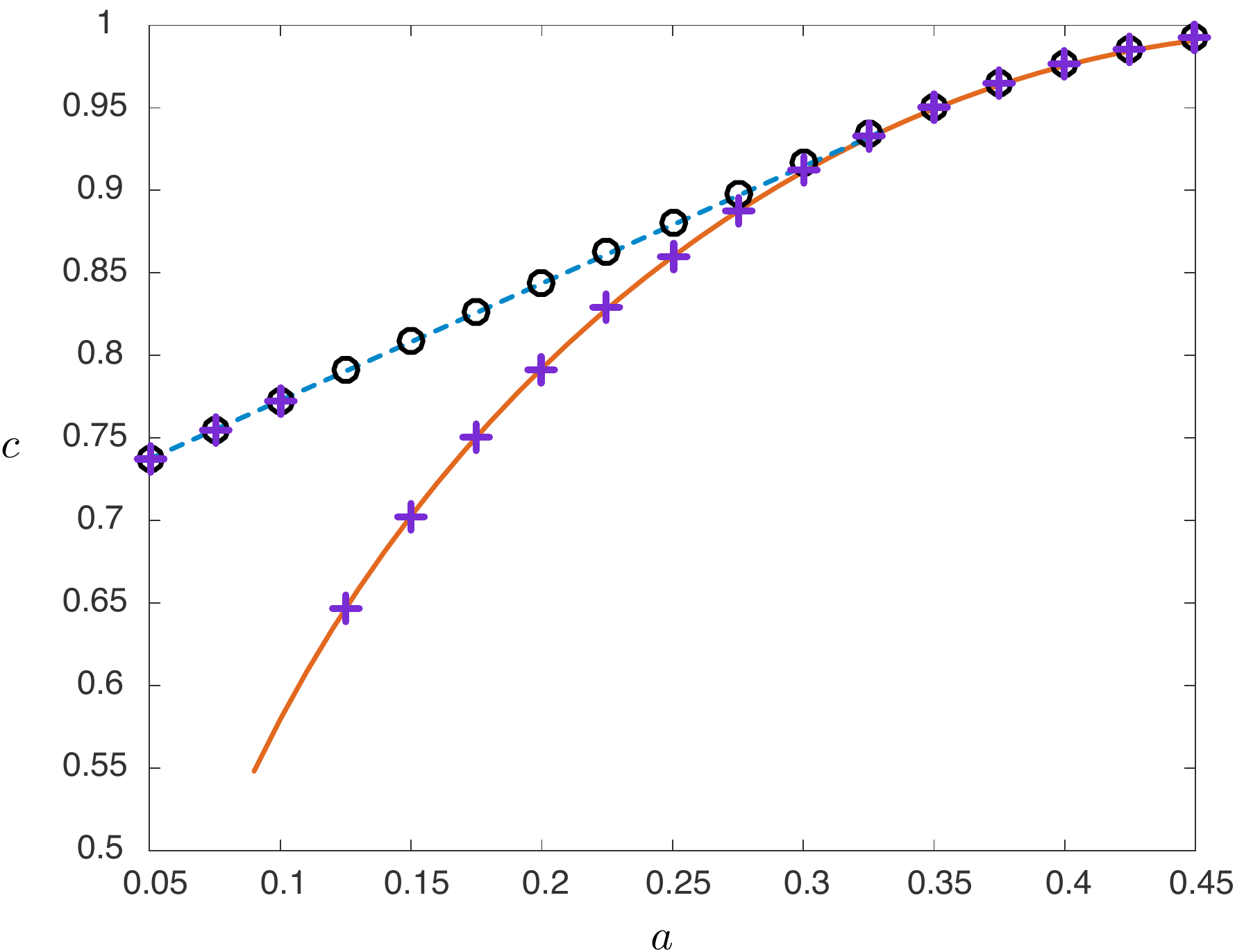}
\caption{Shown are speeds obtained from direct numerical simulations compared with predictions for $\eps=0.001$. The dashed cyan curve and solid magenta curve depict predictions for $c_\mathrm{p}(a,\eps)$ and $c_\mathrm{lin}(a,\eps)$, respectively. The black circles and purple crosses depict numerically computed speeds for profiles obtained from direct numerical simulations with positive and negative initial perturbations, respectively, in order to select for right versus left pulled fronts. Note the crossover from right pulled fronts to pushed fronts as $a$ decreases through $a=1/3$, as well as the disappearance of left pulled fronts as $a$ decreases through $a_\mathrm{b}$.}
\label{fig:speedpredictions}
\end{figure}

Also shown in Figure~\ref{fig:speedpredictions} are numerically computed speeds of profiles obtained in direct numerical simulations, which are in agreement with the predictions. In order to select for left (resp. right) pulled fronts, the initial perturbation from the homogeneous state $(u,w)=(a,0)$ was chosen to be negative (resp. positive). The results for negative perturbations are depicted by purple crosses. Left pulled fronts are obtained for values of $a$ greater than the critical $a_\mathrm{b}$, below which connecting orbits are blocked; for lower values of $a$, the solutions transition to right pushed fronts with the corresponding speed $c_\mathrm{p}(a,\eps)$. The results for positive perturbations are depicted by black circles. Right pulled fronts are obtained for values of $a>1/3$, below which the transition to pushed fronts occurs. This agrees with our expectation that the right pulled fronts lose stability for values of $a<1/3$.

\paragraph{Additional bifurcations for $a=\mathcal{O}(\eps)$.} The results of Theorem~\ref{thm_pushedexistence} are valid for fixed $a$ and sufficiently small $\eps>0$, that is, typically $0<\eps \ll a$. As $a\to 0$, the equilibrium $(u,w) = (a,0)$ moves toward the origin and approaches the lower left fold point on the critical manifold; see Figure~\ref{fig:singular_slow}. When $a=0$, the equilibrium lies exactly at the fold and admits the structure of a canard point~\cite{krupaszmolyan2001}. In~\eqref{eq_pde}, this specific scenario is responsible for a wide range of complex canard-induced dynamics~\cite{cas,carterunpeeling} when unfolding jointly in the parameters $0<a,\eps\ll 1$. 

In our current setup, we therefore expect similar phenomena in this region, in particular traveling canard orbits in~\eqref{eq_twode} as well as spatially homogeneous canard oscillations in the kinetics of~\eqref{eq_pde}. Figure~\ref{fig:oscillatoryfronts} depicts the results of direct numerical simulations for $a=\eps=0.05$ and $a=\eps=0.1$; we observe the appearance of traveling front solutions which leave oscillatory patterns in the wake of the interface. 

\begin{figure}
\hspace{.025\textwidth}
\begin{subfigure}{.45 \textwidth}
\centering
\includegraphics[width=1\linewidth]{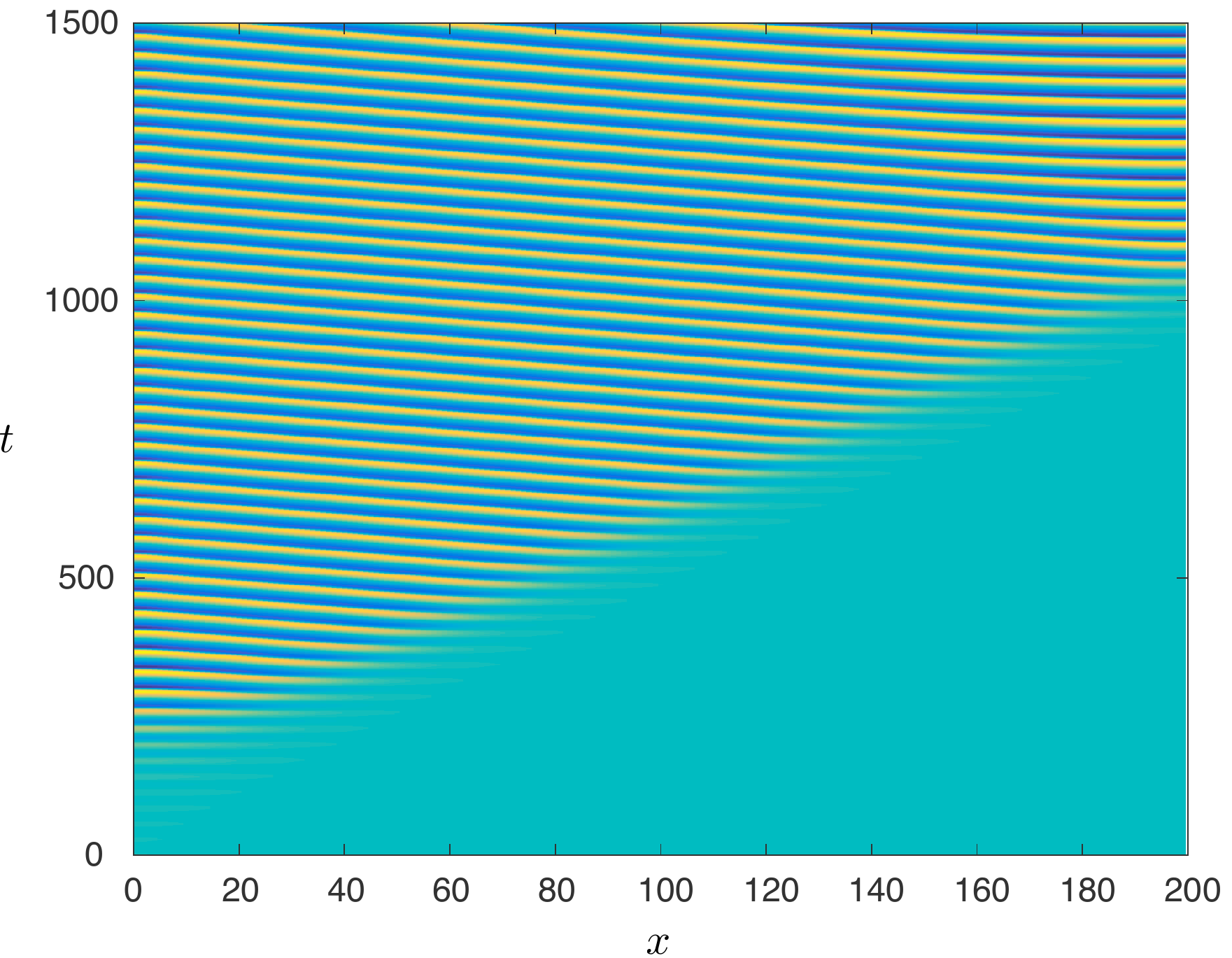}
\end{subfigure}
\hspace{.025\textwidth}
\begin{subfigure}{.45 \textwidth}
\centering
\includegraphics[width=1\linewidth]{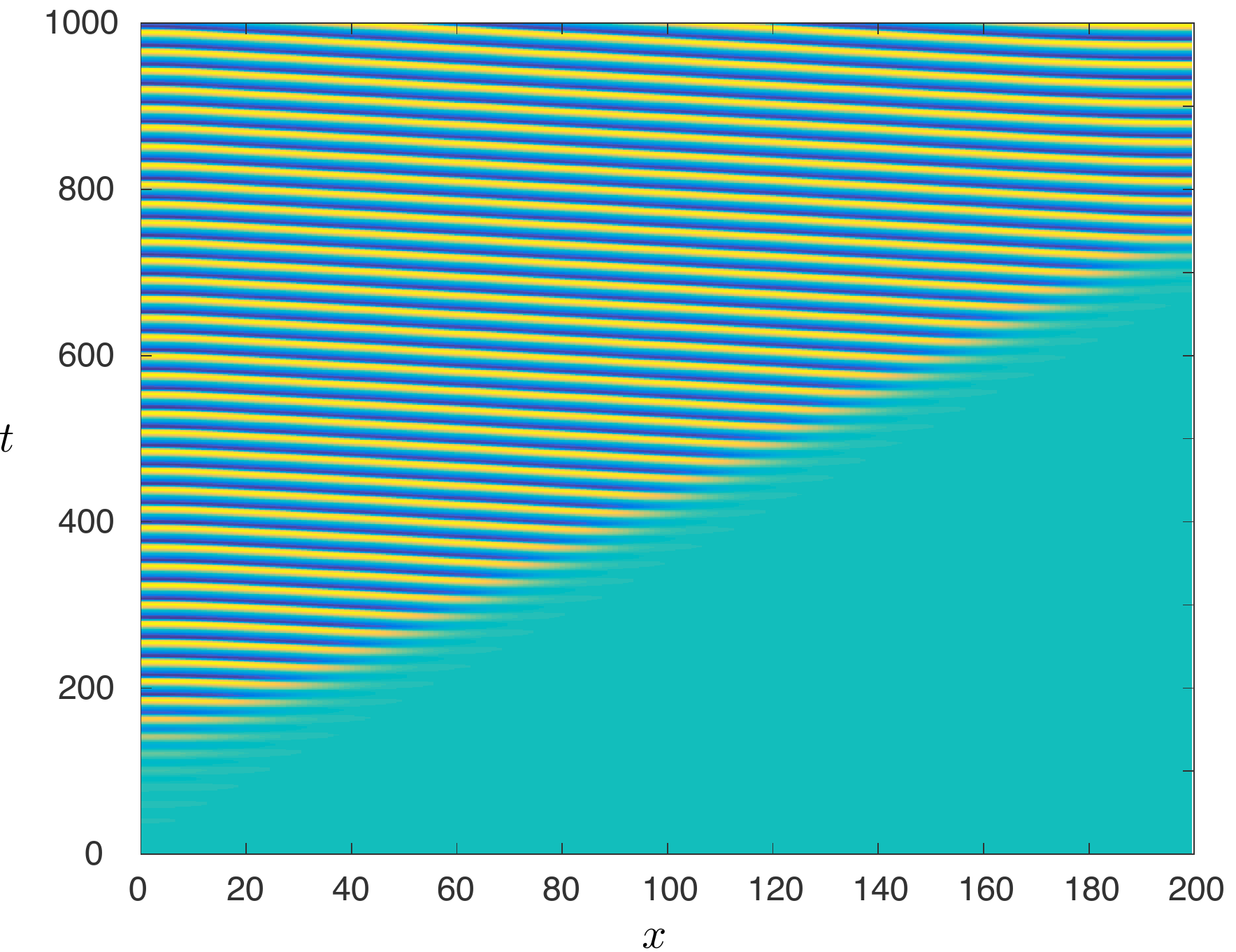}
\end{subfigure}
\hspace{.025\textwidth}
\caption{Shown are spacetime plots of oscillatory fronts for $a=\eps = 0.05$ (left panel) and $a=\eps = 0.1$ (right panel).}
\label{fig:oscillatoryfronts}
\end{figure}

\section{Discussion}\label{sec_discussion}
We presented existence results for periodic wave trains and heteroclinic orbits connecting an equilibrium to those wave trains in a traveling-wave equation for the FitzHugh--Nagumo equation in the oscillatory regime, that is, in the parameter regime when the unique equilibrium of the system is unstable. The heteroclinic orbits represent fronts that describe spreading of oscillations into a region occupied by an unstable equilibrium. We focused on specific heteroclinic orbits with wave speed parameter given by a steepest front selection criterion, which are observed when starting from compactly supported initial perturbations of the unstable state. 

A first natural extension  of our result would be concerned with the stability of the invasion fronts. As mentioned, we expect right-pulled fronts to be stable only for values of $a>1/3$, and left-pulled fronts to be stable throughout. Pushed fronts should be stable, mimicking altogether the results from the scalar Nagumo equation. Here, stability would refer to first spectral stability in suitably exponentially weighted spaces, linear stability, or even nonlinear stability against localized perturbations. A nonlinear stability analysis involves several difficulties, involving first the leading edge, and second the wave trains in the wake of the front. In the leading edge, the analysis for pushed fronts is rather straighforward as exponential weights enforcing decay slightly weaker than the front push the essential spectrum into the left half plane, leaving only a simple eigenvalue at the origin, while allowing for perturbations that correspond to compactly supported initial data \cite{hadeler}. For pulled fronts, exponential weights with rate of decay of the front push the essential spectrum only to the origin, and additional algebraic weights are required to obtain linear and nonlinear decay; see for instance \cite{gallay}. Initial conditions with (one-sided) compact support relative to the unstable equilibrium are not small (or even bounded) perturbations in such a function space. The second difficulty is concerned with the fact that the wave trains in the wake of the primary front are only diffusively stable, thus requiring a decomposition analogous to \cite{gsu}. An analysis allowing for compactly supported perturbations of the unstable state or even perturbations of finite size, independent of $\eps$, appears out of reach at this point. 

In a slightly different direction, one would like to understand transitions to oscillatory fronts as shown in Figure \ref{fig:oscillatoryfronts}. Similar phenomena had been observed in \cite{mesuro} and attributed to an instability of a pushed front with respect to a pair of pinched double roots crossing the imaginary axis \cite[\S 3 \& \S 5]{mesuro}. In the present case, the analysis would likely involve the understanding of the critical canard transition at $a=0$, extending work in \cite{carterunpeeling}.

\end{document}